\documentclass[pra,twocolumn,superscriptaddress,nofootinbib]{revtex4-2}
\usepackage[utf8]{inputenc}
\usepackage{graphicx}
\usepackage{amsmath}
\usepackage{amsthm}
\usepackage{bm}
\usepackage{layout}
\usepackage{float}
\usepackage{amsfonts}
\usepackage{amssymb}%
\usepackage{array}%
\usepackage[margin=0.75in]{geometry}
\usepackage{color}
\usepackage{soul}
\usepackage{mathtools}
\usepackage[colorlinks=true,citecolor=blue]{hyperref}
\theoremstyle{definition}
\newtheorem{theorem}{Theorem}
\newtheorem{claim}{Claim}

\newtheorem{lemma}{Lemma}
\newtheorem{definition}{Definition}

\usepackage[capitalise,compress]{cleveref}

\usepackage{physics}

\usepackage{ulem}

\renewcommand{\emph}[1]{\textit{#1}}

\newcounter{para}

\makeatletter
\newcommand*\bigcdot{\mathpalette\bigcdot@{.5}}
\newcommand*\bigcdot@[2]{\mathbin{\vcenter{\hbox{\scalebox{#2}{$\m@th#1\bullet$}}}}}

\newcommand{\llangle}[1][]{\savebox{\@brx}{\(\m@th{#1\langle}\)}%
  \mathopen{\copy\@brx\kern-0.5\wd\@brx\usebox{\@brx}}}
\newcommand{\rrangle}[1][]{\savebox{\@brx}{\(\m@th{#1\rangle}\)}%
  \mathclose{\copy\@brx\kern-0.5\wd\@brx\usebox{\@brx}}}

\makeatother

\usepackage{array}
\newcolumntype{L}{>{$}l<{$}} 
\newcolumntype{C}{>{$}c<{$}} 
\newcolumntype{R}{>{$}r<{$}} 

\usepackage{xcolor}

\usepackage{mathtools}

\usepackage{physics}
\usepackage{bbm}
\usepackage{booktabs}

\newtheorem*{theorem*}{Theorem}
\newtheorem*{lemma*}{Lemma}

\usepackage{comment}
\usepackage{ragged2e}
\usepackage{atbegshi,picture}
\begin{document}
\newcommand{\Caltech}{California Institute of Technology, Pasadena, CA 91125, USA}
\newcommand{\MIT}{Center for Theoretical Physics, Massachusetts Institute of Technology, Cambridge, MA 02139, USA}
\newcommand{\Stanford}{Department of Electrical Engineering, Stanford University, Stanford, CA, USA}
\newcommand{\MEP}{maximum entropy principle}

\title{
A Maximum Entropy Principle in Deep Thermalization\\ and in Hilbert-Space Ergodicity
}
\date{\today}
\author{Daniel K. Mark}
\affiliation{\MIT}
\author{Federica Surace}
\affiliation{\Caltech}
\author{Andreas Elben}
\affiliation{\Caltech}
\author{Adam L. Shaw}
\affiliation{\Caltech}
\author{Joonhee Choi}
\affiliation{\Stanford}
\author{Gil Refael}
\affiliation{\Caltech}
\author{Manuel Endres}
\affiliation{\Caltech}
\author{Soonwon Choi}
\email{soonwon@mit.edu}
\affiliation{\MIT}

\AtBeginShipoutNext{\AtBeginShipoutUpperLeft{%
  \put(\dimexpr\paperwidth-1cm\relax,-1.5cm){\makebox[0pt][r]{MIT-CTP 5692}}%
}}

\begin{abstract}
We report universal statistical properties displayed by ensembles of pure states that naturally emerge in quantum many-body systems.
Specifically, two classes of state ensembles are considered: those formed by i) the temporal trajectory of a quantum state under unitary evolution or ii) the quantum states of small subsystems obtained by partial, local projective measurements performed on their complements. These cases respectively exemplify the phenomena of ``Hilbert-space ergodicity" and   ``deep thermalization." 
In both cases, the resultant ensembles are defined by a simple principle: the distributions of pure states have maximum entropy, subject to constraints such as energy conservation, and effective constraints imposed by thermalization.
We present and numerically verify quantifiable signatures of this principle by deriving explicit formulae for all statistical moments of the ensembles; proving the necessary and sufficient conditions for such universality under widely-accepted assumptions; and describing their measurable consequences in experiments. We further discuss information-theoretic implications of the universality: our ensembles have maximal information content while being maximally difficult to interrogate, establishing that generic quantum state ensembles that occur in nature hide (scramble) information as strongly as possible.
Our results generalize the notions of
Hilbert-space ergodicity to time-independent Hamiltonian dynamics and deep thermalization from infinite to finite effective temperature.
Our work presents new perspectives to characterize and understand universal behaviors of quantum dynamics using statistical and information theoretic tools.
\end{abstract}

\maketitle

\section{Introduction}
\begin{figure*}[t!]    \includegraphics[width=0.75\textwidth]{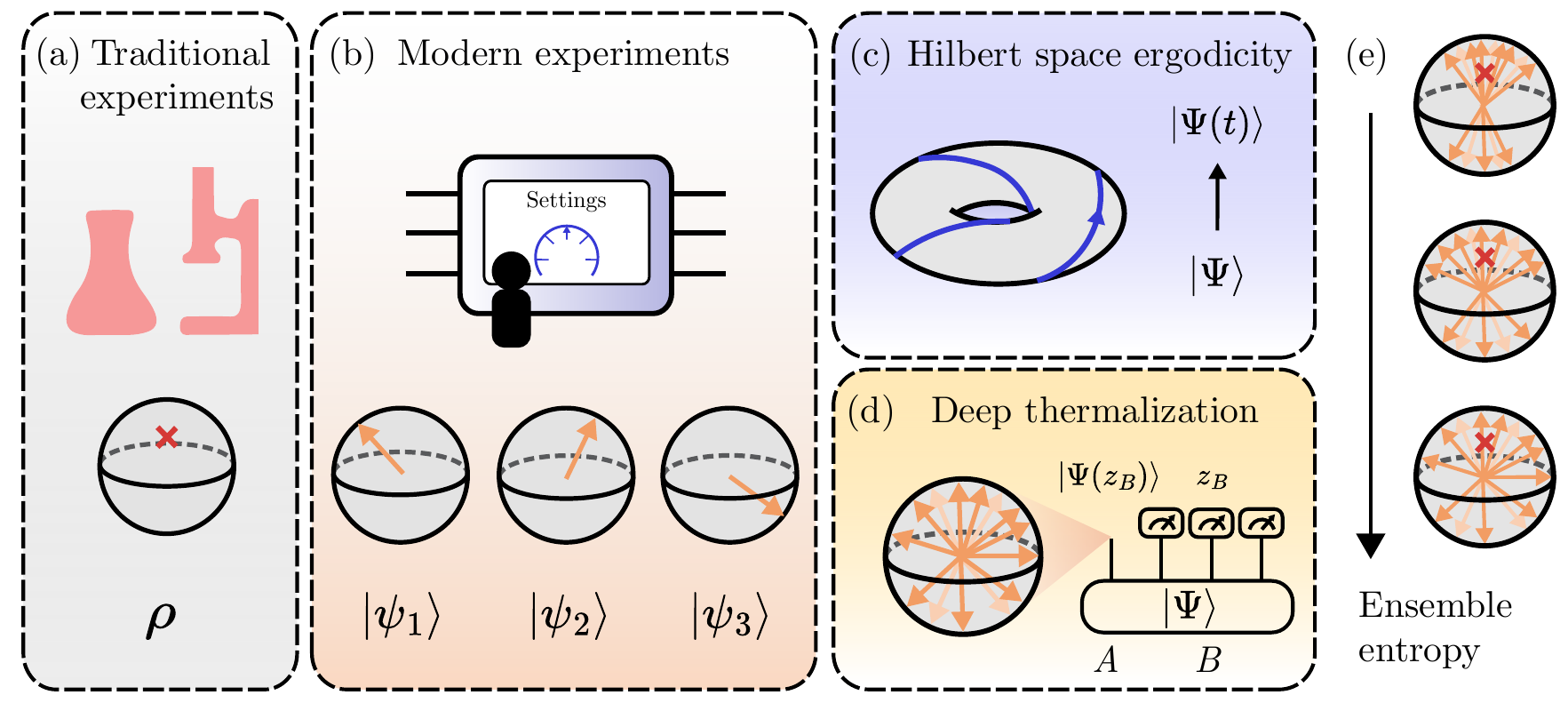}
    \caption{Overview of this work.
    In this work, we study ensembles of states generated by natural dynamics. (a) Traditional experiments on quantum systems can only measure the average properties of such ensembles, described by a density matrix $\rho$ (red cross in Bloch sphere). (b) In contrast, modern experimental capabilities are able to generate and investigate ensembles of labeled quantum states, and study their higher statistical properties. We specifically study two ensembles of states in this work: (c) the \textit{temporal ensemble} of states generated by taking an initial state $\ket{\Psi_0}$ and evolving it to different times under a (fixed) Hamiltonian $H$, and (d) the \textit{projected ensemble} of states $\{\ket{\Psi(z_B)}\}$ generated by projective measurement (performed on a large subsystem) of a single state $\ket{\Psi}$ resulting in many different outcomes $z_B$. These ensembles, respectively, exemplify the phenomena of \textit{Hilbert-space ergodicity} and \textit{deep thermalization}.    
    (e) Our key finding is that these ensembles of states have \textit{maximum entropy}, subject to constraints arising from dynamics. For example, quantum thermodynamics constrains projected ensembles to have a fixed first moment $\rho$. While many ensembles of states share the same first moment $\rho$ (red crosses), the (unnormalized) projected ensemble is the one with the maximum ensemble entropy (bottom).
    } 
    \label{fig:figure0}
\end{figure*}

The second law of thermodynamics is a foundational concept in statistical physics with far reaching implications across many areas of studies, from quantum gravity to information theory and computation~\cite{landau2013statistical,bekenstein1974generalized,cover1999elements,mezard2009information,bennett1985fundamental}. 
At its core is the idea of a \textit{maximum entropy principle}~\cite{jaynes1957information,grandy1980principle,martyushev2006maximum,banavar2010applications,presse2013principles}; that a macroscopic system in equilibrium is described by the state with maximum entropy, subject to physical constraints such as global energy and particle number conservation. 
It is remarkable that such a simple principle enables highly accurate predictions about a large system without knowledge of the vast majority of its degrees of freedom. 

In quantum systems, the maximum entropy principle~\cite{gogolin2011absence,gogolin2016equilibration} is often justified by one of the following two conditions relevant to many conventional experiments [Fig.~\ref{fig:figure0}(a)].
First, one only has access to local or spatially averaged observables~\cite{short2011equilibration,riera2012thermalization,deutsch2018eigenstate,ueda2020quantum}.
Second, one can only measure observables averaged over long time intervals or at late, equilibrium times without fine-grained temporal resolution~\cite{reimann2008foundation, deutsch2018eigenstate,ueda2020quantum}.
In both cases, measurement outcomes become largely agnostic to microscopic details of the initial state or system dynamics and are well-described by thermal density matrices.
It is this {\it inability} to obtain detailed, microscopic information from traditional experiments that render quantum states to be in a probabilistic ensemble and ultimately empower the \MEP~to make accurate and practical predictions of accessible observables.

Modern quantum experiments challenge such traditional settings [Fig.~\ref{fig:figure0}(b)].  
Quantum devices built using atomic or superconducting qubit platforms have realized and probed diverse phenomena ranging from non-equilibrium dynamics~\cite{bernien2017probing,zhang2017observation,neyenhuis2017observation, choi2019probing,mi2021time} to exotic phases of matter~\cite{gross2017quantum,chiu2019string,brydges2019probing,scholl2021programmable,ebadi2021quantum}. These experiments have the capability to fully interrogate the entire quantum system, enabling the measurement of arbitrary, non-local multi-particle correlation functions at the fundamental quantum projection limit. They are also able to deterministically evolve quantum many-body states with a high repetition rate, enabling the study of quantum dynamics with an unprecedented level of accuracy and details. 
Existing maximum entropy principles were not designed to capture such microscopic measurement outcomes, which are non-local and time-resolved.
Therefore, we set out to develop a new approach that may provide an effective statistical model for such data.
Here, we report the discovery of a generalized maximum entropy principle for quantum many-body systems that is well aligned with modern capabilities of experiments.

We specifically study two settings, the \textit{temporal} and \textit{projected ensembles}. The former considers the trajectory of quantum states under time evolution, while the latter considers ensembles of states generated by partial projective measurements performed on the system~[Fig.~\ref{fig:figure0}(c,d)]. In both settings, the average states, or first moments, of the ensembles are described by density matrices $\rho$ which have been extensively studied in the literature~\cite{reimann2008foundation,short2011equilibration,riera2012thermalization,gogolin2016equilibration,deutsch2018eigenstate,ueda2020quantum}. Here, certain universal phenomena apply, such as the fact that the reduced density matrices of small subsystems are well described by thermal Gibbs states~\cite{riera2012thermalization,deutsch2018eigenstate}, or dephasing in the energy-eigenbasis, which leads to the so-called diagonal ensemble as an effective description of equilibrium properties~\cite{gogolin2016equilibration,reimann2008foundation,short2011equilibration}. However, the temporal and projected ensembles of generic systems have not been studied beyond their first moments. Our results extend this universality to higher moments, finding that the first moment uniquely fixes the higher moments.

We find a new type of \textit{maximum entropy principle} that describes ensembles of states in the above settings~[Fig.~\ref{fig:figure0}(e)].  
The entropy in question is an \textit{ensemble entropy}, which quantifies how the states in an ensemble are spread out over the relevant Hilbert space. This is distinct from the conventional von Neumann entropy, which only depends on the average state (density matrix) of the ensemble.  
In broad strokes, our main finding is that in many physical situations, the ensembles of pure states maximize the ensemble entropy up to physical constraints such as energy conservation or local thermalization, thus establishing a new \MEP. 
Special cases of our ensembles of states have been discussed in the literature under the names of ``deep thermalization"~\cite{choi2023preparing,cotler2023emergent,ippoliti2022dynamical,ippoliti2022solvable,claeys2022emergent,lucas2023generalized,wilming2022hightemperature,bhore2023deep,chan2024projected,varikuti2024unraveling,vairogs2024extracting} and ``Hilbert space ergodicity"~\cite{pilatowskycameo2023complete,pilatowskycameo2024hilbertspace} respectively, in particular at infinite effective temperature. Here we present a unified principle that holds under general conditions. 

Ultimately, our work is a theory of fluctuations: the variations in observables over states in an ensemble, which go fundamentally beyond ensemble-averaged observables that are typically studied. In companion experimental work~\cite{shaw2024universal}, we use insights gained from this work to predict and verify universal forms of  fluctuations in a variety of quantities, interpolating from global to local, and in closed as well as open system dynamics.

Our findings have several implications. 
First, we confirm that Hamiltonian dynamics uniformly explores its available Hilbert space. Specifically, the temporal trace of a many-body quantum state satisfies the \MEP~under the constraint of energy conservation. For sufficiently long evolution, the trajectory of the state is only constrained by the population of the energy eigenstates: all other degrees of freedom are randomly distributed with maximum entropy. This improves our understanding of thermalization as a dynamical phenomenon~\cite{mori2017thermalization,mori2018thermalization} and its relationship to static notions of thermal equilibrium~\cite{depalma2015necessity,harrow2022thermalization}. While thermalization has traditionally concerned the dynamics of local quantities ---e.g. how such quantities reach their equilibrium values--- here we study the dynamics of the \textit{global} state, confirming and formalizing our intuition about ergodicity in Hilbert space~\cite{nakata2012phase,pilatowskycameo2023complete}. 
We find necessary and sufficient conditions for Hilbert-space ergodicity and find that curiously, just as in classical dynamics, a system can be ergodic without being chaotic.

Second, we generalize notions of deep thermalization~\cite{choi2023preparing,cotler2023emergent,ippoliti2022dynamical} away from infinite temperature. Deep thermalization refers to the observation that in generic many-body dynamics, not only does the density matrix of a local subsystem reach thermal equilibrium, so do the higher moments of the projected ensemble. Additionally, it has been pointed out that the higher moments may take a longer time than lower moments to converge, implying the possibility of sustained thermalization (in non-local quantities) even when local observables have converged~\cite{ippoliti2022dynamical,ippoliti2022solvable,chan2024projected}. The equilibrium values of such higher moments have considerable structure. At infinite effective temperature, the projected ensemble~\cite{choi2023preparing,cotler2023emergent} has been found to approach the \textit{Haar ensemble}~\cite{ambainis2007quantum,harrow2013church}, a paradigmatic uniform random ensemble of quantum states. We extend this phenomenology to finite temperature, identifying finite temperature generalizations of the Haar ensemble with maximum entropy which describe our ensembles of states. Our result should be taken as complementary to, but distinct from conventional theories of quantum thermalization such as the eigenstate thermalization hypothesis (ETH)~\cite{gogolin2016equilibration,deutsch2018eigenstate}. It does not rely on the ETH, but when used in conjuction, implies that all higher moments of the projected ensemble reach thermal equilbrium with universal, Gibbs-like structure.

Third, our approach predicts and, where relevant, confirms the existence of a universality in natural many-body states, revealed in their higher moments. Not only do mean quantities captured by density matrices $\rho$ exhibit universal thermal behavior, so do higher-order moments such as the variance and skewness, as dictated by our \MEP. Higher order quantities beyond local observables have received recent attention in non-equilibrium dynamics~\cite{rosenberg2023dynamics,wienand2023emergence}, and our results suggest at the broader applicability of higher moments --- the \MEP~holds not only at the level of density matrices~\cite{rigol2008thermalization,polkovnikov2011nonequilibrium}, but may also hold at the level of many-body wavefunctions~\cite{brody2000information,anza2022quantum,hahn2023statistical}. This possibility has been recently investigated for energy eigenstates in Ref.~\cite{hahn2023statistical}. Our results support such an approach based on the \MEP, and suggest that they may be more widely applicable; here we primarily apply them to time-evolved many-body states.

Finally, we draw further connections between thermalization and quantum information theory. Ensembles of quantum states are natural objects for quantum communication: fundamental results such as the Holevo bound investigate the use of such ensembles to transmit information~\cite{holevo1998capacity}. Our results indicate that ensembles of states that result from many-body dynamics are both maximally difficult to compress, distinguish, or use for information transmission, providing a new perspective on the nature of information scrambling and hiding in many-body dynamics~\cite{kukuljan2017weak,nahum2018operator,mi2021information,barratt2022transitions,ippoliti2023learnability}. 

The maximum entropy principle has repeatedly exceeded expectations in its ability to describe complex, interacting systems~\cite{jaynes1957information,grandy1980principle,martyushev2006maximum,banavar2010applications,presse2013principles}. Most results have so far been restricted to classical or single-particle quantum physics. Our work represents a first step to adapting this principle to the modern setting of many-body physics. 
While we were able to firmly establish this principle in certain settings, we find that it may apply more generally in many-body physics. Indeed, see Ref.~\cite{altland2022maximum} for a recent application of this principle to many-body physics to full-counting statistics in systems with a $U(1)$ symmetry. These developments hint at broader applications of the maximum entropy principle in many-body physics. 

\section{Summary of results}
\label{sec:summary_of_results}

\begin{table*}[]
\renewcommand{\arraystretch}{1.5}
    \centering
    \begin{tabular*}{\textwidth}{@{\extracolsep{\fill}}p{0.16\textwidth}  p{0.12\textwidth}  p{0.19\textwidth}  p{0.19\textwidth}  p{0.24\textwidth}}
    \toprule
    Ensemble & \RaggedRight{Statistical description} & \RaggedRight{Maximum entropy principle}  & \RaggedRight{Measurable signature} & \RaggedRight{Information theoretic properties}\\
    \midrule
    Temporal ensemble   & \RaggedRight{Random~phase ensemble  (Sec.~\ref{sec:temporal_ensemble}) } & \RaggedRight{Maximum entropy with fixed energy populations $\{|c_E|^2\}$}  & \RaggedRight{Porter Thomas (PT) dist.~in outcome probabilities $p_z(t)$ (Thm.~\ref{thm:temp_ens_PT}).} & \RaggedRight{Constant mutual information [Eq.~\eqref{eq:temp_MI}]\newline$I(Z; T) \approx (1-\gamma)/\ln 2$ } \\
    \RaggedRight{Projected ensemble: \newline Energy non-revealing basis} & \RaggedRight{Scrooge ensemble (Sec.~\ref{sec:scrooge})} & \RaggedRight{Maximum entropy with fixed first moment $\rho_A$}&  \RaggedRight{PT dist.~in joint probabilities $p(o_A,z_B)$ [Eq.~\eqref{eq:distorted_Gaussian_kth_moment_result}]}  & \RaggedRight{Minimal mutual information [Eq.~\eqref{eq:subentropy}] \newline $I(O_A;Z_B) \approx Q(\rho)$ }  \\
    \RaggedRight{General measurement basis}& \RaggedRight{Generalized Scrooge ensemble (Sec.~\ref{sec:general_scrooge})} & \RaggedRight{Collection of maximum-entropy ensembles with fixed first moments $\rho_d(z_B)$}& \RaggedRight{PT dist. in normalized probabilities $p_{o_A}(z_B)/p_{d,o_A}(z_B)$ [Eq.~\eqref{eq:PT_in_gen_Scrooge}]} & {Minimal~interaction} information [Eq.~\eqref{eq:weighted_subentropies}] \newline $I(O_A;Z_B;T) \approx \mathbb{E}_{z_B} Q(\bar{\rho}(z_B))$ \\
    \bottomrule
    \end{tabular*}
    \caption{Summary of key results in this work. We describe the key ensembles considered in this work: the temporal ensemble and the projected ensembles generated by measurements in a basis $\{z_B\}$ under two different cases: when $\{z_B\}$ is \textit{energy non-revealing} and when it is general. 
    Their statistical descriptions are probability distributions which have maximum entropy under appropriate constraints. We finally discuss their measurable signatures and information theoretic properties.
    }
    \label{tab:summary}
\end{table*}    

In this work, we study ensembles of quantum states.
These ensembles $\mathcal{E}$ are collections of pure quantum states $\ket{\Psi_j}$, labeled by an index $j$, with an associated probability distribution $p_j$, that is, $\mathcal{E}  = \{p_j, \ket{\Psi_j}\}$. 
Quantum state ensembles can be discrete (i.e.~the index $j$ takes discrete values, such as a discrete set of measurement outcomes) or continuous (the index $j$ takes continuous values, such as time). In situations where the label is not important, we simply write $\mathcal{E}  = \{P(\Psi) ,\ket{\Psi}\}$. 

Quantum state ensembles have traditionally been studied in quantum information for tasks such as quantum communication~\cite{nielsen_chuang_2010}. However, until recently they have not been necessary in the context of many-body physics. Under traditional measurements, observations made on an ensemble of states only depend on the first moment $\rho^{(1)}_\mathcal{E} \equiv \int d\Psi P(\Psi) |\Psi\rangle\!\langle \Psi|$. This is because without the ability to determine the index $j$, measurements of an observable $O$ would yield its expected value, averaged over the ensemble, i.e.~$\int d\Psi P(\Psi) \langle \Psi|O|\Psi \rangle = \text{tr}(O \rho^{(1)}_\mathcal{E})$. Therefore, in this setting, one can only operate at the level of mixed states $\rho^{(1)}_\mathcal{E}$~\cite{reimann2008foundation,short2011equilibration,riera2012thermalization,gogolin2016equilibration,deutsch2018eigenstate,ueda2020quantum}.

Modern experiments challenge this assumption and may involve full knowledge of the ensemble $\mathcal{E}$. In modern quantum devices, one can determine the index $j$ in a variety of settings. For example, in the ensembles we study, $j$ will be the evolution time, which can be precisely controlled, or $j$ will be the outcome of a projective measurement which is recorded by the experimenter. When the index $j$ can be determined, we can study a wider range of properties, such as the \textit{fluctuations} of the expectation values $\langle \Psi_j | O |\Psi_j\rangle$. The simplest such quantity is the variance $\int d\Psi P(\Psi)\langle \Psi | O |\Psi\rangle^2 - \big[\int d\Psi P(\Psi) \langle \Psi | O |\Psi\rangle\big]^2$. This quantity depends on more than the first moment $\rho^{(1)}_\mathcal{E}$. Specifically, it depends on the second moment $\rho^{(2)}_\mathcal{E} \equiv \int d\Psi P(\Psi) |\Psi\rangle\!\langle \Psi|^{\otimes 2}$ since the first term is equal to $\int d\Psi P(\Psi)\langle \Psi | O |\Psi\rangle^2 = \text{tr}(\rho^{(2)}_\mathcal{E} O^{\otimes 2})$. In general, the statistical properties of quantum state ensembles can be systematically characterized through their moments
\begin{align}
 \rho^{(k)}_\mathcal{E} \equiv \int d\Psi P(\Psi) \ketbra{\Psi}{\Psi}^{\otimes k},
 \end{align}
for positive integer $k$\footnote{For discrete ensembles, the integral reduces to a sum over all states in the ensemble, see Sec.~\ref{sec:def_pure_state_ensembles}.}. 

Recent work has illustrated that rich physical phenomena are encoded in the higher moments of state ensembles. As a concrete, but non-exhaustive example, higher moment quantities are necessary to detect the measurement-induced phase transition in monitored quantum circuits~\cite{fisher2023random,potter2022entanglement,li2018quantum,skinner2019measurement}. In this setting, an ensemble of pure states is generated by the different outcomes of mid-circuit measurements. A phase transition is observed in quantities such as the bipartite entanglement entropy averaged over states in the ensemble. Their average $k$-th R{\'e}nyi entropies can be computed from the $k$-th moment of the ensemble, through the quantity $\int d\Psi P(\Psi) \text{tr}_A(\text{tr}_B(|\Psi\rangle\!\langle \Psi|)^k)$, where $A$ and $B$ are a bipartition of the system.

Without reference to a specific observable, the fluctuations of quantum state ensembles can be studied by examining how the states are distributed in Hilbert space. We quantify this spread through an \textit{ensemble entropy}, which is the Shannon entropy of the probability distribution $P(\Psi)$. For continuous ensembles, this requires the choice of a reference distribution, in our case, the Haar ensemble.  Thus, for continuous quantum state ensembles, the  ensemble entropy is given  as the (negative) Kullback-Liebler (KL) divergence between $P(\Psi)$ and the Haar ensemble
\begin{align}
    \text{Ent}(\mathcal{E}) &\equiv -\mathcal{D}(\mathcal{E} \Vert \text{Haar}) \label{eq:ens_entropy_intro}\\
    &\equiv -\int d\Psi~P(\Psi) \log_2\left(\frac{P(\Psi)}{P_\text{Haar}(\Psi)}\right)~. \nonumber
\end{align}

We find that the ensembles we study are statistically described by universal continuous\footnote{Even when the ensembles are finite, such as the projected ensemble, we treat them as finite samples from an underlying continuous ensemble.} state ensembles with maximum ensemble entropy, up to constraints imposed by energy conservation or effective constraints imposed by thermalization. We note that the ensemble entropy differs from the traditionally considered von Neumann entropy $S(\rho^{(1)}_\mathcal{E}) = - \text{tr}[\rho^{(1)}_\mathcal{E} \log_2 \rho^{(1)}_\mathcal{E}]$ which only depends on the average state $\rho^{(1)}_\mathcal{E}$. Instead, $ \text{Ent}(\mathcal{E})$ depends on the distribution of the individual states $P(\Psi)$, which is determined by all moments of $\mathcal{E}$, and has distinct information-theoretic properties (Fig.~\ref{fig:info_theory}). We also remark that while similar notions have been studied in the past, e.g.~in Refs.~\cite{brody2000information,anza2022quantum}, this work explicitly shows the link between this entropy measure and ensembles of states obtained from dynamics, the first to do so to our knowledge.

In the following, we demonstrate this result by focusing on two specific settings ---temporal and projected ensembles--- and their relationship. Finally, we discuss implications of our results, summarized in Table~\ref{tab:summary}.

\subsection{Temporal ensembles are random phase ensembles}
The first ensemble of states we study is the \textit{temporal ensemble}, consisting of quantum states generated by time evolution of an initial state $\ket{\Psi_0}$ under a Hamiltonian $H$~[Fig.~\ref{fig:temporal_ensemble}(a)], with probabilities being uniform on the time interval $t\in [-\tau/2,\tau/2]$.
\begin{align}
    \mathcal{E}_\text{Temp.}(\tau) &\equiv \left\{\left(1/\tau,\ket{\Psi(t)}\right)\right\}_{t=-\tau/2}^{\tau/2} \\
    &\equiv \left\{(1/\tau,\exp(-i H t) \ket{\Psi_0})\right\}_{t=-\tau/2}^{\tau/2}~. \nonumber
\end{align}
This ensemble consists of the trajectory of a quantum state as it evolves under $H$. The statistical properties of the trajectory quantify its degree of ergodicity in Hilbert space, as our results will make clear. 

Conventionally, quantum ergodicity is formulated in terms of the energy eigenstates of a system~\cite{stechel1984quantum,srednicki1999approach,goldstein2010long}. While this applies to time-independent Hamiltonian evolution or Floquet dynamics, it is not applicable to quantum dynamics in general. Furthermore, the notion of ergodicity is fundamentally dynamical in nature, and its definition in terms of static quantities such as energy eigenstates illustrates a tension in our understanding of quantum ergodicity.
Recent works~\cite{pilatowskycameo2023complete,pilatowskycameo2024hilbertspace} challenge this formulation. Specifically, they show that the ergodicity of certain time-dependent Hamiltonian dynamics cannot be formulated using conventional approaches based on the statistical properties of \mbox{(quasi)-energy} eigenstates, but instead can be characterized by the statistical properties of their trajectories.
The central idea of Refs.~\cite{pilatowskycameo2023complete,pilatowskycameo2024hilbertspace} is to test if a time-evolved state uniformly visits all ambient points in the Hilbert space, a property dubbed (complete) Hilbert-space ergodicity.

Our work generalizes the notion of Hilbert-space ergodicity~\cite{pilatowskycameo2023complete} to systems with energy conservation.
The presence of energy conservation explicitly disallows  ergodicity over all of Hilbert space. Nevertheless, we find that many-body dynamics under time-independent Hamiltonian evolution is still Hilbert-space ergodic, in the sense that the evolved states uniformly explore their available Hilbert space, which is constrained due to energy conservation.
In order to understand the effects of energy conservation, consider the simple case of a single qubit initially polarized in the $x$-direction, subject to a magnetic field along the $z$-axis. Over time, the qubit periodically traces out an equatorial circle on the Bloch sphere, instead of uniformly covering the sphere. In contrast, quantum dynamics without such conservation laws, specifically a quasiperiodically driven system, can uniformly explore Hilbert space~\cite{pilatowskycameo2023complete}. 

As in classical dynamics~\cite{lebowitz1973modern}, we find that a quantum many-body system may be Hilbert-space ergodic without being quantum chaotic (i.e. being integrable). We show that Hilbert-space ergodicity holds as long as the system satisfies the \textit{no-resonance conditions} (Definition~\ref{def:k-no-resonance}), colloquially the absence of higher order resonances in the spectrum (Definition~\ref{def:k-no-resonance}). This condition is widely assumed to be true in chaotic systems~\cite{srednicki1999approach,reimann2008foundation,linden2009quantum,short2011equilibration}, and we numerically observe that it is also true in certain integrable systems such as the XXZ model, but not in others such as the transverse field Ising model.

To prove our claim of Hilbert-space ergodicity, we show that in the $\tau\rightarrow \infty$ infinite interval limit, the temporal ensemble is equal to the \textit{random phase ensemble}~\cite{nakata2012phase,nechita2021graphical}. This is the ensemble of states with fixed amplitudes $|c_E| \equiv |\langle E| \Psi_0\rangle |$ in some basis $\{|E\rangle \}$, here the energy eigenbasis, and with complex phases $\text{arg}(c_E)$ independently random from the uniform distribution $\text{Unif}([0,2\pi))$. In a time-evolved state, the amplitudes $|c_E|^2$ are fixed by the initial state and Hamiltonian, and the phases are the only degrees of freedom. 

\noindent \textbf{Theorem~\ref{thm:temp_ens}} (informal): \textit{The infinite-time temporal ensemble obtained by evolving an initial state is equal to a random phase ensemble if the Hamiltonian $H$ satisfies all $k$-th no-resonance conditions. Conversely, for almost every initial state, if the infinite-time temporal ensemble is equal to a random phase ensemble, the Hamiltonian $H$ satisfies all $k$-th no-resonance conditions.}

The above Theorem is somewhat of a folklore result~\cite{linden2009quantum,goldstein2010long,nakata2012phase}, often stated without proof. Here we provide a full statement and proof, for several reasons. Firstly, Theorem~\ref{thm:temp_ens} has a natural interpretation in terms of a maximum entropy principle.
The random phase ensemble has maximum ensemble entropy [Eq.~\eqref{eq:ens_entropy_intro}] on this restricted Hilbert space, and in Sec.~\ref{subsec:finite_temp_temp_ens}, we quantify the rate of convergence of the temporal ensemble to the random phase ensemble with increasing interval $\tau$. Next, this leads to a novel type of universality in the many-body setting, which we outline below.

Unlike the above example of a two-level system, the exponentially larger Hilbert space in many-body systems leads to behaviour much more akin to ergodicity in the full Hilbert space. Specifically, we show that the trajectory appears pseudorandom in the following sense: in the limit of large Hilbert space dimension $D$, the temporal ensemble is \textit{statistically equal} to the Haar ensemble, \textit{deformed} as follows
\begin{equation}
\ket{\Psi(t)} \overset{\text{stat.}}{\approx} \sqrt{D\rho_d}\ket{\phi}~,~\ket{\phi}\sim \text{Haar},  
\label{eq:temporal_stat_equality}
\end{equation}
where $\rho_d \equiv \sum_E |E\rangle\!\langle E|\Psi_0 \rangle\!\langle \Psi_0|E\rangle\!\langle E|$ is the diagonal ensemble. By ``statistically equal," we mean that Eq.~\eqref{eq:temporal_stat_equality} is not true on the level of wavefunctions, but the moments $\rho^{(k)}$ of the ensembles are very close. To show this, we develope a simple analytic expression for the moments of the temporal ensemble that holds in the $D\rightarrow \infty$ limit~(Sec.~\ref{subsec:temporal_ensemble_product_form}). One consequence of this pseudorandomness is our finding that the measurement probabilities, upon appropriate rescaling, universally follow the Porter-Thomas distribution (below, and Sec.~\ref{subsec:temp_ens_PT_dist}), a hallmark of random quantum states which had been observed in quantum dynamics without conservation laws~\cite{arute2019quantum}. We discuss information-theoretic aspects of this result and also prove an analog of the ergodic theorem, a result in the theory of dynamical systems~\cite{brin2002introduction}, where we relate the statistics defined over time and configuration space (discussed below and in Sec.~\ref{subsec:ergodicity}).

\subsection{Projected ensembles are (generalized) Scrooge ensembles}
Not only do temporal ensembles satisfy the maximum entropy principle, so do \textit{projected ensembles}. Projected ensembles, introduced in a modern context by Refs.~\cite{choi2023preparing,cotler2023emergent} (and previously studied in Ref.~\cite{goldstein2016universal}), are ensembles of states obtained from a single, larger many-body state $\ket{\Psi}$. To do so, we partition the system into two parts $A$ and $B$, such that the Hilbert space factorizes (for simplicity of discussion) as $\mathcal{H} = \mathcal{H}_A \otimes \mathcal{H}_B$. We then measure $B$ in a basis $\{\ket{z_B}\}$. Each measurement outcome $\ket{z_B}$ projects $\ket{\Psi}$ into a distinct pure state $\ket{\Psi(z_B)} \in \mathcal{H}_A$~[Fig.~\ref{fig:figure0}(d)], defined as
\begin{equation}
    \ket{\Psi(z_B)} \equiv (\mathbb{I}_A \otimes \bra{z_B})\ket{\Psi}/\sqrt{p(z_B)}~,
\end{equation}
where the normalization constant is the probability of obtaining the outcome $z_B$, $p(z_B) \equiv \bra{\Psi}(\mathbb{I}_A \otimes \ketbra{z_B}{z_B})\ket{\Psi}$. This defines the projected ensemble:
\begin{align}
    \mathcal{E}_\text{Proj.} = \{p(z_B), \ket{\Psi(z_B)} \}_{z_B}~.
\end{align}
For concreteness, we primarily study global states $\ket{\Psi} = \exp(-iHt)\ket{\Psi_0}$ obtained from time-evolution by a generic, ergodic many-body Hamiltonian $H$, which rapidly reach thermal equilibrium at a temperature set by the energy of the initial state~\cite{rigol2008thermalization,polkovnikov2011nonequilibrium}. We will also study eigenstates of $H$ that exhibit similar thermal properties.

In Refs.~\cite{choi2023preparing,cotler2023emergent}, it was found that when the state $\ket{\Psi}$ is obtained from time-evolution under a chaotic many-body Hamiltonian at infinite effective temperature, the projected ensemble of small subsystems is very well-approximated by the Haar ensemble. This phenomenon is known as \textit{deep thermalization}~\cite{ippoliti2022dynamical}. This has since been rigorously shown in a variety of settings, including deep random circuits~\cite{cotler2023emergent,chan2024projected}, dual-unitary models~\cite{ho2021exact,ippoliti2022dynamical,ippoliti2022solvable,claeys2022emergent,shrotriya2023nonlocality}, free fermion models~\cite{lucas2023generalized}, or Hamiltonian dynamics with certain assumptions~\cite{wilming2022hightemperature}. Various aspects including the rate of deep thermalization~\cite{ippoliti2022dynamical,ippoliti2022solvable,shrotriya2023nonlocality,chan2024projected}, the effect of symmetries~\cite{bhore2023deep,varikuti2024unraveling}, and the presence of magic in the global state~\cite{vairogs2024extracting} have been investigated. The majority of what we know so far is restricted to the setting of infinite temperature dynamics. In this work, we set out to study the projected ensembles formed in finite temperature dynamics. 

At finite temperature, projected ensembles are not close to the Haar ensemble. To see this, note that the first moment of the projected ensemble $\rho^{(1)}_\mathcal{E}\equiv \sum_{z_B} p(z_B) \ketbra*{\Psi(z_B)}{\Psi(z_B)} = \text{tr}_B[\ketbra{\Psi}{\Psi}]$ is precisely equal to the reduced density matrix (RDM). In turn, only at infinite temperature is the RDM equal to the first moment of the Haar ensemble, the maximally mixed state $\mathbb{I}_A/D_A$. This follows from a standard result in quantum thermalization~\cite{srednicki1999approach} that the reduced density matrix $\rho_A$ is well described by a  Gibbs state, i.e.~$\rho_A\propto \exp(-\beta H_A)$, where $H_A$ is the Hamiltonian truncated to the subsystem $A$, and $\beta^{-1}$ is an effective temperature. It has remained an open question as to how to describe these projected ensembles at finite temperature. 

In this work, we find that a maximum entropy principle still applies: under an ``energy non-revealing" condition on the measurement basis $\{|z_B\rangle\}$ which we shall discuss below, the projected ensembles are described by the ensembles that have maximum entropy, over all ensembles with a given first moment. These ensembles have been previously studied and were termed \textit{Scrooge ensembles}~\cite{jozsa1994lower,goldstein2006distribution, reimann2008typicality} for their unique information-theoretic properties\footnote{As a technical point, Scrooge ensembles only have maximum entropy when they are regarded as ensembles of \textit{unnormalized states}~\cite{goldstein2006distribution}, see Sec.~\ref{sec:def_pure_state_ensembles}.}. Each density matrix $\rho$ uniquely defines a Scrooge ensemble $\text{Scrooge}(\rho)$, and our result is that at finite temperatures, the projected ensemble is the Scrooge ensemble corresponding to the appropriate thermal RDM.

\noindent \textbf{Claim \ref{claim:Scrooge}} (Corollary of Theorem~\ref{thm:moment_of_general_projected_ensemble} below, informal): \textit{The projected ensemble is statistically described by the Scrooge ensemble, if the measurement basis $\{\ket{z_B}\}$ is energy non-revealing.}

The Scrooge ensemble may be essentially understood as a deformation of the Haar ensemble such that its first moment equals $\rho$~\cite{jozsa1994lower}~[Fig.~\ref{fig:projected*}(b)]. 
Specifically, we can form the Scrooge ensemble by first sampling Haar-random states $\ket{\psi} \sim \text{Haar}$, then \emph{deforming} them by a given density matrix 
$\rho$, i.e. 
\begin{equation}
\ket{\Phi} \sim \text{Scrooge}(\rho)~\Longleftrightarrow~ \ket{\Phi} \propto \sqrt{\rho}\ket{\psi},~ \ket{\psi} \sim \text{Haar} ~.
\end{equation}
The probability $p(\Phi)$ is further weighted by the norm $\bra{\psi}\rho\ket{\psi}$, leading to the distribution of states 
 \begin{align}
&\mathcal{E}_{\text{Scrooge}}(\rho)= \left\{
\begin{aligned}
  &p(\Phi) =D\bra{\psi}\rho\ket{\psi},\\
  &\ket{\Phi} = \frac{\sqrt{\rho}\ket{\psi}}{\bra{\psi}\rho\ket{\psi}^{1/2}}
\end{aligned}\right\}_{\ket{\psi}\sim \text{Haar}}.
\end{align}
In Sec.~\ref{subsec:scrooge_moment}, we give closed form expressions for the resultant probability distribution function $p(\Phi)$, as well as its $k$-th moments.

Scrooge ensembles do not always describe the projected ensemble. They only do so in the special case when the measurement basis $z_B$ is uncorrelated with the energy, defined precisely in Sec.~\ref{subsec:gen_scrooge_to_scrooge}. Under these circumstances, we can treat the projected states $\ket{\Psi(z_B)}$ as independent, identical samples from a distribution. However, this is not always true: when $z_B$ is correlated with the subsystem energy, it is also correlated with the projected state $\ket*{\Psi(z_B)}$, such that the energies $E_A \equiv \langle \Psi(z_B)|H_A |\Psi(z_B)\rangle$ and $E_B \equiv \langle z_B | H_B |z_B \rangle$ are anticorrelated. Under these general conditions, we introduce the \textit{generalized Scrooge ensemble} to describe the projected ensemble (Section~\ref{sec:general_scrooge}). Each projected state $|\Psi(z_B)\rangle$ should be treated as a sample from a \textit{different} Scrooge ensemble for each $z_B$, i.e.~where the first moment is not the RDM $\rho_A$, but is a $z_B$-dependent mixed state $\bar{\rho}(z_B)$, which we may determine by the time-average of the projected state, $\mathbb{E}_t[|\Psi(z_B,t)\rangle \! \langle \Psi(z_B,t)|]$. Therefore, the generalized Scrooge ensemble is more accurately a collection of ensembles, where the projected ensemble is comprised of exactly one sample from each ensemble, but we will abuse notation and refer to it as an ensemble. This is our second key result.

\noindent \textbf{Theorem~\ref{thm:moment_of_general_projected_ensemble}} (informal): \textit{At \textrm{sufficiently} long times $t$,\footnote{While our proof requires times that are exponential in system size, we find from numerical experiments that we only require times that grow weakly (at most polynomially) with system size.} the projected ensemble generated by a time-evolved state $\ket{\Psi(t)}$ is approximately equal to the generalized Scrooge ensemble, which depends on the measurement basis $\{z_B\}$, Hamiltonian $H$, and initial state $\ket{\Psi_0}$, if the Hamiltonian $H$ satisfies all $k$-th no-resonance conditions. The converse direction holds for almost every initial state $\ket{\Psi_0}$.}

When the measurement basis $\{z_B\}$ is \textit{uncorrelated} with the global energy, the generalized Scrooge ensemble reduces to the Scrooge ensemble (Claim~\ref{claim:Scrooge}). Specifically, this happens when the mixed state $\bar{\rho}(z_B)$ is independent of $z_B$ (up to a rescaling), giving a concrete definition for ``energy non-revealing" bases. 

We prove this result using our knowledge of the temporal ensemble. We show that the temporal trajectory of the projected state $|\Psi(z_B,t)\rangle$ is a Scrooge ensemble (here, because of the subsystem projection, there are no hard constraints on the amplitudes, as there were in the case of the global state). We then utilize our ergodic theorem (Sec.~\ref{subsec:ergodicity}) to translate our results for a distribution over time to a distribution over the outcomes $z_B$.

While Theorem~\ref{thm:moment_of_general_projected_ensemble} concerns the time-evolved states $\ket{\Psi(t)}$, we numerically observe that the generalized Scrooge ensemble also describes the projected ensembles obtained from energy eigenstates, suggesting that our \MEP~holds for a broader class of many-body states.

In projected ensembles, the maximum entropy principle sheds light on the nature of correlations between system and bath.
Specifically, the projected ensemble reveals correlations between the subsystem $A$ and its complement $B$, which we quantify in terms of the mutual information. 
We find that the \emph{nontrivial} correlations between them are minimal and basis independent. Here, nontrivial correlations mean any amount of correlation (mutual information) that is beyond the thermal correlations that remain after infinite-time averaging. In other words, any meaningful correlations between $A$ and $B$ originate only from energy conservation, and other symmetries if they exist.
Below and in Sec.~\ref{subsec:scrooge_mutual_info} and~\ref{subsec:gen_Scrooge_info}, we discuss consequences of these correlations for tasks such as communication and compression.

\subsection{Implications and discussion}
Below, we summarize several broader points that we found applicable to both the temporal and projected ensemble.

\textit{Ergodic theorem ---} Remarkably, the temporal and projected ensembles are closely related. We prove that a version of the ergodic theorem holds in many-body dynamics. This enables us to translate statements about distributions over time (temporal ensemble) into statements about distributions over measurement outcomes, at fixed, typical points in time. In particular, we use this to prove
Theorem~\ref{thm:moment_of_general_projected_ensemble}, a statement about the projected ensembles at typical points in time.

\textit{Operational meaning of ensemble entropy ---} In this work, we use the ensemble entropy to define our \MEP. We also identify its operational interpretation: the ensemble entropy quantifies the information required to classically represent an ensemble of states. Our maximally entropic ensembles are therefore the most difficult to compress, formalizing our intuition that the states generated by natural, chaotic many-body dynamics lack overall structure that can be exploited for compression. 

Specifically, the ensemble entropy quantifies the number of classical bits required to store an ensemble of states up to numerical precision $\epsilon$~[Fig.~\ref{fig:info_theory}(a)]. This should be contrasted with the minimum number of \textit{qubits} required to store the ensemble in quantum memory, which was established by Schumacher to depend only on the first moment $\rho_\mathcal{E}$ of the ensemble through the von Neumann entropy $S(\rho_\mathcal{E})$~\cite{schumacher1995quantum}.

\textit{Information-theoretic consequences ---} Not only are our ensembles of states maximally difficult to classically store, they are also maximally difficult to distinguish by measurement, as quantified by their \textit{accessible information}~\cite{jozsa1994lower,nielsen_chuang_2010}.
The accessible information quantifies the maximum rate of classical information that can be transmitted by sending states from the ensemble over a quantum channel (in terms of bits of information per channel use) [Fig.~\ref{fig:info_theory}(b)]. It is upper-bounded by the well-known Holevo bound~\cite{holevo1998capacity}, and lower-bounded by Josza, Robb, and Wooters~\cite{jozsa1994lower}. They found that the Scrooge ensemble achieves this lower bound, dubbed the \textit{subentropy} $Q(\rho)$ [Eq.~\eqref{eq:subentropy}]. They dubbed the ensemble the ``Scrooge ensemble" because it is uniquely ``stingy" in the amount of information revealed under measurement. Therefore, the projected ensemble not only has maximum ensemble entropy, it also minimizes the amount of information that can be extracted. These information theoretic properties are directly measurable [Fig.~\ref{fig:projected*}(f) and Fig.~\ref{fig:generalized_scrooge}(d)] but have not yet been observed in experiment. 

In the case of temporal ensembles, the accessible information quantifies the amount of information that can be gained about the evolution time of the state. We find that it takes on a universal value $(1-\gamma)/\ln 2 \approx 0.61$, where $\gamma$ is the Euler-Mascheroni constant [Fig.~\ref{fig:temporal_ensemble}(d)]. 

\begin{figure}
    \centering
\includegraphics[width=0.95\columnwidth]{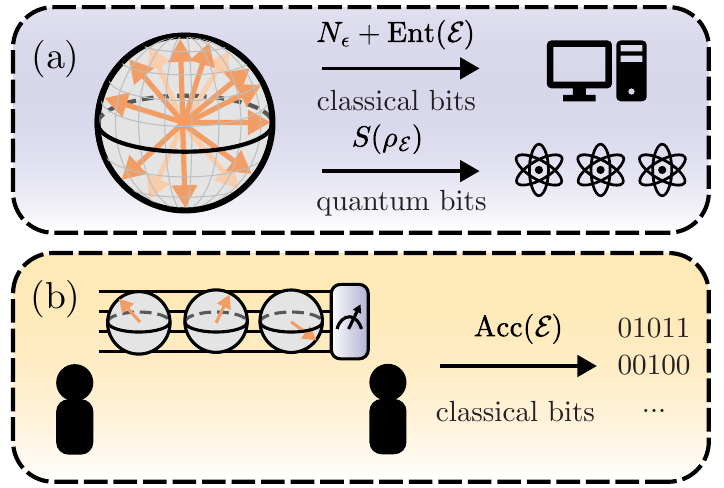}
    \caption{Information theoretic implications of our work. Our results indicate that two tasks are maximally difficult for our ensemble of states: (a) Information compression of an ensemble of states $\mathcal{E}$. The ensemble entropy quantifies the task of classically representing $M$ states from $\mathcal{E}$ up to precision $\epsilon$ (denoted here by a fine grid partitioning the Hilbert space into $2^{N_\epsilon}$ elements). This task requires $M\cdot [N_\epsilon+\text{Ent}(\mathcal{E})]$ classical bits. This classical storage task differs from the analogous quantum task: storing the same $M$ states in quantum memory only requires $M S(\rho_\mathcal{E})$ qubits, where the $S(\rho_\mathcal{E})$ is the von Neumann entropy of the average state $\rho_\mathcal{E}$.  (b) We also consider  a task originally studied by Holevo~\cite{holevo1998capacity}: sending classical information by sending states from $\mathcal{E}$ over a quantum channel. This channel capacity is given by the accessible information $\text{Acc}(\mathcal{E})$, which our ensembles of states uniquely minimize.}
    \label{fig:info_theory}
\end{figure}

\textit{Measurable signature: universal Porter-Thomas distribution ---} Our \MEP~selects distributions of states that are ``closest" to the Haar ensemble, satisfying certain constraints. 
The relationship between our ensembles of states to the Haar ensemble gives rise to a measurable signature of our \MEP: the \textit{Porter-Thomas (PT) distribution}~\cite{porter1956fluctuations}. Originally identified in nuclear physics settings, the PT distribution has received renewed attention in many-body physics. This distribution represents the distribution of overlaps between Haar random states $\ket{\psi}$ with a fixed (but arbitrary) state $\ket{\Phi}$, i.e.~the probability $p(\Phi) \equiv |\langle \Phi | \psi \rangle |^2$ of measuring $\ket{\Phi}$ from a Haar-random state $\ket{\psi}$~\cite{boixo2018characterizing}. Not only is the PT distribution a measurable signature of Haar-random states, we find that it is also present in many natural many-body states. In companion work~\cite{shaw2024universal}, we experimentally and theoretically investigate the emergence of the Porter-Thomas distribution in global quantities, its relationship to local quantities, and how the Porter-Thomas distribution is modified in the presence of noise, using our results to learn about the noise in an experiment.

\subsection{Organization of paper}
The rest of the paper is organized as follows. In Section~\ref{sec:def_pure_state_ensembles}, we provide an overview of the central object of our work --- ensembles of pure states and their information-theoretic and measurable properties. We then proceed to study our two ensembles of states: the temporal ensemble (Section~\ref{sec:temporal_ensemble}) and the projected ensemble in a special and a general setting (Sections~\ref{sec:scrooge} and \ref{sec:general_scrooge}). Finally, in Section~\ref{sec:eigenstate_PT} we discuss the PT distribution in the properties of many-body eigenstates.

\section{Ensembles of pure states}
\label{sec:def_pure_state_ensembles}
Ensembles of pure states generated by quantum many-body dynamics are the central objects of study in this work. In this section, we provide an overview of ensembles of pure states and their properties. After defining normalized and unnormalized ensembles of states, we introduce two quintessential ensembles: the Haar and Gaussian ensembles, which we use as the reference distributions for our ensemble entropy. We then discuss the quantities we use to analyze our ensembles of states: the $k$-th statistical moments, ensemble entropy, mutual information, and PT distribution.

Ensembles of pure states are collections of pure states in a Hilbert space of dimension $D$, weighted by a probability distribution. 
\begin{definition}
An \textit{ensemble of pure states} is a probability distribution of states $\ket{\Psi_j}$ weighted by probabilities $p_j$:
\begin{equation}
    \mathcal{E} = \{p_j, \ket{\Psi_j}\}
\end{equation}
when the $p_j$'s are omitted, we take the states $\ket{\Psi_j}$ to be uniformly distributed over indices $j$~\footnote{We ignore subtleties about the global phase of the states $|\Psi_j\rangle$, and assume that the states are gauge-fixed to, for example, have real positive values of the first entry $\langle u_1|\Psi_j\rangle$ in some basis $\{|u_i\rangle\}$.}. 
\end{definition}

Our distributions may be continuous or discrete. In the continuous case, the distribution is a measure over normalized wavefunctions, which we can specify in terms of its components $\Psi_z \equiv \braket{z}{\Psi}$, subject to the constraint that $\sum_z \abs{\Psi_z}^2 = 1$:
\begin{equation}
    P(\Psi) d\Psi = \int P(\{\Psi_z\}) \delta(\Vert\Psi \Vert^2-1) \prod_z d^2 \Psi_z~, 
\end{equation}
where $\prod_z d^2 \Psi_z$ is the Euclidean measure on $\mathbb{C}^{\otimes D}$.

\textit{Unnormalized ensembles of states ---} We will also find it useful to define distributions of \textit{unnormalized} states. In this case, the distribution is simply a measure over $D$-dimensional complex vectors, defined in terms of its components $\tilde{\Psi}_z \equiv \braket*{z}{\tilde{\Psi}}$. For consistency, we will use tildes throughout this work to denote unnormalized states $\ket*{\tilde{\Psi}}$.

An unnormalized ensemble can be mapped to a normalized one by normalizing each states. This results in the probability distribution $P(\Psi) d\Psi = \int d\tilde{\Psi} P(\tilde{\Psi}) \delta(\ket*{\tilde{\Psi}}\!/\Vert \tilde{\Psi} \Vert - \ket{\Psi}) $. For each state $\ket{\Psi}$, the normalized ensemble averages over the distribution of the norm $\Vert \tilde{\Psi} \Vert$ to give the measure $P(\Psi)$.
 Therefore, the unnormalized ensemble contains more information than the normalized one: specifically information about fluctuations of the norm $\Vert \tilde{\Psi} \Vert$.

\subsection{Gaussian and Haar random ensembles of states}
The ensemble entropy, which quantifies our \MEP, is defined with reference to the Haar (or Gaussian) ensemble. These ensembles are unique in being invariant under arbitrary unitary transformations and have been the subject of considerable study as paradigmatic distributions of random quantum states~\cite{harrow2013church}.

\textit{Random unnormalized states: the Gaussian ensemble ---} The Gaussian ensemble is a probability distribution over unnormalized states, where each component $\tilde{\Psi}_z\equiv \braket*{z}{\tilde{\Psi}}$ is an independent complex Gaussian variable.
\begin{align}
P_\text{Gauss}(\tilde{\Psi}) d \tilde{\Psi} = \prod_z \frac{D}{\pi} \exp[-D \vert \tilde{\Psi}_z \vert^2] d^2 \tilde{\Psi}_z~.
\label{eq:Gaussian_ens}
\end{align}
The normalization above is chosen so that the average norm is $\mathbb{E}[\braket*{\tilde{\Psi}}] = 1$.

\textit{Random normalized states: the Haar ensemble ---} The Haar ensemble is the unique distribution of normalized states that is invariant under any unitary transformation $U$. It is naturally related to the Gaussian ensemble: one can sample from the Haar ensemble by sampling states $\ket{\Psi}$ from the Gaussian ensemble, then normalizing them. This induces the measure
\begin{align}
P_\text{Haar}(\Psi)d\Psi  = \int P_\text{Gauss}(\tilde{\Psi}) \delta(\ket{\Psi} - \ket*{\tilde{\Psi}}/\Vert \tilde{\Psi} \Vert) d \tilde{\Psi} ~.
\label{eq:Gaussian_ens}
\end{align}
The Haar ensemble has been widely studied in quantum information science and utilized for a myriad of applications, see e.g.~Ref.~\cite{harrow2013church}. In particular, the moments of the Haar ensemble play a key role, which we now discuss.

\subsection{Statistical $k$-th moments}
An ensemble of states contains a large amount of information, beyond that of its density matrix. In this work, we systematically study and quantify ensembles of states by their higher-order moments. 

Higher order quantities have received less attention in many-body physics because they have been inaccessible to traditional experiments, requiring a large number of controlled repetitions in order to access higher order statistical properties. Modern quantum devices address these limitations with both high repetition rates and a high degree of controllability~\cite{choi2023preparing,shaw2023benchmarking}.

In this work, we find that not only do the average properties (first moments) thermally equilibrate, so too do the higher moments of our natural ensembles of states. We first define the $k$-th moments.

The \textit{$k$-th moment} of an ensemble $\mathcal{E}$ is, for discrete ensembles
\begin{align}
     \rho^{(k)}_\mathcal{E} &\equiv \sum_{j} p_j \ketbra{\Psi_j}{\Psi_j}^{\otimes k},
 \end{align}
for continuous ensembles is
\begin{equation}
     \rho^{(k)}_\mathcal{E} \equiv \int d\Psi P(\Psi) \ketbra{\Psi}{\Psi}^{\otimes k},
\label{eq:normalized_kth_moment}
\end{equation}
and for unnormalized ensembles we have (denoting the moments with tildes)
\begin{equation}
 \tilde{\rho}^{(k)}_\mathcal{E} = \int d\tilde{\Psi} P(\tilde{\Psi}) \ketbra*{\tilde{\Psi}}{\tilde{\Psi}}^{\otimes k}.   
\label{eq:unnormalized_kth_moment}
\end{equation}
 
The $k$-th moment of an ensemble can be used to describe the expectation value of any linear operator $O^{(k)}\!:\!\mathcal{H}^{\otimes k} \rightarrow \mathcal{H}^{\otimes k}$ which acts on $k$-copies of the Hilbert space. The ensemble average of $O^{(k)}$ can be obtained from the $k$-th moment: $\mathbb{E}_\mathcal{\ket*{\Psi} \sim E}[\bra{\Psi}^{\otimes k} O^{(k)}\ket{\Psi}^{\otimes k}] = \text{tr}[O^{(k)} \rho^{(k)}_{\mathcal{E}}]$. For example, $O^
{(k)} = O^{\otimes k}$ describes the $k$-th moment of the expectation values $\mathbb{E}_\mathcal{E}[\bra{\Psi}O\ket{\Psi}^k]$. Here, we use $\mathbb{E}_\mathcal{E}[f(\Psi)]$ to denote the ensemble average $\int d\Psi P(\Psi) f(\Psi)$.

The $k$-th moment can be thought of as a generalization of the density matrix. Namely, the first moment $\rho^{(1)} = \sum_j p_j \ketbra{\Psi_j}{\Psi_j}$ is the conventional density matrix, and describes the average state of the ensemble. All higher moments $\rho^{(k)}_\mathcal{E}$ are likewise positive semi-definite operators and, for ensembles of normalized states, have unit trace.

\textit{Moments of Haar ensemble ---} The higher moments of the Gaussian and Haar ensembles are particularly simple. The moments of the Haar ensemble have analytical expression~\cite{harrow2013church}:
\begin{align}
    \rho^{(k)}_\text{Haar} &\equiv \int_\text{Haar} d\Psi \ketbra{\Psi}{\Psi}^{\otimes k} =  \frac{\sum_{\sigma \in S_k} \text{Perm}(\sigma)}{\prod_{i=0}^{k-1}(D+i)}~,
    \label{eq:Haar_kth_moment}
\end{align}
with similar expression for the moments of the Gaussian ensemble.
Here, $\text{Perm}(\sigma)$ is the operator that permutes $k$ states by the permutation $\sigma \in S_k$, i.e.~$\text{Perm}(\sigma)[\ket*{\Psi_1} \otimes \cdots \otimes \ket*{\Psi_k}] = \ket*{\Psi_{\sigma(1)}} \otimes \cdots \otimes \ket*{\Psi_{\sigma(k)}}$. This result follows from the Schur-Weyl duality~\cite{christandl2006structure}. Its simplicity enables the design of useful quantum algorithms such as classical shadow tomography~\cite{huang2020predicting,elben2023randomized}, as well as to analytical solutions of models of quantum dynamics such as random quantum circuits~\cite{fisher2023random}. We will find that our ensembles of states have $k$-th moments with similar structure as the Haar $k$-th moments, enabling phenomena such as the emergence of the Porter-Thomas distribution.

\textit{Trace distance of higher moments --- } Higher moments also enable the systematic and quantitative comparison of two ensembles of states. Specifically, we compare their $k$-th moments through their \textit{trace distance}.
\begin{align}
    \Delta^{(k)}_\text{Haar} &\equiv \frac{1}{2} \Vert \rho^{(k)}_\mathcal{E} - \rho^{(k)}_\text{Haar}\Vert_* ~,
    \label{eq:Haar_trace_distance}
\end{align}
where the trace (or nuclear) norm $\Vert \cdot \Vert_*$ of a matrix is the sum of the absolute values of its eigenvalues. In particular, the trace distance sets an upper bound on how well any $k$-copy observable $O^{(k)}$ can distinguish two ensembles of states~\cite{nielsen_chuang_2010}.

Refs.~\cite{choi2023preparing,cotler2023emergent} found that with increasing system size, projected ensembles of  infinite effective temperature states --- either obtained from quench dynamics of simple initial states or chosen from high energy eigenstates --- have closer $k$-th moments to the Haar moment (in other words, projected ensembles form better approximate $k$-designs). Analogously, we find that the $k$-th moments of the projected and (generalized) Scrooge ensembles converge in trace distance with increasing system size [Fig.~\ref{fig:projected*}(b) and Fig.~\ref{fig:generalized_scrooge}(b)].

\subsection{Ensemble entropy}
The central claim of our work is that natural many-body dynamics generates ensembles of states with maximum entropy. The entropy in question, the ensemble entropy, differs from the conventional von Neumann entropy $S(\rho) \equiv -\tr[\rho \log(\rho)]$. This is not suitable as an entropy of an ensemble of states, since it only depends on the first moment $\rho$, and there are many ensembles of states which share the same average state $\rho$. Instead, our ensemble entropy is the Shannon entropy of the distribution of states over Hilbert space. For continuous distributions, this entropy must be defined relative to a reference distribution. 
\begin{definition}
    The \textit{ensemble entropy} $\text{Ent}(\mathcal{E})$ is the negative of the Kullback-Leibler (KL) divergence between $\mathcal{E}$ and the reference measure. For normalized ensembles, the reference distribution is the Haar measure [Eq.~\eqref{eq:ens_entropy_intro}]
\begin{align}
    \text{Ent}(\mathcal{E}) &\equiv -\mathcal{D}(\mathcal{E} \Vert \text{Haar}) \label{eq:entropy_normalized_ensemble}\\
    &= -\int d\Psi~P(\Psi) \log_2\left(\frac{P(\Psi)}{P_\text{Haar}(\Psi)}\right)~, \nonumber
\end{align}
while for unnormalized ensembles, the reference distribution is the Gaussian measure
\begin{align}
    \text{Ent}(\mathcal{E}) &\equiv -\mathcal{D}(\mathcal{E} \Vert \text{Gauss}) \label{eq:entropy_unnormalized_ensemble}\\
    &= -\int d\tilde{\Psi}~P(\tilde{\Psi}) \log_2\left(\frac{P(\tilde{\Psi})}{P_\text{Gauss}(\tilde{\Psi})}\right)~. \nonumber
\end{align}
\end{definition}
The ensemble entropy is non-positive, achieving zero only when the two distributions are equal. See Refs.~\cite{brody2000information,anza2022quantum} for similar attempts to quantify the entropy of state ensembles.

\textit{Connection to von Neumann entropy ---} While the ensemble entropy may seem an unconventional notion of entropy, in fact it can be related to the von Neumann entropy of higher moments of the ensemble. For finite ensembles $\mathcal{E} = \{p_j, |\Psi_j\rangle\}$, the quantity
\begin{equation}
    \text{lim}_{k\rightarrow \infty} S(\rho^{(k)}_\mathcal{E})~.
\end{equation}
converges to the Shannon entropy $-\sum_j p_j \log_2 p_j$, which in turn converges to the ensemble entropy (plus a divergent constant), when the discrete ensemble is appropriately selected (see below). This follows from the fact that with increasing number of copies $k$, two states $\ket{\psi}^{\otimes k}, \ket{\phi}^{\otimes k}$ become increasingly distinguishable (orthogonal), since $|\braket{\phi}{\psi}|^{2k}$ quickly vanishes when $|\braket{\phi}{\psi}| \neq 1$. Therefore, for any discrete ensemble $\{p_j, \ket{\Psi_j}\}$, the quantities $p_j$ and $\ket{\Psi_j}^{\otimes k}$ converge to the eigenvalues and eigenvectors of $\rho^{(k)}$ respectively. A similar argument can be made for continuous distributions, and we obtain $ \text{Ent}(\mathcal{E}) = \text{lim}_{k\rightarrow \infty} S(\rho^{(k)}_\mathcal{E}) + C$, with a divergent constant $C$.

\textit{Ensemble entropy and information compression ---} The ensemble entropy has a simple operational meaning: it quantifies the minimum number of bits required to store the states in the ensemble in classical memory, up to a \textit{fixed precision} $\epsilon$. If storing a single such state requires $N_\epsilon$ classical bits, we find that storing a \textit{large} number of states $M$ drawn from the ensemble $\mathcal{E}$ requires $M[N_\epsilon + \text{Ent}(\mathcal{E})]$ classical bits (note that $\text{Ent}(\mathcal{E}) \leq 0$), as illustrated in Fig.~\ref{fig:info_theory}(a). 
Our maximum entropy principle indicates that the required storage for many-body states generated by time evolution is maximal.

To see this, note that storing states up to accuracy $\epsilon$ amounts to partitioning the Hilbert space in a uniform $\epsilon$-net, i.e.~into a discrete set of points which are a distance $\epsilon$ apart~\cite{cover1999elements}. To store the states in the ensemble, each state in the continuous Hilbert space is assigned to one of these discrete points. Storing these states thus amounts to storing these labels --- they may either be stored directly or they may be compressed according to how often the labels appear. In the former case, each label requires $N_{\epsilon}$ bits, if there are $2^{N_{\epsilon}}$ points in the net. In the latter, the optimal rate of compression is given by the relative entropy between the discrete probabilities on this $\epsilon$-net and the uniform distribution~\cite{cover1999elements}. In the limit of small $\epsilon$, the discrete and continuous distributions converge, and therefore the overall amount of compression achieved is equal to the ensemble entropy $\text{Ent}(\mathcal{E})$ per state, and the overall number of bits required to store $M$ states is $M[N_\epsilon + \text{Ent}(\mathcal{E})]$.

Finally, we note that this task is different from the task of storing an ensemble of states in quantum memory, i.e. compressing the states into a smaller number of \textit{qubits}. This can be done by Schumacher compression~\cite{schumacher1995quantum}, with compression rate given by the von Neumann entropy of the first moment $S(\rho_\mathcal{E})$. In contrast, the task of classical approximate storage depends on all moments of the ensemble.

\subsection{Mutual information}
\label{subsec:mutual_info}
In addition to the ensemble entropy, we study another information-theoretic property: the mutual information, which quantifies the information-theoretic correlations between two variables. We find that not only is the ensemble entropy maximized in our ensembles, the mutual information is also minimized. These two quantities are a priori unrelated, and lead to distinct operational consequences.

The \textit{mutual information} $I(X;Z)$ between random variables $X$ and $Z$ is defined as
\begin{align}
&I(X;Z) = H(X)-H(X|Z)\\
&= -\sum_x p(x) \log_2[p(x)] + \sum_z p(z) \sum_x p(x|z) \log_2[p(x|z)]~.\nonumber
\end{align}
This can be interpreted as the information gain (reduction in entropy) of the variable $X$ from measurements of the variable $Z$. In this work, we take the logarithm to base 2, which quantifies the mutual information in units of bits.

The mutual information is relevant in several information theoretic tasks, including quantifying the maximum rate of information transmission through a noisy classical channel~\cite{jozsa1994lower}, where the variables $X$ and $Z$ are here the channel input and output. Below, we focus on the significance of mutual information to a quantum communication task.

Holevo studied the use of a noiseless quantum channel to transmit classical information~\cite{holevo1998capacity}. Given an ensemble of (possibly mixed) states $\mathcal{E} = \{p_x, \rho_x\}$, Holevo considered a task where one party (Alice) transmits a classical message $\{x_1,x_2,\dots\}$ to another (Bob), by sending a sequence of states $\{\rho_{x_1},\rho_{x_2},\dots\}$. Bob decodes this message using measurements $Z$ on the received states. When the quantum states $\ket*{\psi_{x_i}}$ are not orthogonal, the symbols $x_i$ cannot be unambiguously determined: the maximum classical transmission rate is therefore given by the \textit{accessible information}
\begin{equation}
\text{Acc}(\mathcal{E}) \equiv \text{sup}_{Z} I(X;Z),    
\label{eq:accessible_information}
\end{equation}
the mutual information maximized over complete measurement bases $Z$.

Holevo provided an upper bound for the accessible information: the \textit{Holevo-$\chi$ quantity}  $\chi(\mathcal{E}) \equiv S(\rho_\mathcal{E}) - \sum_x p_x S(\rho_x)$~\cite{holevo1998capacity}. In an ensemble of pure states, the Holevo-$\chi$ quantity is equal to the von Neumann entropy $\chi(\mathcal{E}) = S(\rho_\mathcal{E})$. For a given first moment $\rho_\mathcal{E}$, the Holevo bound is saturated when the ensemble of states is formed using the eigenbasis of $\rho_\mathcal{E}$, i.e.~$\mathcal{E} = \{\lambda_j, \ket{j}\}$.

Josza, Robb and Wooters~\cite{jozsa1994lower} subsequently established a lower bound for the accessible information, dubbed the \textit{subentropy} $Q(\rho_{\mathcal{E}})$. That is, for any ensemble of pure states $\mathcal{E}$ with a given first moment $\rho_{\mathcal{E}}$, the accessible information is bounded by:
\begin{align}
Q(\rho_{\mathcal{E}}) \leq \text{Acc}(\mathcal{E}) \leq S(\rho_{\mathcal{E}})~.
\end{align}

Like the von Neumann entropy, the subentropy only depends on the first moment $\rho_\mathcal{E}$. For any ensemble $\mathcal{E}$, the subentropy is defined as the mutual information $I(X;Z)$ averaged over all complete projective measurement bases $Z$. This is equivalent to averaging over unitaries $U$ which rotate measurement bases and has expression~\cite{jozsa1994lower,datta2014properties}
\begin{equation}
    Q(\rho_\mathcal{E}) \equiv - \sum_j \left(\prod_{i\neq j} \frac{\lambda_j}{\lambda_j - \lambda_i} \right) \log_2 \lambda_j~, \label{eq:subentropy}
\end{equation}
where $\lambda_j$ are the eigenvalues of $\rho_\mathcal{E}$.  More recently, the subentropy has received attention in studies of the dynamics of information in many-body systems~\cite{ippoliti2023learnability}.

Josza, Robb and Wooters established that the Scrooge ensemble saturates this lower bound (Section~\ref{sec:scrooge},~\cite{jozsa1994lower}). Scrooge ensembles are therefore ensembles with minimum accessible information (or, the most information ``stingy" ensembles), leading to their name.

(Generalized) Scrooge ensembles describe projected ensembles, therefore our projected ensembles have minimal accessible information. That is, in natural many-body systems, it is maximally hard to predict the measurement outcome of one subsystem from measurements on the other. 

\subsection{Porter-Thomas distribution}
\label{subsec:PT_dist}
Finally, we turn to a measurable signature of our \MEP, the Porter-Thomas (PT) distribution. Originally studied in nuclear physics~\cite{porter1956fluctuations}, the PT distribution has received renewed attention due to its presence in Haar-random many-body states, approximated by the output of deep random quantum circuits~\cite{boixo2018characterizing}.

Concretely, a positive random variable $X$ follows the \textit{Porter-Thomas distribution} with mean $\mu$ [i.e.~$X\sim \text{PT}(\mu)$] if its probability distribution function is
    \begin{align}
        \text{Pr}[X=x] dx &= \exp(-x/\mu) dx/\mu
    \end{align}    

In particular, the overlap of states $\ket{\Psi}$ drawn from the Haar random distribution to any fixed state $\ket{\Phi}$ follows a PT-distribution with $\mu = 1/D$~\cite{arute2019quantum}
\begin{equation}
    \text{Pr}_{\ket{\Psi}\sim \text{Haar}}[\abs{\braket{\Psi}{\Phi}}^2=x ]dx = \exp(-D x) d D x~.
\end{equation}

Notably, this is a statement about the \textit{probability distribution of probabilities} $|\langle \Psi |\Phi \rangle|^2$. While relatively unknown in physical settings, this has received attention in the statistical literature under several names, including the \textit{fingerprint} and \textit{histogram of the histogram}~\cite{valiant2011power,valiant2017estimating}. The fingerprint is sufficient to describe the properties of a distribution that do not depend on its labels, most notably its Shannon entropy.

The probability-of-probabilities can be studied in modern quantum experiments. For example, the PT distribution expected from Haar-random states has been predicted and subsequently experimentally verified in deep random quantum circuits~\cite{boixo2018characterizing,arute2019quantum} and many-body states at infinite temperature~\cite{cotler2023emergent}. Even though the ensembles we study in this work are not the Haar ensemble, we find that the \MEP~leads to the same PT distribution, when observables are appropriately normalized. In companion work, we experimentally and theoretically study this object in detail~\cite{shaw2024universal}.

\section{Temporal ensembles}
\label{sec:temporal_ensemble}
The most natural setting of quantum dynamics is time evolution under a time-independent Hamiltonian. It is an important endeavor to understand the phenomena that can arise in this setting. Despite the fact that it is simple to state, relatively few tools are available to analyze the dynamics under an interacting many-body Hamiltonian, and one must often turn to more analytically tractable models of quantum dynamics such as random quantum circuits~\cite{fisher2023random}.

In this section, we show universal statistical properties of the temporal ensemble --- the time-trace of a state evolving under Hamiltonian dynamics. Intuitively, Hamiltonian dynamics conserves energy, and its trajectory in Hilbert space cannot be unconstrained. In fact, the Schr\"odinger equation conserves the population $|c_E|^2$ of each energy eigenstate $\ket{E}$, where $|c_E|^2 \equiv |\langle\Psi_0|E\rangle|^2$ for an initial state $\ket{\Psi_0}$. This imposes a total of $D$ (the Hilbert space dimension) constraints, and the only remaining degrees of freedom are the complex phases $\text{arg}(\langle E|\Psi(t)\rangle)$. One might conjecture that these phases are uniformly and independently random, satisfying the \MEP. 

Indeed, this expectation is correct: the temporal ensemble is (statistically) described by the random phase ensemble, the ensemble of states with fixed magnitudes and independent, uniformly random phases in the energy eigenbasis. This has been a folklore result in several works, see e.g.~Refs.~\cite{linden2009quantum,goldstein2010long,nakata2012phase}, however, it has to our knowledge not been presented with the same level of rigor as in this work. Here, we rigorously state our result with a necessary and sufficient condition --- the $k$-th no-resonance condition --- and discuss its corollaries. 

The first moments of the temporal ensemble are heavily used in the quantum thermalization literature. Seminal works in quantum chaos and thermalization introduced the \textit{diagonal ensemble} $\rho_d \equiv \sum_E |c_E|^2 \ketbra*{E}{E}$ to characterize observables at thermal equilibrium~\cite{rigol2008thermalization,polkovnikov2011nonequilibrium}.
However, this only captures the average values of observables and does not address higher moments such as their variance over time. Previous work either bounds~\cite{short2011equilibration,reimann2008foundation} or predicts the dependence of fluctuations with the total Hilbert space dimension~\cite{srednicki1999approach,nation2018offdiagonal,nation2019ergodicityprobes}, but here we provide approximate formulae for all moments. We apply our result on the temporal ensemble to a novel class of observables: global projective measurements, discovering universal behaviour in its higher moments and information-theoretic properties.

We first formalize the above intuitive expectation and present the sketch of our proof. We define the \textit{temporal ensemble} as the set of $\ket{\Psi(t)}$ explored by a \textit{fixed} initial state $\ket{\Psi_0}$ time-evolving under a Hamiltonian $H$ over a time interval $t \in [-\tau/2,\tau/2]$.
\begin{align}
    \mathcal{E}_\text{Temp.}(\tau) = \{\exp(-i H t) \ket{\Psi_0}\}_{t=-\tau/2}^{\tau/2} 
\end{align}
When $\tau\rightarrow \infty$, we equate this to the \textit{random phase ensemble}, the distribution of states with fixed magnitudes in a given basis $\abs{\braket{j}{\Psi}} = \abs{c_j} \equiv \abs{\braket{j}{\Psi_0}}$, with i.i.d.~random phases $\phi_j \sim \text{Unif}([0,2\pi))$.
\begin{gather}
    \mathcal{E}_\text{Rand.~Phase} = \Big\{\sum_{j=1}^D \abs{c_j} e^{i\phi_j} \ket{j}~\Big|~\phi_j \sim \text{Unif}([0,2\pi)) \Big\},
    \nonumber
\end{gather}
This ensemble was studied in Ref.~\cite{nakata2012phase}, which characterizes properties such as the entanglement of its typical states. Ref.~\cite{nechita2021graphical} develops a graphical calculus to perform the $k$-copy averages we discuss in this work, although we will not need its full generality here (see also Ref.~\cite{mark2022benchmarking,liu2023predicting}).

Our result, which equates the temporal ensemble (in the limit $\tau\rightarrow \infty$) with the random phase ensemble, relies on the following assumption:

\begin{definition}[$k$-th no-resonance condition]
\label{def:k-no-resonance}
A Hamiltonian $H$ satisfies the $k$-th \textit{no-resonance condition}~\cite{reimann2008foundation,linden2009quantum,kaneko2020characterizing,huang2021extensive} if for any two sets of $k$ eigenvalues $\{E_{\alpha_j}\}_{j=1}^k$ and $\{E_{\beta_j}\}_{j=1}^k$ the equation
\begin{equation}
    E_{\alpha_1} + E_{\alpha_2} + \cdots + E_{\alpha_k} = E_{\beta_1} + E_{\beta_2} + \cdots + E_{\beta_k}~ \label{eq:k_no_resonance}
\end{equation}
is true if and only if the sets of indices $(\alpha_1, \alpha_2, \dots, \alpha_k)$ and  $(\beta_1, \beta_2, \dots, \beta_k)$ are equal up to reordering. 
\end{definition}

In other words, the two sets of indices must be related by some permutation $\sigma \in S_k$, $\beta_j = \alpha_{\sigma(j)}$. 
It is easy to see that the $k$-th no-resonance condition implies the $k'$-th no-resonance conditions, for $k'<k$. The no-resonance conditions are considered generically satisfied and are used in works such as Refs.~\cite{srednicki1999approach,reimann2008foundation,linden2009quantum,short2011equilibration}: it is believed that if the condition does not hold in a given ergodic system, any small perturbation will generically break any $k$-th resonances~\cite{kaneko2020characterizing}. We note recent work~\cite{riddell2023noresonance} that studies near-violations of the $k$-th no-resonance condition, concluding that they are generically small and bounding their effects on certain quantities.

The no-resonance condition for $k=1$ simply states that the energy eigenspectrum has no degeneracies. However, this may not be satisfied in generic Hamiltonians. The presence of non-Abelian symmetries ensures that there are degenerate multiplets of states. Strictly speaking, this violates all no-resonance conditions. However, for the purposes of the temporal ensemble, we may disregard these degeneracies: the initial state $\ket{\Psi_0}$ projects the degenerate eigenspace onto a single eigenstate. Our conclusions below hold as long as the spectrum of $H$, with degeneracies removed, satisfies the $k$-th no-resonance conditions~\cite{mark2022benchmarking}. We call this the \textit{no-resonance condition modulo degeneracies}.

With the above definitions, we are ready to state our result equating the temporal and random phase ensembles.
\begin{theorem}
\label{thm:temp_ens}
Given an initial state $\ket{\Psi_0}$ and a Hamiltonian $H$, the infinite-time temporal ensemble is equal to the random phase ensemble:
\begin{equation}
\lim_{\tau \rightarrow \infty} \mathcal{E}_\text{Temp.}(\tau) = \mathcal{E}_\text{Rand.~Phase}~,    
\end{equation}
with fixed energy populations $|c_E|^2 = \abs{\braket{\Psi_0}{E}}^2$, if $H$ satisfies all $k$-th no-resonance conditions modulo degeneracies. The converse holds as long as $|\Psi_0\rangle$ has non-zero population on the eigenspace of every non-degenerate eigenvalue $E$, i.e.~for almost every $|\Psi_0\rangle$\footnote{If there is a symmetry and $|\Psi_0\rangle$ is supported only on every eigenstate in a symmetry sector, we may conclude that the $k$-th no-resonance conditions modulo degeneracies hold for the eigenvalues in that sector.}.
\end{theorem}
\noindent\textit{Proof sketch.}
Here we sketch the idea of the proof, deferring its details to Appendix~\ref{app:temp_ens_equals_random_phase_ens_proof}. 
The general idea is to evaluate the $k$-th moment $\rho_\text{temp.}^{(k)}$ of the temporal ensemble, and show that it matches the moments of the random phase ensemble. The evaluation of $\rho_\text{temp.}^{(k)}$ is possible for an infinite interval $\tau = \infty$ based on the no-resonance condition Eq.~\eqref{eq:k_no_resonance}, and the corresponding moments of the random phase ensemble can also be explicitly evaluated. We note that having matching moments does not always imply that two ensembles are equal to one another. Under certain conditions such as the (complex) Carleman condition~\cite{schmudgen2017}, this equivalence can be made. The evaluated moments satisfy the complex Carleman condition, completing our proof. The converse direction can be easily shown by noting that if the $k$-th no-resonance condition is not satisfied, the moment $\rho_\text{temp.}^{(k)}$ will have additional terms that are not present in the random phase ensemble (as long as all coefficients $c_E$ are non-zero) and hence the temporal and random phase ensembles will differ.
\hfill \qedsymbol

\textit{Ensemble entropy --- }
The ensemble entropy is always divergent in any distribution of states with fixed magnitudes, since the constrained space is a sub-dimensional manifold of the full Hilbert space. This can be intuitively understood from the following property of the KL divergence $\mathcal{D}(\mathcal{E}\Vert \text{Haar})$.
Assume we are given samples from $\mathcal{E}$ and promised that the underlying ensemble is either $\mathcal{E}$ or the Haar ensemble. The KL divergence quantifies how many samples one needs to correctly conclude that the ensemble is $\mathcal{E}$ and not the Haar ensemble~\cite{cover1999elements}. In our case, we can rule out the Haar ensemble even with a \textit{single} wavefunction $|\Psi(t)\rangle$ from the temporal ensemble. This is because the state $|\Psi(t)\rangle$ has the fixed amplitudes $|\langle E|\Psi(t)\rangle| = |c_E|$, which occurs with probability zero in the Haar ensemble. Therefore, given complete knowledge of the state $|\Psi(t)\rangle$, the constraints $\{|c_E|\}$, and the promise that the underlying ensemble is either $\mathcal{E}$ or the Haar ensemble, this discrimination task can be done infinitely quickly, and hence the KL divergence is infinitely large. While artificial in practice, this operational meaning for the KL divergence accounts for its divergence in the temporal ensemble.

However, the \MEP~still holds in a certain sense: we can separate the entropy into contributions from the phase and magnitude degrees of freedom; the contribution from the phase is maximized by independent, uniformly distributed phases (Appendix~\ref{app:temp_ens_entropy}). We find a constant, non-divergent contribution
\begin{align}
    \text{Ent}(\mathcal{E}_\text{Temp.}) &\equiv -\mathcal{D}(\mathcal{E}_\text{Temp.} \Vert \text{ Haar}) \label{eq:temporal_ensemble_entropy}\\
    &=\sum_E \log(D\abs{c_E}^2)
    +\lim_{\delta\rightarrow 0} D \left(\log(\delta/2)-1\right)~.\nonumber
\end{align}
We interpret this expression as proportional to the volume $\prod_E |c_E|$ of the constrained space. Remarkably, this expression will reappear in the ensemble entropy of the projected ensemble (Section~\ref{sec:Scrooge_entropy}). 

\subsection{Asymptotic product form}
\label{subsec:temporal_ensemble_product_form}
In this work, we shall be interested in the moments of the random phase ensemble. Here, we provide a simplified, approximate form of the $k$-th moments, which enables their further analysis.
\begin{theorem}[Asymptotic product form]\label{thm:temp_ens_product_form}
In the limit of small purity $\tr(\rho_d^2)$, the random phase ensemble has the \textit{product form}:
\begin{align}
    \rho^{(k)}_\text{Temp.} = \rho_d^{\otimes k} \sum_{\sigma \in S_k} \text{Perm}(\sigma) + O(\tr(\rho_d^2))
\end{align}
Here, $\rho_d \equiv \sum_E \abs{c_E}^2 \ketbra{E}{E}$ is the \textit{diagonal ensemble}, a common construction in the thermalization literature~\cite{rigol2008thermalization,polkovnikov2011nonequilibrium}, and its purity $\tr(\rho_d^2)$ is typically exponentially small in total system size. 
\end{theorem}

\noindent\textbf{Proof sketch:} We provide the full proof of Theorem~\ref{thm:temp_ens_product_form} in Appendix~\ref{app:temporal_ensemble_error}, and only state the key idea, which is the following: the $k$-th moment of the temporal ensemble can be written as a sum over all lists of labels of energy eigenstates $\{\alpha_j\}$ and over all permutations of $\{\alpha_j\}$ (Theorem~\ref{thm:temp_ens}). Crucially, $\{\alpha_j\}$ is allowed to have repeated elements. Only if no elements are repeated will the number of unique permutations of $\{\alpha_j\}$ be $k!$. Using the same counting also for lists with repeated elements, the expression for the moments considerably simplifies.
\begin{align}
    \rho^{(k)}_\text{Temp.} &\approx \sum_{\sigma \in S_k}\sum_{\{\alpha_j\}} \bigotimes_{j=1}^k c_{\alpha_j}  c^*_{\sigma(\alpha_j)}\ketbra{E_{\alpha_j}}{E_{\sigma(\alpha_j)}} \label{eq:temp_ens_thm_eq2}\\
    &= \rho_d^{\otimes k} \sum_{\sigma \in S_k} \text{Perm}(\sigma)
\end{align}
Eq.~\eqref{eq:temp_ens_thm_eq2} is not exact because it over-counts permutations when $\{\alpha_j\}$ contains repeated elements: there are additional correction terms to compensate for this over-counting. Careful consideration gives that the trace norm of all correction terms can be bounded by $O(\tr(\rho_d^2))$, with a $k$-dependent prefactor (Appendix~\ref{app:temporal_ensemble_error}). \hfill \qedsymbol

Theorem~\ref{thm:temp_ens_product_form} explicitly demonstrates how the higher order moments of the temporal ensemble depend on the first moment $\rho_d$. Specifically, the moments of the temporal ensemble have a structural resemblance to the moments of the Haar ensemble [Eq.~\eqref{eq:Haar_kth_moment}], up to a correction which is typically exponentially small, formalizing our intuition of maximally entropic ensembles as ``distorted" Haar ensembles. This asymptotic behavior will also be common to projected ensembles.

\subsection{Finite-time temporal ensembles}
\label{subsec:finite_temp_temp_ens}
Theorem~\ref{thm:temp_ens} studies the properties of temporal ensembles over infinite time intervals. Here, we investigate how temporal ensembles over finite time intervals converge to their infinite-interval values. We analytically study its convergence rate in random-matrix models and present numerical evidence that the above convergence rate also holds for realistic Hamiltonians. Specifically, we study the convergence of the finite time $k$-th moment $\rho^{(k)}_\text{Temp.}(\tau)$ to its infinite-time limit $\rho^{(k)}_\text{Temp.} \equiv \rho^{(k)}_\text{Temp.}(\tau=\infty)$. This has expression:
\begin{align}
    &\rho^{(k)}_\text{Temp.}(\tau) \equiv \frac{1}{\tau}\int_{-\tau/2}^{\tau/2} dt \ketbra{\Psi(t)}{\Psi(t)}^{\otimes k} \label{eq:finite_time_ensemble}\\
    &=\! \!\sum_{\{\alpha_j\},\{\beta_j\}} \!\!\text{sinc}\Big(\sum_j (E_{\alpha_j}-E_{\beta_j})\frac{\tau}{2}\Big) \Bigg[\bigotimes_j c_{\alpha_j} c_{\beta_j}^* \ketbra*{E_{\alpha_j}}{E_{\beta_j}}\Bigg]\nonumber. 
\end{align}

We consider the convergence both in terms of the trace norm and Frobenius norm, with the latter being analytically easier to study. For the first moment, i.e.~$k=1$, we numerically observe that the distance in Frobenius and trace norms both decay as $\tau^{-1}$ (see App.~\ref{app:finite_time_ensembles}). This holds for a generic chaotic model, the mixed field Ising model (the main model we study in this work, Sec.~\ref{subsec:numerical_details}) as well as for Hamiltonians sampled at random from the Gausian Unitary Ensemble (GUE).  We attribute this asymptotic $1/\tau$ scaling to the $1/x$ asymptotic behavior of the $\text{sinc}(x)$ function in Eq.~\eqref{eq:finite_time_ensemble}. In Appendix~\ref{app:finite_time_ensembles}, we support these numerical findings by analytical calculations in random matrix theory, showing that the average of the squared Frobenius norm decays as $\tau^{-2}$. This suggests that the average Frobenius norm decays as $\tau^{-1}$, consistent with our numerical observations. 

Interestingly, for higher order moments $k=2,3$, we observe different behavior. For the mixed field Ising model as well as individual samples from the GUE, the Frobenius distance decays at intermediate times as $\tau^{-1/2}$ before a crossover time, after which it decays as $\tau^{-1}$ [Fig.~\ref{fig:temporal_ensemble}(b)]. This crossover happens in times which are exponentially long in system size, which we attribute to the presence of a minimum distance of gaps $E_i-E_j+E_k-E_l$ (for $k=2$) for each individual Hamiltonian. While such minimum distance is exponentially small in system size, it has a finite, nonzero value $\Delta_{\min}$ for finite-sized systems. Only at times much longer than the inverse of this minimum gap does the asymptotic $1/(\Delta_{\min}\tau)$ behavior become clear.

Notably, we find that for the first moment, $k=1$, the GUE ensemble mean closely resembles the behavior of individual instances and decays as $\tau^{-1}$. In contrast, for higher moments,  $k=2$, the behavior of individual samples and the sample mean differ. With an increasing number of samples, the crossover is less pronounced. We relate these findings to statistical properties of the eigenvalues of GUE Hamiltonians. It is well known that the eigenvalues exhibit level repulsion, i.e.\ the probability to find two distinct eigenvalues close to each other $E_i-E_j \approx 0$ is vanishing. In contrast, as suggested by Ref.~\cite{riddell2023noresonance} and the analytical and numerical results presented in Appendix~\ref{app:finite_time_ensembles}, we find that the gaps $E_i-E_j+E_k-E_l$ of eigenvalues exhibit only weak repulsion, leading to the crossover timescale growing with the size of the sample.

\begin{figure*}
    \centering
    \includegraphics[width = 0.9\textwidth]{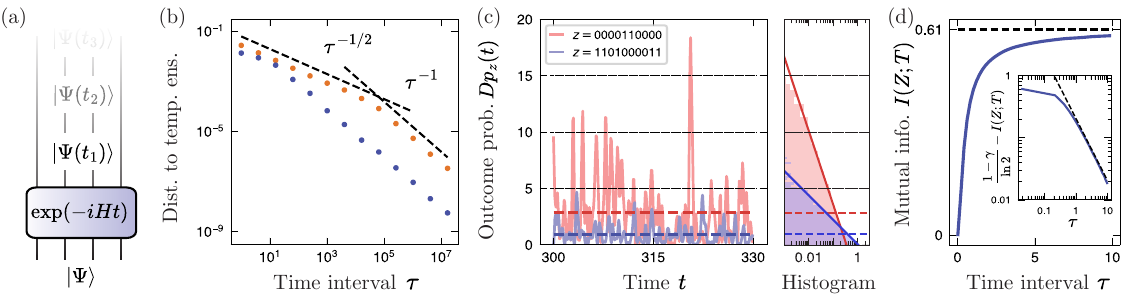}
    \caption{Temporal ensembles from natural dynamics. (a) The temporal ensemble is the collection of states formed by time-evolving a given state $\ket{\Psi}$ by a Hamiltonian $H$, for various times $t$. (b) 
    Temporal ensembles formed over a finite interval of width $\tau$ approach the infinite-time ensemble as $\tau$ increases: the distance (here, the Frobenius norm) between their $k$-th moment decreases as $1/\tau^{1/2}$, before decreasing as $\tau^{-1}$. Data shown here for $k=2$, for a state evolving under the mixed field Ising model (MFIM) in red, and a Gaussian Unitary Ensemble (GUE)-random Hamiltonian in blue. (c) Signatures of the temporal ensemble can be measured through the fluctuations of quantities such as the outcome probability $p(z,t) \equiv \abs{\braket{z}{\Psi(t)}}^2$ (normalized by a factor of $D$). Their histogram of values over time reveals a universal, Porter-Thomas (exponential) distribution. Histograms and time traces for two values of $z$ shown in red and blue. (d) The mutual information $I(Z;T)$ as a function of the time interval $\tau$ of the ensemble. $I(Z;T)$ quickly saturates to the theoretical value of $(1-\gamma)/\ln 2$, agreeing with the predicted $(1-\gamma - \sqrt{\pi}/(2\sigma_H \tau))/\ln 2$ [inset, dashed, Eq.~\eqref{eq:finite_time_mutual_information}]. In (c,d), data is obtained from quench evolution in the MFIM, for the initial state in Eq.~\eqref{eq:initial_state} with $\theta=0.6$, and the units of time are determined by the $O(1)$ Hamiltonian parameters of Eq.~\eqref{eq:MFIM}.}
    \label{fig:temporal_ensemble}
\end{figure*}

\subsection{Porter-Thomas distribution as measurable signature of temporal ensembles}
\label{subsec:temp_ens_PT_dist}
An immediate consequence of Theorem~\ref{thm:temp_ens_product_form} is the emergence of the Porter-Thomas distribution (Sec.~\ref{subsec:PT_dist}) in the probabilities of measuring a given outcome $z$, as a function of time. This result was originally presented in Ref.~\cite{mark2022benchmarking}. For clarity, we denote this distribution as $p_z(t) \equiv p(z,t) \equiv \braket{z}{\Psi(t)}\!\braket{\Psi(t)}{z}$~\footnote{The subscript $z$ indicates that the measurement outcome $z$ is being fixed and the distribution $p_z(t)$ is taken over time $t$.}, which can be experimentally estimated by repeated measurements of the state $\ket{\Psi(t)}$. 

Unlike local observables, whose expectation values quickly equilibrate in time--- $\langle O(t)\rangle$ has exponentially small relative fluctuations about its equilibrium value $\mathbb{E}_t[\langle O(t)\rangle]$\cite{reimann2008foundation}--- global observables such as the projector $|z\rangle\langle z|$ never equilibrate. Theorem~\ref{thm:temp_ens_PT} indicates that they have $O(1)$ relative fluctuations which follow a universal, exponential distribution. We utilized this result in Ref.~\cite{mark2022benchmarking} to propose a protocol for many-body benchmarking in analog quantum simulators. Here, we will use the PT distribution as a signature of our maximum entropy principle; it will be a common signature among the temporal ensemble, projected ensemble and beyond.

\begin{theorem}\label{thm:temp_ens_PT}
Given an initial state $\ket{\Psi_0}$, a Hamiltonian $H$ which satisfies the $k$-th no-resonance condition, and a fixed state $\ket{z}$, the outcome distribution $p_z(t)$ follows an approximate Porter-Thomas (or, exponential) distribution over time, with mean $p_\text{avg}(z)\equiv \mathbb{E}_t[p_z(t)]$. Specifically, the $k$-th moments satisfy:
\begin{align}
    \mathbb{E}_t[p_z(t)^k]&= p_\text{avg}(z)^k\left[ k! + O(D_\beta^{-1})\right]
\end{align}
where $D_\beta^{-1} \equiv \sum_{z,E}\abs{\braket{z}{E}}^4 |c_E|^4/ p_\text{avg}(z)$ is the inverse \textit{effective dimension}, which captures the effective size of the Hilbert space explored during quench evolution that is accessible by measurements in $\{\ket{z}\}$. Our notation $\beta$ associates $D_\beta$ with the temperature of the initial state, which to a large extent controls the size of the Hilbert space explored.
\end{theorem}

\begin{proof}
Its $k$-th moment can be easily obtained from the $k$-th moment of the temporal ensemble, which gives the desired
$\mathbb{E}_t[p_z(t)^k] \equiv \bra{z}^{\otimes k} \rho^{(k)}_\text{Temp.} \ket{z}^{\otimes k} \approx k!~ p_\text{avg}(z)^k$. This relation follows from Theorem~\ref{thm:temp_ens_product_form}. We conclude that $p_z(t)$ follows a Porter-Thomas distribution\footnote{More accurately, we bound the distance between the $k$-th moments of $p_z(t)$ and the Porter-Thomas distribution. We do not currently have a means to directly bound the distance between distributions, such as the total variation distance.}, with a mean value equal to $p_\text{avg}(z) \equiv \mathbb{E}_t[p_z(t)] = \langle z| \rho_d |z\rangle$. This is a simple illustration of how the first moment $p_\text{avg}(z)$ uniquely determines all moments of the distribution $p_z(t)$. This relation is only approximate because of the correction terms in Theorem~\ref{thm:temp_ens_product_form}. We are able to (loosely) bound the effects of these correction terms on the $k$-th moments of $p(z,t)$ in terms of the parameter $D_\beta^{-1}$. This is not a straightforward application of the bound provided in Theorem~\ref{thm:temp_ens_product_form}, which is too weak for our purposes. We refer the reader to Ref.~\cite{mark2022benchmarking} for a detailed proof, or to Appendix~\ref{app:proof_temp_ens_of_proj_st} for a very similar proof.
\end{proof}

\subsection{Numerical model}
\label{subsec:numerical_details}
We verify our theory predictions via explicit exact numerical simulations of interacting quantum many-body systems.
As an illustration, we work with a paradigmatic model of many-body quantum chaos, the one-dimensional mixed field Ising model (MFIM) with open boundary conditions.
\begin{equation}
    H_\text{MFIM} = \sum_{j=1}^N h_x X_j + h_y Y_j + \sum_{j=1}^{N-1} J X_j X_{j+1}~,
    \label{eq:MFIM}
\end{equation}
where $X_j$ and $Y_j$ are Pauli matrices on site $j$. We use parameter values established in Ref.~\cite{kim2014testing}: $(h_x,h_y,J) = (0.890,0.9045,1)$.
Unless otherwise stated, we use this model in this and future sections for our numerical simulations.

In both the projected and temporal ensemble, we study the time-evolution of the initial state
\begin{equation}
    \ket*{\Psi_0} = \left[\text{exp}(i \theta X/2) \ket{0}\right]^{\otimes N}~, \label{eq:initial_state}
\end{equation}
where the parameter $\theta$ is used to tune the state from infinite temperature ($\theta=0$) to positive or negative temperature. Unless otherwise stated, we present data for $\theta = 0.6$, which has energy density $E/N \approx 0.51$, far from the center of the spectrum.

Here, we illustrate the content of Theorem~\ref{thm:temp_ens_PT} in Fig.~\ref{fig:temporal_ensemble}(c), where we plot the probabilities $p_z(t)$ for two choices of $z$, which exhibit large fluctuations over time, about different average values $p_\text{avg}(z)$. The histograms of $\{p_z(t)\}$ show good agreement with the PT distributions $p_z(t)\sim \text{PT}[p_\text{avg}(z)]$ even over a relatively short time interval, confirming our analysis.

\subsection{Ergodicity}
\label{subsec:ergodicity}
Theorem~\ref{thm:temp_ens_PT} is a statement about the distribution of $p(z,t)$ over time (denoted $p_z(t)$ to make explicit that the variable $z$ is held constant). We can make an analogous statement about a distribution over $z$, holding $t$ constant. At a fixed, late time $t$, the quantities $\tilde{p}_t(z) \equiv p(z,t)/p_\text{avg}(z)$ are approximately distributed according to the PT distribution with mean 1.
\begin{equation}
    \tilde{p}_t(z) \sim \text{PT}(1)~,
\end{equation}
This was originally derived in Ref.~\cite{mark2022benchmarking} by showing that the weighted $k$-th moment $m^{(k)}(t) \equiv \sum_z p_\text{avg}(z) \tilde{p}_t(z)^k$ for a typical, late time $t$ satisfies:
\begin{equation}
    m^{(k)}(t) = k! + O(D_\beta^{-1/2})~.
\label{eq:ergodicity}
\end{equation}

We interpret this equivalence of distributions over time and over outcomes $z$ as a kind of ``ergodicity," in the sense of the \textit{ergodic theorem} in the study of classical dynamical systems~\cite{brin2002introduction}, which states that in an ergodic system, the average over time is equal to the ``spatial" average over configurations. Such an equivalence will be essential to our characterization of projected ensembles in Section~\ref{sec:general_scrooge}.

\subsection{Mutual information: building a (bad) ``clock''}
\label{subsec:temp_mutual_info}
The temporal ensemble has unique information theoretic properties. In particular, we find near-universal behavior in the \textit{mutual information} between the random variables $Z$ and $T$ which respectively represent the measurement outcomes $z$ and evolution times $t$.

To illustrate the operational meaning of the mutual information, let us consider the following task. We are given several copies of the time-evolved state $\ket{\Psi(t)}$. 
We have knowledge of the initial state of the system as well as its Hamiltonian, but not the duration of the time evolution $t$, which is assumed to be uniformly distributed on some large interval $t \in [t_0, t_0+\tau]$.
Given multiple single-copy measurements performed in an optimal basis, how much information about $t$ can we learn?

Before performing detailed analysis, we can make a prediction based on our maximum entropy principle: since a temporal ensemble is maximally entropic, we expect that the amount of information must be small no matter which basis the measurements are performed.
Indeed, our results imply that the (near) optimality is achieved by any generic measurement basis including the conventional configuration basis $\ketbra{z}{z}$. Furthermore the information gain per measurement takes the universal value $(1-\gamma)/\ln 2 \approx 0.6099$, independent of the Hamiltonian, initial state, measurement basis, and system size.
 
In order to formally define the mutual information $I(Z;T)$, we must make the following definitions. We let time be a uniformly distributed variable over a interval of (long but arbitrary) length $\tau$. We then define a joint probability distribution over measurement outcomes $z$ and times $t$, by $P(z,t) = p_z(t)/\tau$, which is normalized: $\int_{-\tau/2}^{\tau/2} dt \sum_z P(z,t) = 1$. We shall also need the marginal distributions $P(t) \equiv \sum_z P(z,t) = 1/\tau$, $P(z) \equiv \int dt P(z,t) = p_\text{avg}(z)$ and the conditional distribution $P(t|z) = P(t,z)/P(z) = \tilde{p}_z(t)/ \tau$. 

Then, the mutual information can be computed as 
\begin{align}
    &I(Z;T) = H(T)-H(T|Z)~,\\
    &= - \int dt P(t) \log_2(P(t))  \nonumber\\
    &~~~~+ \sum_z P(z) \int_{-\tau/2}^{\tau/2} dt P(t|z) \log_2 P(t|z)~,\\
    &\approx \sum_z p_\text{avg}(z) \int_0^\infty d\tilde{p}~\text{Pr}(\tilde{p}) \tilde{p} \log_2 \tilde{p}  + O(D_\beta^{-1})\label{eq:temp_MI_line}\\
    &\approx (1-\gamma)/\ln 2 ~ \label{eq:temp_MI}.
\end{align}
where in Eq.~\eqref{eq:temp_MI_line} we have used Theorem~\ref{thm:temp_ens_PT} to replace the integral over time with an integral over $\tilde{p}$, which follows a PT distribution, i.e. with probability density $\text{Pr}(\tilde{p})d\tilde{p} = \exp(-\tilde{p})d\tilde{p}$.

It is remarkable that this universal value does not depend on system sizes, and, in particular, does not vanish with any parameters.
This implies that, despite the fact that ergodic dynamics hides information about time $t$ as strongly as possible, a finite amount (0.6099 bits) of information still cannot be concealed, in the large $\tau$ limit.
In other words, our result implies that a certain level of temporal fluctuations in observables is inevitable in the ergodic dynamics of generic pure states, putting the lower bound on the ability of unitary dynamics to hide temporal information. 
We also note that estimating the evolution duration $t$ is equivalent to determining an overall strength scaling factor of a many-body Hamiltonian by evolving to a fixed time $t$, and hence this ``clock" may equivalently be regarded as a sensor for the overall Hamiltonian strength.

\subsubsection{Finite-time mutual information}
\label{subsubsec:finite_temp_mutual_info}

We note that the mutual information discussed above should be distinguished from more conventional metrics for clocks such as sensitivity. The sensitivity at very small intervals $\Delta t$ (i.e. in the limit of many measurements) depends instead on the energy uncertainty $\sigma_H \equiv \left(\langle \psi_0|H^2|\psi_0\rangle -\langle \psi_0|H|\psi_0\rangle^2\right)^{1/2}$, which is not universal. The mutual information instead describes how quickly a large interval $\tau$ can be refined by measurements.

Our universal result above is valid in the limit of large time intervals $\tau$, over which $p(z,t)$ follows the approximate PT distribution. At shorter intervals (but still at sufficiently late times $t$), the mutual information $I(Z;T)$ takes on a smaller value, due to $p(z,t)$ having a non-zero correlation time. This is illustrated in Fig.~\ref{fig:temporal_ensemble}(d). In fact, $I(Z;T)$ approaches its late-time value as a power law:
\begin{equation}
    I(Z;T) \approx \left(1-\gamma - \frac{\sqrt{\pi}}{2\sigma_H \tau}\right)/\ln 2~,
    \label{eq:finite_time_mutual_information}
\end{equation}
where $\sigma_H = (\bra{\Psi_0}H^2\ket{\Psi_0}-\bra{\Psi_0}H\ket{\Psi_0}^2)^{1/2}$ is the uncertainty in energy of the initial state. This is illustrated in the inset of Fig.~\ref{fig:temporal_ensemble}(d).

We present the full derivation in Appendix~\ref{app:finite_temporal_mutual_information}. In brief, we obtain this result by estimating the autocorrelation time of $p(z,t)$, using approximations of the many-body spectrum. We find that the correlation time is inversely proportional to the energy uncertainty of the state $\sigma_H$. 
If we interrogate $p(z,t)$ when t is an integer multiple of the correlation time, we get an effective number $\sigma_H \tau$ of points in times for which $p(z,t)$ is distinct.  This gives a mutual information that approaches the value $1-\gamma$ at the rate $(\sigma_H\tau)^{-1}$. The coefficient $\sqrt{\pi}/2$ follows from a detailed calculation. This result complements Section~\ref{subsec:finite_temp_temp_ens}, providing a convergence timescale of finite-time temporal ensembles from the perspective of the mutual information. We leave to future work the analysis of the ultimate sensitivity at very small intervals $\tau \ll \frac{\sqrt{\pi}}{2\sigma_H (1-\gamma )}$, at which point Eq.~\eqref{eq:finite_time_mutual_information} breaks down.

Our results in this section describe universal, statistical features of the trajectory of a many-body state under chaotic many-body dynamics. They confirm an intuitive, folklore heuristic that chaotic many-body dynamics resembles random unitary dynamics. Here, we find signatures of Haar random states such as the PT distribution in the time-evolved many-body states. We attribute this quasi-randomness to our maximum-entropy principle. This translates to our finding that these temporal trajectories can be used for certain protocols such as benchmarking and sensing.
Furthermore, their performance in these protocols is universal and independent of details such as the Hamiltonian, initial state, and measurement basis.

\section{Projected Ensembles}
\label{sec:scrooge}
\begin{figure*}
    \centering
    \includegraphics[width=\textwidth]{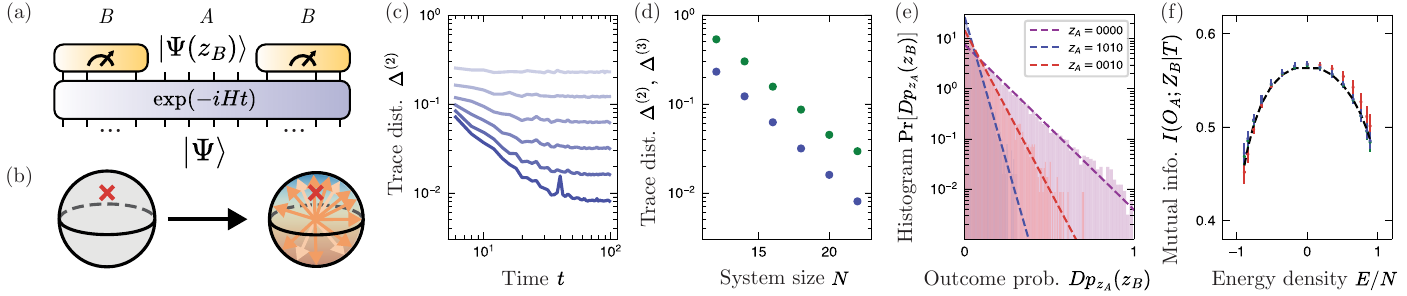}
    \caption{Projected ensembles from natural dynamics. (a) The projected ensemble is formed by performing projective measurements on a large subsystem $B$ of a many-body state, here obtained from time-evolution of an initial state. (b) In this setting, the projected ensemble is described by the \textit{Scrooge ensemble}, which is a unique probability distribution of states (color gradient) for each value of the first moment (red cross, in our setting equal to the reduced density matrix). (c) At constant energy density $E/N \approx 0.5$, the trace distance between the second moments of the projected and Scrooge ensembles decreases with quench evolution time before saturating to an equilibrium value which depends on system size.  We time-evolve a product state with the mixed-field Ising model (MFIM) on a one-dimensional open chain of length $N$ [Eq.~\eqref{eq:MFIM}], from $N=12$ (lighter) to $N=22$ (darker) and fix the subsystem $A$ to be four sites near the center of the chain. (d) The equilibrium trace distance decreases exponentially with system size for both second (blue) and third moment (green). (e) A signature of the Scrooge ensemble is the emergence of the Porter-Thomas distribution in measurement outcome distributions; here we plot a histogram of the normalized outcome distributions $D p_{z_A}(z_B)$, with fixed $z_A$ and different values of $z_B$. Histograms for three different values of $z_A$ (blue, red, purple) follow PT distributions with different means $D \mathbb{E}_{z_B}[p_{z_A}(z_B)]$. (f) The mutual information $I(O_A;Z_B|T)$ (conditioned at fixed times) is an information-theoretic signature of the Scrooge ensemble: its value is independent of the basis $O_A$, and is equal to the subentropy $Q[\rho_A]$ [Eq.~\eqref{eq:subentropy}] (black dashed). Here we plot $I(O_A;Z_B|T)$ for three different observables $O_A = \{X^{\otimes 4}, Y^{\otimes 4},Z^{\otimes 4}\}$ (red, green and blue respectively), against the energy density $E/N$ of the state, controlled by the polarization angle $\theta$ of the initial state [Eq.~\eqref{eq:initial_state}].
    }
    \label{fig:projected*}
\end{figure*}

Next, we turn to the projected ensemble, which is the second setting of our maximum entropy principle. 
When a many-body state is bipartitioned into a small system $A$ and its complement $B$, the large subsystem $B$ acts like a ``thermal bath" to the smaller subsystem $A$. This observation is instrumental to the emergence of thermal equilibration in a closed quantum system. Nevertheless, the extent to which subsystem $B$ truly acts as a bath is not fully characterized, owing to the complexity of the overall quantum many-body state. Projected ensembles offer a lens into the nature of correlations between subsystem and bath. Without the information of the bath, the subsystem is described by the reduced density matrix $\rho_A = \tr_B(\ketbra{\Psi}{\Psi})$. In order to go further, it is necessary to obtain information about the bath. In projected ensembles, this is done by projecting the bath onto definite states $|z_B\rangle$, typically by performing measurements in a product basis. This in turn projects $\ket{\Psi}$ into a pure state $\ket*{\Psi(z_B)} \in \mathcal{H}_A$. Taken over all measurement outcomes, a single wavefunction $\ket{\Psi}$ generates a large ensemble of states $\{\ket{\Psi(z_B)}\}$. By studying the statistical properties of this ensemble of states, as well as any correlations between $\ket{\Psi(z_B)}$ and $z_B$, the projected ensemble construction provides a new way to study system-bath correlations. 

Our maximum entropy principle states the following: there are correlations between $\ket{\Psi(z_B)}$ and $z_B$ which are ``thermal" in the sense that they arise from (and \emph{only} from) energy conservation. Furthermore, there are random fluctuations over these thermal correlations. In Section~\ref{sec:general_scrooge}, we find that these fluctuations are universal and follow a maximum entropy principle. 

In this section, we study a special case. In certain measurement bases, the correlations between $\ket*{\Psi(z_B)}$ and $z_B$ vanish. Nevertheless, effects of energy conservation remain: a mixture of all projected states equals $\rho_A$, which is, up to subleading corrections, a thermal Gibbs state with temperature set by the bath~\cite{srednicki1999approach}. Therefore, while the projected states $\ket*{\Psi(z_B)}$ are effectively identical independent random samples from a distribution of states, the underlying distribution is not Haar-random, but has first moment $\rho_A$. In this section, we find that this distribution has maximum entropy under this first moment constraint, a distribution that has been previously studied under the name ``Scrooge ensemble"~\cite{jozsa1994lower,goldstein2006distribution,goldstein2016universal}.

Explicitly, the projected ensemble is defined as the set of states obtained by projectively measuring a pure state $\ket{\Psi}$ on a large subsystem $B$~[Fig.~\ref{fig:projected*}(a)]. Specifically, the subsystem $B$ is measured, giving a random outcome $z_B$, which occurs with probability $p(z_B)$:
\begin{equation}
p(z_B) \equiv \bra{\Psi}\left(\mathbb{I}_A \otimes \ketbra{z_B}{z_B}   \right)\ket{\Psi}~.
\end{equation}
This measurement outcome projects the initial state $\ket{\Psi}$ onto a state $\ket{\Psi(z_B)}$ in the Hilbert space of the remaining subsystem, denoted $A$:
\begin{equation}
    \ket{\Psi(z_B)} \equiv \left(\mathbb{I}_A \otimes \bra{z_B}\right)\ket{\Psi}/\sqrt{p(z_B)} \in \mathcal{H}_A
\end{equation}
The projected ensemble is the set of states
\begin{equation}
 \mathcal{E}_\text{Proj.} = \{p(z_B), \ket{\Psi(z_B)}\}   
\end{equation}

While this object was first introduced in different contexts, e.g.~in Ref.~\cite{goldstein2016universal}, the projected ensemble was later studied by some of the authors in Refs.~\cite{choi2023preparing,cotler2023emergent} as a means of quantifying the notion of randomness in a \textit{single} many-body state (as opposed to ensembles of them), by studying its projected ensemble. 

\subsection{Results at infinite temperature: approximate $k$-designs}
In Ref.~\cite{cotler2023emergent}, it was found that a wide class of many-body states induce projected ensembles that are statistically close to the Haar ensemble. This closeness was measured by the trace distance $\Delta^{(k)}_\text{Haar}$ between the $k$-th moments of the projected ensemble and the Haar ensemble. 

Specifically, it was proven that the projected ensembles are approximate $k$-designs when the global states are typical states drawn from the Haar ensemble or from a state $k'$-design (with $k'> k$). This was also numerically demonstrated in natural many-body states, obtained from either (a) time-evolution of an initial state or (b) the eigenstates of a chaotic many-body Hamiltonian, close to infinite temperature, defined as states $\ket{\Psi_0}$ that satisfy $\bra{\Psi_0}H\ket{\Psi_0} = \tr(H)/D$. Subsequent work has rigorously established the emergence of the projected $k$-design in specific settings, such as in dual-unitary models~\cite{ho2021exact,ippoliti2022dynamical,ippoliti2022solvable,claeys2022emergent,shrotriya2023nonlocality}, free fermion models~\cite{lucas2023generalized}, or under assumptions of the reduced density matrix~\cite{wilming2022hightemperature}. Further work has studied the effects of symmetry~\cite{bhore2023deep,varikuti2024unraveling}, the effect of quantum magic in the global state~\cite{vairogs2024extracting} on the projected ensemble, and deep thermalization timescales~\cite{ippoliti2022dynamical,ippoliti2022solvable,chan2024projected}.

In the above studies, it was important that the initial state $\ket{\Psi_0}$ be at infinite temperature. This is because the first moment of the projected ensemble is the reduced density matrix $\rho_A \equiv \tr_B[\ketbra{\Psi}{\Psi}]$. For thermal states, $\rho_A$ is close to a Gibbs state~\cite{srednicki1999approach}; only at infinite temperature is this equal to the first moment of the Haar ensemble: the maximally mixed state $\mathbb{I}/D_A$. Therefore, away from infinite temperature, the projected ensembles of such states cannot equal the Haar ensemble. It has remained an open question as to which ensemble describes the finite-temperature projected ensemble.

\subsection{Projected ensembles at finite temperature: Scrooge ensemble}
Our main result of this section is that under appropriate conditions, the projected ensemble of a finite temperature state is described, by the \textit{Scrooge ensemble}~\cite{jozsa1994lower,goldstein2006distribution,goldstein2016universal}. Introduced in Ref.~\cite{jozsa1994lower}, any density matrix $\rho$ has a corresponding Scrooge ensemble, which we denote $\text{Scrooge}[\rho]$~[Figure~\ref{fig:projected*}(b)].
This is most easily understood as a ``$\rho$-distortion" of the Haar ensemble, defined by the probability distribution. The state $\ket{\Psi}$ has a higher probability when it has a larger overlap with the principal axes of $\rho$. 
We can sample from the Scrooge ensemble by ``distorting" wavefunctions sampled from the Haar ensemble, as described in Section~\ref{sec:summary_of_results}. We can also describe the Scrooge ensemble in terms of its distribution function~\cite{goldstein2006distribution}.
\begin{align}
    &P_\text{Scr}(\Psi) d\Psi \nonumber\\
    &= \frac{D!}{2~\pi^D \det \rho} \bra{\Psi}\rho^{-1} \ket{\Psi}^{-D-1} P_\text{Haar}(\Psi) d\Psi \label{eq:Scrooge}\\
    &=\frac{D!}{2~\pi^D \prod_{m=1}^D \lambda_m} \left(\sum_{m=1}^D \frac{\abs{\Psi_m}^2}{\lambda_m}\right)^{-D-1}~  P_\text{Haar}(\Psi) d\Psi~, \nonumber
\end{align}
where $\{\lambda_m, \ket{m}\}$ is the eigensystem of $\rho$, and $\Psi_m \equiv \braket{m}{\Psi}$.
In the context of the projected ensemble, the first moment $\rho$ is the reduced density matrix $\rho_A$ of $\ket{\Psi}$. 

In the ideal limit in which the measurement basis is a \textit{Haar-random basis} on $B$, 
Ref.~\cite{goldstein2016universal} proves that the resulting projected ensemble is the Scrooge ensemble. 

However, this limit of Haar-random measurements is not realistic to our setting, in which the measurement basis is typically a complete basis of unentangled product states. This may have qualitatively different behavior from the highly entangled Haar-random basis states. Furthermore, in realistic systems, the measurement basis may be strongly correlated with the Hamiltonian realized in the system.

It will turn out that the first condition does not affect the behavior of the projected ensemble, but the second condition does. 
While we provide a precise definition in Section~\ref{sec:general_scrooge}, ``correlations between measurement basis and the Hamiltonian" may be taken to be the distribution of expectation values $\bra{z_B}H_B\ket{z_B}$, where $H_B$ are the terms of the Hamiltonian $H$ restricted onto the large subsystem $B$. We find that the projected states $\ket*{\Psi(z_B)}$ are correlated with the energy $\bra{z_B}H_B\ket{z_B}$. If there is a large spread of values $\bra{z_B}H_B\ket{z_B}$, the projected states are distinguishable. Conversely, if $\bra{z_B}H_B\ket{z_B}$ is independent of $z_B$, the projected states are independent of their label $z_B$, and the projected ensemble can be regarded as a collection of identical samples from the same distribution. In brief, we claim:
\begin{claim}
\label{claim:Scrooge}
    The projected ensemble is statistically described by the Scrooge ensemble, if the measurement basis $\{\ket{z_B}\}$ is energy non-revealing.
\end{claim}

When such correlations are not negligible, the projected ensemble is described by a \textit{generalized Scrooge ensemble}, discussed in Section~\ref{sec:general_scrooge}. There, we provide a proof of our general claim for the projected ensembles of time-evolved states, based on our results about the temporal ensemble in Section~\ref{sec:temporal_ensemble}.

Our conventional understanding of many-body dynamics is that on sufficiently small subsystems, a many-body state will rapidly reach thermal equilibrium, i.e. its reduced density matrix will be close to that of a Gibbs state at an appropriate temperature~\cite{srednicki1999approach}. As with the temporal ensemble, our results imply that not only is the first moment of the projected ensemble (the reduced density matrix) a universal function of the temperature, so are all higher statistical moments. 

\subsection{Moments of the Scrooge ensemble}
\label{subsec:scrooge_moment}
We directly test our claim by comparing the $k$-th moments of the projected ensemble to those of the Scrooge ensemble. In order to do so, here we provide explicit expressions for the $k$-th moment of the Scrooge ensemble, which have previously not been available in the literature. 

Our expression for the $k$-th moment is in terms of its matrix elements in ($k$-copies of) the eigenbasis $\{\lambda_m, \ket{m}\}$ of $\rho$. We find that the nonzero elements of the Scrooge $k$-th moment $\rho^{(k)}_{\text{Scr.}}$ are those which permute the $k$ copies. For a basis element $\ket{\pmb{m}}\equiv \ket{m_{1},\dots,m_{k}}$, $\langle \pmb{m} |\rho^{(k)}_{\text{Scr.}} | \pmb{m'} \rangle$ is non-zero only  $\ket{\pmb{m'}}$ is a permutation of $\ket{\pmb{m}}$, i.e.~$\ket{\pmb{m'}} = \ket{\sigma(\pmb{m})},~\sigma \in S_k$. Furthermore, the matrix elements only depend on $\pmb{m}$ and are independent of the permutation $\sigma$. 
\begin{equation}
\rho^{(k)}_{\text{Scr.}} = \sum_{\pmb{m}} \sum_{\pmb{m'} = \sigma(\pmb{m})}    \rho^{(k)}_{\text{Scr.},\pmb{m}} \ketbra{\pmb{m}}{\pmb{m'}}~.
\end{equation}
This structure is similar to the $k$-th moments of the Haar ensemble [Eq.~\eqref{eq:Haar_kth_moment}]. However, unlike in the  Haar ensemble, the coefficients $\rho^{(k)}_{\text{Scr.},\pmb{m}}$ are nontrivial and are given by
\begin{equation}
\rho^{(k)}_{\text{Scr.},\pmb{m}}=\left(\prod_{i=1}^D \lambda^{-1}_i\right)\frac{\partial^{k} \Lambda_k}{\partial \mu_{m_1}\dots \partial \mu_{m_k}} {\Bigg\vert}_{\mu_i = \lambda_i^{-1}}~,  \label{eq:scrooge_kth_moment}
\end{equation}
where 
\begin{align}
    \Lambda_k(\mu_1,\dots, \mu_D)&\equiv \sum_{j=1}^{D} \frac{\mu^{k-2}_j\ln \mu_j}{\prod_{i\neq j} (\mu_j-\mu_i)}.
\end{align}
We derive this result in Appendix~\ref{app:Scrooge_kth_moment}. While this expression appears complicated, it reduces to a simple form in the limit that $\rho$ is a sufficiently mixed state. We illustrate this for $k=2$. When $\sum_m \lambda_m^{2n} \ll 1$ for every $n>1$, Eq.~\eqref{eq:scrooge_kth_moment} simplifies to the product form (c.f. Theorem~\ref{thm:temp_ens_product_form})
\begin{align}
    \rho^{(2)}_\text{Scr.} \approx \rho^{\otimes2} (\mathbb{I}+\mathbb S)~,
\end{align}
where $\mathbb{I}$ and $\mathbb{S}$ are the identity and swap operators permuting two copies of the Hilbert space $\mathcal{H}^{\otimes 2}$.

\subsection{Numerical evidence: trace distance between $k$-th moments}
We numerically support our claim by identifying the projected ensemble with the Scrooge ensemble. 
Detailed information about the model of our numerical simulations is provided in Section~\ref{subsec:numerical_details}. 
In Figure~\ref{fig:projected*}(c,d), we compare the $k=2,3$ moments of the Scrooge ensemble $\rho^{(k)}_\text{Scr.}$ with those of the projected ensemble
\begin{equation}
    \rho^{(k)}_\text{Proj.} \equiv \sum_{z_B} p(z_B) |\Psi(z_B)\rangle\langle \Psi(z_B)|^{\otimes k}~,
    \label{eq:proj_ens_moment}
\end{equation}
obtained from a time-evolved state $\ket{\Psi(t)}$ under the MFIM. As a function of quench evolution time $t$, the trace distance $\Delta^{(k)}$ decreases as a power law with $t$~\cite{cotler2023emergent}, before saturating at a value which depends on system size, illustrated for $k=2$ in Figure~\ref{fig:projected*}(c). This saturation value decreases exponentially with increasing system size for both $k=2$ and $k=3$ [Fig.~\ref{fig:projected*}(d)].

\subsection{Unnormalized ensembles of states}
It is mathematically more convenient to treat the projected ensemble as an ensemble of the \textit{unnormalized states} $\ket*{\tilde{\Psi}(z_B)} \equiv (\mathbb{I}_A \otimes \bra{z_B})\ket{\Psi} = \sqrt{p(z_B)}|\Psi(z_B)\rangle$, such that the projection operation $\ketbra{\Psi}{\Psi} \rightarrow \ketbra*{\tilde{\Psi}(z_B)}{\tilde{\Psi}(z_B)}$ acts linearly on the state $\ketbra{\Psi}{\Psi}$. The $k$-th moment of the ensemble of unnormalized states [Eq.~\eqref{eq:unnormalized_kth_moment}, reproduced below], differs from that of the ensemble of normalized states [Eq.~\eqref{eq:normalized_kth_moment}]. 
\begin{align}
\tilde{\rho}^{(k)} &\equiv \frac{1}{D_B}\sum_{z_B} \ketbra*{\tilde{\Psi}(z_B,t)}{\tilde{\Psi}(z_B,t)}^{\otimes k} \equiv \int d\tilde{\Psi} \ketbra*{\tilde{\Psi}}{\tilde{\Psi}}^{\otimes k}~.
\end{align}

The primary reason for considering the unnormalized ensemble moments is that we will be able to derive analytical expressions for the corresponding moments of projected ensembles generated by time-evolved states. 

In terms of unnormalized ensembles, our claim is that the projected ensemble is statistically described by the unnormalized-state analog of the Scrooge ensemble; which we refer to as the \textit{unnormalized Scrooge ensemble}, and denote $\tilde{\text{S}}\text{crooge}[\rho]$. This is simply a distortion of the Gaussian ensemble of unnormalized states [Eq.~\eqref{eq:Gaussian_ens}, Ref.~\cite{goldstein2006distribution}]. In the eigenbasis of $\rho$, the probability distribution function factorizes:
\begin{align}
    P_{\tilde{\text{S}}\text{cr.}}(\tilde{\Psi}) d\tilde{\Psi} \label{eq:max_S_dist} &= \prod_{m=1}^{D}\frac{1}{\lambda_m \pi} \exp[-\lambda_m^{-1} \vert\tilde{\Psi}_m\vert^2] d^2\tilde{\Psi}_m~.
\end{align}
Note that we can map the unnormalized $\tilde{\text{S}}\text{crooge}[\rho]$ into the normalized ensemble $\text{Scrooge}[\rho]$ by the procedure described in Sec.~\ref{sec:def_pure_state_ensembles}.

Indeed, the $k$-th moments of $\tilde{\text{S}}\text{crooge}[\rho]$ are far simpler than those of $\text{Scrooge}[\rho]$. They are \textit{exactly} of the product form:
\begin{align}
   &\tilde{\rho}^{(k)}_{\tilde{\text{S}}\text{cr.}} =  \int d \tilde{\Psi}~\exp[-\lambda_j^{-1} \vert\tilde{\Psi}_j\vert^2] \sum_{\substack{\pmb{m},\pmb{m'},\\\pmb{m'}= \sigma(\pmb{m})}} \prod_{i=1}^k |\tilde{\Psi}_{m_i}|^2 \ketbra{\pmb{m}}{\pmb{m'}}   \nonumber\\
   &= \rho^{\otimes k} \sum_{\sigma \in S_k} \text{Perm}(\sigma)~.\label{eq:distorted_Gaussian_kth_moment_result}
\end{align}

\subsection{Porter-Thomas distribution in outcome probabilities}
\label{subsec:Scrooge_PT}
The simple $k$-th moments of the unnormalized Scrooge ensembles lend themselves to a simple, measurable signature of our claim. We consider the joint probabilities $p_{o_A}(z_B) \equiv \vert\langle o_A|\tilde{\Psi}(z_B)\vert^2 = |\langle o_A,z_B|\Psi\rangle|^2$. Eq.~\eqref{eq:distorted_Gaussian_kth_moment_result} predicts that if we fix $o_A$, the distribution of probabilities $p_{o_A}(z_B)$ over the bitstrings $z_B$ will be a PT distribution $p_{o_A}(z_B) \sim \text{PT}\left(\bra{o_A}\rho_A \ket{o_A}\right)$. In Fig.~\ref{fig:projected*}(e), $A$ is a subsystem of four qubits. We plot the histogram of $p_{o_A}(z_B)$ for three different outcomes $o_A = 0000,1010$ or 0010, all of which follow PT distributions.

 In contrast to the joint probability $p_{o_A}(z_B)$, the distribution of the conditional probabilities $p(o_A|z_B) \equiv p(o_A,z_B)/p(z_B)$ is non-universal. While the moments of the joint probability $p_{o_A}(z_B)$ correspond to the moments of the unnormalized Scrooge ensemble, the moments of the conditional probability $p(o_A|z_B)= |\langle o_A | \Psi(z_B)\rangle|^2$ are given by the moments of the normalized Scrooge ensemble [Eq.~\eqref{eq:scrooge_kth_moment}].
 The predicted distribution of $p(o_A|z_B)$ is non-universal and depends on the observable $O_A$, hence is more challenging to verify, illustrating the conceptual utility of ensembles of unnormalized states.

\subsection{Maximally entropic ensembles of unnormalized states}
\label{sec:Scrooge_entropy}
That the Scrooge ensemble describes the projected ensemble is no coincidence: it is known that the \textit{unnormalized} Scrooge ensemble has maximum entropy, subject to the constraint that its first moment is $\rho_A$~\cite{parfionov2006lazy,goldstein2006distribution} (see Ref.~\cite{brody2000information} for a discussion of maximum-entropy normalized ensembles). Showing this is a simple application of Lagrange multipliers, but we reproduce it here for completeness. 
We want to solve for the probability distribution $P_\text{MaxEnt}$ which minimizes the relative entropy $D\left(\mathcal{E}\Vert\text{Gauss}\right)$, subject to the first moment constraint. This optimization may be independently solved in the eigenbasis of $\rho_A$, i.e. $P_\text{MaxEnt}$ factorizes in the eigenbasis. In each component, $P_\text{MaxEnt}$ has exponential form:
$P_\text{MaxEnt}(\tilde{\Psi}_j)/P_\text{Gauss}(\tilde{\Psi}_j) \propto \exp(-c_j \vert\tilde{\Psi}_j\vert^2)$. The unknown values $c_j$ are solved by enforcing the first moment constraint. With $\rho_A = \sum_j \lambda_j \ketbra{j}{j}$, we have
\begin{align}
    \lambda_j &= \int \vert\tilde{\Psi}_j\vert^2 P_\text{MaxEnt}(\tilde{\Psi}) d \tilde{\Psi} \\
    &= \frac{1}{Z} \left[\prod_{k\neq j} \frac{D}{D+c_k} \right] \frac{D}{(D+c_j)^2} = \frac{1}{D+c_j}~,
\end{align}
where in the second line we have performed the relevant Gaussian integrals, and we note that the normalization constant $Z$ satisfies $Z= \prod_{k}D/(D+c_k)$. It follows that $c_j = \lambda_j^{-1} - D$, precisely equal to the unnormalized Scrooge ensemble in Eq.~\eqref{eq:max_S_dist}. This derivation is fairly general and may apply to multiple settings, e.g.,~recent work in the context of maximum entropy states under charge conservation~\cite{altland2022maximum}.

The corresponding maximized value of the ensemble entropy is
\begin{align}
    &\text{Ent}(\mathcal{E}) = -\mathcal{D}(\mathcal{E}\Vert \text{Gauss}) \\
    &= \sum_j \int d\tilde{\Psi}_j P_{\tilde{\text{S}}\text{cr.}}(\tilde{\Psi}_j) \left[\left(\lambda_j^{-1}-D\right) \vert \tilde{\Psi}_j \vert^2 + \log(D \lambda_j)\right]\nonumber\\ 
    &= \sum_j\log(D \lambda_j)~,
\end{align}
identical to the (non-divergent part of the) ensemble entropy of the temporal ensemble~[Eq.~\eqref{eq:temporal_ensemble_entropy}], strongly suggesting a connection between the two ensembles.

\subsection{Basis-independent mutual information}
\label{subsec:scrooge_mutual_info}
The Scrooge ensemble was originally studied for its information theoretic properties: it is the ensemble of states with minimal \textit{accessible information} (Section~\ref{subsec:mutual_info})~\cite{jozsa1994lower}.  Specifically, Josza, Robb, and Wooters~\cite{jozsa1994lower} lower bounded the accessible information of any ensemble of states with first moment $\rho$ by the \textit{subentropy} $Q(\rho)$ [Eq.~\eqref{eq:subentropy}], and found that the Scrooge ensemble achieves this lower bound. A non-trivial, testable consequence of their result is that the mutual information $I(O_A;Z_B|T)$ is equal to the subentropy, for any complete basis $O_A$. 
\begin{align}
    I(O_A; Z_B|T) \approx Q(\rho_A)~. \label{eq:subentropy}
\end{align}
Here, the \textit{conditional} mutual information $I(O_A;Z_B|T)$ refers the mutual information between the variables $O_A$ and $Z_B$ of the probability distribution $p(o_A,z_B,t)$ at \textit{fixed times} $t$, i.e. conditioning on the time $T$. As a reminder, we denote variables with upper case letters and their values in lower case. Although time will not play a role here (as long as the evolution time is sufficiently long), we primarily adopt this notation anticipating our result in Sec.~\ref{subsec:interaction_info}, where the dependence on time will be important.

The operational meaning of the mutual information can be understood by considering the following scenario.
Suppose Alice and Bob share an entangled quantum many-body state obtained from ergodic Hamiltonian dynamics of duration $t$. 
Bob performs a measurement on his (larger) portion of the state in the computational basis, resulting in a measurement outcome $z_B$. Alice performs the measurement of an arbitrary (potentially optimized) local observable leading to an outcome $o_A$. How much information do Alice and Bob share in this exercise?
Or, equivalently, how much can Alice or Bob know about the measurement outcomes of the other from their own measurement outcomes?
Our result implies that they know \emph{as little as possible}, while some amount of information is inevitably shared due to energy conservation. 
Furthermore, this feature is universal, independent of the choice of Alice's observable as long as Bob's measurement basis realizes a Scrooge ensemble (we discuss the effects of measurement basis below).

We verify this property in Fig.~\ref{fig:projected*}(f). We choose three different measurement bases $O_A \in \{X^{\otimes 4},Y^{\otimes 4},Z^{\otimes 4}\}$, i.e. the $X$, $Y$ and $Z$ bases on the subsystem $A$. We verify that the mutual information $I(O_A;Z_B|T)$ of a time-evolved state $\ket{\Psi(t)}$ is independent of $O_A$ and agrees with the subentropy $Q(\rho_A)$ across a wide range of energy densities of the initial state $\ket{\Psi_0}$, i.e. over a wide range of temperatures of the reduced density matrix $\rho_A$. 

Taken together, our results provide a surprising universality in the behavior of projected ensembles, even away from infinite temperature. When the basis $Z_B$ is uncorrelated with the projected states $\ket{\Psi(z_B)}$, they are statistically described by independent samples from the Scrooge ensemble  $\ket{\Psi(z_B)} \sim \text{Scrooge}[\rho_A]$. We discussed measurable signatures of the Scrooge ensemble: statistics of the outcomes probabilities $p_{o_A}(z_B)$, and universal behaviour of the conditional mutual information $I(O_A;Z_B|T)$. 
However, we have not yet provided a derivation for the emergence of the Scrooge ensemble. Furthermore, our results in this section will not hold for generic measurement bases. In the following section, we resolve both issues by deriving a \textit{generalized Scrooge ensemble} from the temporal ensemble.

\section{From temporal to projected ensembles: The effect of measurement basis}
\label{sec:general_scrooge}
As aforementioned, properties of the projected ensemble depend on the measurement basis on $B$: our findings in Section~\ref{sec:scrooge} only hold for certain measurement bases. In this section, we seek to extend our description of projected ensembles to general bases. Instead of Scrooge ensembles, we will find that projected ensembles are described by a \textit{generalized Scrooge ensemble}, which we introduce, that accounts for correlations between $z_B$ and the projected states $\ket{\Psi(z_B)}$. 

For projected ensembles generated by from time-evolved states, we derive this result from properties of the temporal ensemble. Our finding is more general than time-evolved states: we numerically find that the projected ensembles arising from energy eigenstates are also well described by generalized Scrooge ensembles. We discuss measurable and information-theoretic signatures of the generalized Scrooge ensemble, which mirror that of the Scrooge ensemble. Finally, we identify conditions under which the generalized Scrooge ensemble reduces to the Scrooge ensemble. These results give a complete description of the projected ensembles generated by time-evolved states, proven for generic Hamiltonian dynamics, and broaden our understanding of deep thermalization.

As an illustrative example, we consider the mixed-field Ising model (MFIM). For convenience, we repeat its Hamiltonian [Eq.~\eqref{eq:MFIM}] here.
\begin{equation}
    H_\text{MFIM} = \sum_{j=1}^N h_x X_j + h_y Y_j + \sum_{j=1}^{N-1} J X_j X_{j+1}~,
\end{equation}
The projected ensemble is equal to the Scrooge ensemble when the subsystem $B$ is measured in the $Z_B$ basis. In other bases, e.g. the $X_B$ basis, there are correlations between the measurement outcomes $\ket{x_B}$ and the projected states $\ket{\Psi(x_B)}$. In particular, their energies are anti-correlated: measurement outcomes $\ket{x_B}$ with high energies $\bra{x_B} H_B \ket{x_B}$ are likely to result in projected states with lower energies $\bra{\Psi(x_B)}H_A \ket{\Psi(x_B)}$ (and vice versa), where $H_A$ and $H_B$ are the terms of the Hamiltonian restricted to the subsystems $A$ and $B$ respectively. Note that the energies we refer to here are local energies on subsystems $A$ and $B$ and therefore such energetic correlations are not strict, but exist in a statistical sense, see e.g. Ref.~\cite{murthy2019structure}. Under these conditions, one must specify exponentially many probability distributions which take into account the correlations between $x_B$ and $\ket{\Psi(x_B)}$ arising from energy or other conservation laws. This is in contrast to what we have discussed in Sec.~\ref{sec:scrooge}, where we could treat the projected states $\ket{\Psi(x_B)}$ with different labels (measurement outcomes) $x_B$ as samples from the same probability distribution. 
To emphasize the effect of measurement bases, in this section we will refer to such measurement outcomes and projected states as $x_B$ and $\ket{\Psi(x_B)}$ respectively.

\textit{Special case: Magnetization-sectored bases ---} To gain intuition about the effects of measurement basis, it is illustrative to consider a special case: one where the outcomes can be partitioned into sets (``sectors"),  within which there are no further correlations between the outcomes and projected states. For example, the MFIM we consider has a Zeeman field in the $Y$-axis. When this Zeeman field is sufficiently large, the sector $M$ of bitstrings $\ket{y_B}$ with the same magnetization have approximately the same energy on $B$, but have different energies from bitstrings from other sectors. One consequence of this is the following: consider the ``sector-resolved reduced density matrix," defined as
\begin{align}
\rho_A(M) &=\frac{\sum_{y_B \in M} p(y_B) \ketbra*{\Psi(y_B,t)}{\Psi(y_B,t)}}{\sum_{y_B \in M} p(y_B)}~.
\end{align}
These mixed states $\rho_A(M)$ show systematic differences across magnetization sectors $M$. Accordingly, it is not appropriate to take the projected states from $|\Psi(y_B,t)\rangle$ as identical samples from the same ensemble of states. Rather, projected states corresponding to the same sector $M$ are drawn from the same ensemble. As we will show below, these are Scrooge ensembles, with different first moments: 
\begin{equation}
   \forall y_B \in M~,~ \ket*{\Psi(y_B)} \sim \text{Scrooge}[\rho_A(M)]~.
\end{equation}
However, this sector picture is only valid in special cases. Our generalized Scrooge ensemble extends this observation to the general case.

\begin{figure*}
    \centering
    \includegraphics[width=0.9\textwidth]{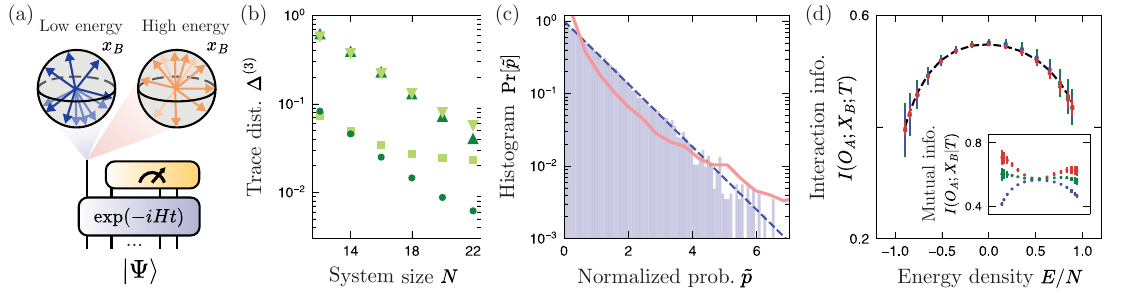}
\caption{Generalized Scrooge ensembles describe projected ensembles. (a) When the measurement basis on $B$ (here denoted $X_B$) is correlated with the total energy of the system, the projected states $\ket{\Psi(x_B)}$ are correlated with the measurement outcome $x_B$, here illustrated by differences between outcomes with ``low" (orange) and ``high" (blue) energies. This leads to deviations from the Scrooge ensemble, and the projected ensemble is instead described by a generalized Scrooge ensemble (Claim~\ref{claim:generalized_ensemble}). (b) We verify this conjecture numerically with the trace distance $\Delta^{(3)}$ between the third moments of the projected ensemble and either the Scrooge ensemble 
(squares, downward triangle) or the generalized Scrooge ensemble (circles and upward triangles). We consider projected ensembles of states time-evolved (for sufficiently long times) under the MFIM [Eqs.~(\ref{eq:MFIM}-\ref{eq:initial_state})] and measured in the $X_B$ basis.
The difference between the two ensembles is more obvious for subsystems of two qubits (squares, circles) than for subsystems of four qubits (upward, downward triangles), due to the larger bath relative to the system size. While the distance to the Scrooge ensemble saturates to a non-zero value, the distance to the generalized Scrooge ensemble continues to decrease. (c) Consistent with the generalized Scrooge ensemble, the histogram (blue) of the rescaled probabilities $\tilde{p}(x_B) \equiv p_{x_A}(x_B)/\mathbb{E}_t[p_{x_A}(x_B)]$ follows a PT distribution, while the bare probabilities $D p_{x_A}(x_B)$ (red line) do not, with $x_A = 0000$ plotted here as a representative. (d) As with the mutual information in Scrooge ensembles, in generalized Scrooge ensembles the \textit{interaction information} $I(O_A;X_B;T)$ displays a universal behavior, independent of basis $O_A$, and agrees with the \textit{weighted sum of subentropies} [Eq.~\eqref{eq:weighted_subentropies}] (dashed). Inset: in contrast, the mutual information $I(O_A;X_B|T)$ is no longer independent of the basis $O_A$, here illustrated in red, green, and blue respectively by $O_A = \{X^{\otimes 4},Y^{\otimes 4},Z^{\otimes 4}\}$, due to time-averaged (thermal) correlations between the measurement bases $O_A$ and $X_B$. }
    \label{fig:generalized_scrooge}
\end{figure*}

\subsection{Generalized Scrooge Ensemble}
Our key idea is to treat each projected state $\ket*{\Psi(x_B)}$ as a random sample from a \textit{different} Scrooge ensemble $\text{Scrooge}[\bar{\rho}(x_B)]$, where $\bar{\rho}(x_B)$ is a density matrix associated to each outcome $x_B$. We define the density matrix $\bar{\rho}(x_B)$ be the time-averaged state
\begin{align}
    \bar{\rho}(x_B) &\equiv  \mathbb{E}_t[|\tilde{\Psi}(x_B)\rangle\langle \tilde{\Psi}(x_B)|]/\mathbb{E}_t[\langle\tilde{\Psi}(x_B)|\tilde{\Psi}(x_B)\rangle]\nonumber\\
&=\left(\mathbb{I}_A \otimes \bra{x_B}\right) \rho_d \left(\mathbb{I}_A \otimes \ket{x_B}\right)/p_d(x_B)~,\label{eq:average_state_xB}\\
    p_d(x_B) &\equiv \mathbb{E}_t[\langle\tilde{\Psi}(x_B)|\tilde{\Psi}(x_B)\rangle] \nonumber\\
    &=  \tr\left[\left(\mathbb{I}_A \otimes \bra{x_B}\right) \rho_d \left(\mathbb{I}_A \otimes \ket{x_B}\right)\right]~.
\end{align}
 $p_d(x_B)$ is the time-averaged probability of measuring the outcome $x_B$, and we express both $\bar{\rho}(x_B)$ and $p_d(x_B)$ in terms of the diagonal ensemble $\rho_d \equiv \mathbb{E}_t[|\Psi(t)\rangle\langle \Psi(t)|]$. We refer to this collection of probability distributions, each sampled from once, as the generalized Scrooge ensemble.

\begin{claim}
\label{claim:generalized_ensemble}
    The states $\ket*{\Psi(x_B)}$ of the projected ensemble are sampled from the \textit{generalized Scrooge ensemble}
    \begin{align}
        \ket*{\Psi(x_B)} \sim \text{Scrooge}[\bar{\rho}(x_B)]~.
    \end{align}
\end{claim}

In order to support this claim, we compare the $k$-th moments of the normalized projected ensemble with the predicted $k$-th moment
\begin{equation}
    \rho^{(k)}_\text{Gen.~Scr.} = \sum_{x_B} p_d(x_B) \rho^{(k)}_\text{Scrooge}(\bar{\rho}(x_B))~.
\label{eq:gen_scrooge_moment}
\end{equation}
As illustrated in Fig.~\ref{fig:generalized_scrooge}(b), the higher moments of the projected ensemble agree with $\rho^{(k)}_\text{Gen.~Scr.}$ and differ from $\rho^{(k)}_\text{Scrooge}$ [Eq.~\eqref{eq:scrooge_kth_moment}], with exponential convergence with increasing system size. Here, we have calculated the average states $\bar{\rho}(x_B)$ by numerically time-averaging $\ketbra*{\tilde{\Psi}(x_B)}{\tilde{\Psi}(x_B)}$ over a long interval. Having numerically demonstrated our claim, we proceed to derive it with results from the temporal ensemble.

\textit{Temporal ensemble of projected states ---} In order to derive results about projected ensembles, we first establish a result about the temporal ensemble of a projected state. Specifically, we consider the set of unnormalized states $\ket*{\tilde{\Psi}(x_B,t)}$ for a fixed outcome $x_B$. We show that the time trajectory (temporal ensemble) is described by the unnormalized ensemble $\tilde{\text{S}}\text{crooge}[\bar{\rho}(x_B)]$.

\begin{lemma}[Temporal $k$-th moment of the projected state]
\label{lemma:temp_ens_of_proj_st}
Consider the temporal ensemble comprised of the (unnormalized) projected states $\ket*{\tilde{\Psi}(x_B,t)} \equiv (\mathbb{I}_A \otimes \bra{x_B})\ket{\Psi(t)}$. The $k$-th moment of this temporal ensemble is equal to the $k$-th moment of the corresponding unnormalized Scrooge ensemble, up to small corrections:
\begin{align}
    &\rho^{(k)}_\text{Temp. Proj.}(x_B) \equiv \mathbb{E}_t \left[\ketbra*{\tilde{\Psi}(x_B,t)}{\tilde{\Psi}(x_B,t)}^{\otimes k} \right]\label{eq:temp_ensemble_of_proj_state} \\
    &= p_d(x_B)^k\tilde{\rho}^{(k)}_{\tilde{\text{S}}\text{cr.}}[\bar{\rho}(x_B)] + O(\Delta_\beta(x_B)p_d(x_B)^{k})~,\nonumber 
\end{align}
where $\Delta_\beta(x_B) \equiv  \tr[(\mathbb{I}_A\otimes \ketbra{x_B}{x_B})^{\otimes 2}\rho_d^{(2)}]/p_d(x_B)^2$ is a number, typically exponentially small in system size, characterizing the validity of our approximation and where $\rho_d^{(2)} \equiv \abs{c_E}^4 \ketbra{E,E}{E,E}$.
\end{lemma}
\begin{proof}
    The leading term follows from Theorem~\ref{thm:temp_ens_product_form}. However, the bound $O(\Delta_\beta(x_B))$ does not immediately follow from the bound given in Theorem~\ref{thm:temp_ens_product_form}, which is too weak and should be improved for this context.

    We outline the proof here, with full details in Appendix~\ref{app:proof_temp_ens_of_proj_st}. We use an intermediate step in the proof of Theorem~\ref{thm:temp_ens_product_form} in Appendix~\ref{app:temporal_ensemble_error}. This step expresses the difference $\rho^{(k)}_\text{Temp. Proj.}(x_B) -p_d(x_B)^k\tilde{\rho}^{(k)}_{\tilde{\text{S}}\text{cr.}}[\bar{\rho}(x_B)] $ in terms of lists of energies  $(E_1,\dots , E_k)$--- those with at least one repeated element. This allows us to replace two copies of the diagonal matrix $\rho_d^{\otimes 2} \equiv \sum_{E_1, E_2} \abs{c_{E_1}}^2 \abs{c_{E_2}}^2 \ketbra{E_1,E_2}{E_1,E_2}$ with the term $\rho_d^{(2)} \equiv \sum_E \abs{c_E}^4 \ketbra{E,E}{E,E}$. The relative size of the correction term is then $\Delta_\beta(x_B)$, which is typically exponentially small in total system size (Appendix~\ref{app:gen_Scrooge_proof}). 
\end{proof}

\textit{Projected ensemble of time-evolved states ---} While Lemma~\ref{lemma:temp_ens_of_proj_st} is interesting in its own right, we use our notion of ergodicity (Sec.~\ref{subsec:ergodicity}) to turn it into a statement about the projected ensemble at fixed (but generically late) points in time.

\begin{theorem}
\label{thm:moment_of_general_projected_ensemble}
    At \textit{typical} (late) times $t$, the weighted $k$-th moment $\tilde{\rho}^{(k)}_\text{Proj.}(t)$ of the unnormalized projected ensemble is equal to the $k$-th moment of the generalized Scrooge ensemble:
    \begin{align}
       \tilde{\rho}^{(k)}_\text{Proj.}(t) &\equiv \sum_{x_B} \frac{\ketbra*{\tilde{\Psi}(x_B,t)}{\tilde{\Psi}(x_B,t)}^{\otimes k}}{p_d(x_B)^{k-1}} \label{eq:kth_moment_proof}\\
       &= \sum_{x_B} p_d(x_B) \tilde{\rho}^{(k)}_{\tilde{\text{S}}\text{cr.}}[\bar{\rho}(x_B)]+O(\Delta^{1/2}_\beta) ~,\label{eq:general_projected_moment}
    \end{align}
    where $\Delta_\beta \equiv \sum_{x_B} p_d(x_B) \Delta_\beta(x_B)$, the weighted mean of $\Delta_\beta(x_B)$.
\end{theorem}

\begin{proof} 
    From Lemma~\ref{lemma:temp_ens_of_proj_st}, the \textit{time-averaged} $k$-th moment $\mathbb{E}_t[\tilde{\rho}^{(k)}_\text{Proj.}(t)]$ is approximately equal to the weighted sum $\sum_{x_B} p_d(x_B) \tilde{\rho}^{(k)}_{\tilde{\text{S}}\text{cr.}}[\bar{\rho}(x_B)]$. We then show in Appendix~\ref{app:gen_Scrooge_proof} that the variance over time $\text{var}_t [\tilde{\rho}^{(k)}_\text{Proj.}(t)]$ is exponentially small, and therefore the projected $k$-th moment  $\tilde{\rho}^{(k)}_\text{Proj.}(t)$ is close to its temporal average. The temporal fluctuations are of size $\Vert \text{var}_t[\tilde{\rho}^{(k)}_\text{Proj.}(t)]\Vert = O(\Delta_\beta)$. While the most general bound on $\Delta_\beta$ that we can prove is $\Delta_\beta \leq 1$, we argue in Appendix~\ref{app:gen_Scrooge_proof} that it is typically exponentially small.
\end{proof}

Finally, we note that the $k$-th moment described in Eq.~\eqref{eq:kth_moment_proof} is a slightly different moment than both the moments of the normalized and unnormalized projected ensembles. Each projected state is weighted by the factor $p_d(x_B)^{1-k}$ in order to have a quantity that can be bounded by a small quantity ($\Delta_\beta$). The norm $\braket*{\tilde{\Psi}(x_B)}{\tilde{\Psi}(x_B)}$ fluctuates about the value $p_d(x_B)$. These fluctuations are not universal\footnote{More precisely, the norm is a random variable which is a weighted sum of PT variables, where the weights are the eigenvalues of $\bar{\rho}(x_B)$. The resulting distribution is known as a \textit{hypoexponential distribution} and is discussed extensively in companion work~\cite{shaw2024universal} but depends on the spectrum of $\bar{\rho}(x_B)$.}. Accordingly, the average value of $\braket*{\tilde{\Psi}(x_B)}{\tilde{\Psi}(x_B)}^k/p_d(x_B)^k$ is equal to the $k$-th moment of this distribution, which is larger than 1. With increasing size of the subsystem $A$, these values concentrate around 1 and therefore, our weighted $k$-th moment in Eq.~\eqref{eq:kth_moment_proof} converges to the normalized ensemble moment.

\subsection{Porter-Thomas distribution in generalized Scrooge ensembles}
\label{subsec:gen_Scrooge_PT}

In Sections~\ref{subsec:temp_ens_PT_dist} and~\ref{subsec:Scrooge_PT}, we discussed the Porter-Thomas distribution as a measurable signature of the temporal and Scrooge ensembles respectively. Perhaps unsurprisingly, the PT distribution is also present in the generalized Scrooge ensemble, but only after suitable data processing. In Sec.~\ref{subsec:Scrooge_PT}, we established that in the Scrooge ensemble, the joint probabilities $p_{o_A}(z_B)$ follow a PT distribution with mean $\mathbb{E}_{z_B}[p_{o_A}(z_B)] \equiv \bra{o_A}\rho_A \ket{o_A}/D_B$. In the generalized Scrooge ensemble, since each projected state $\ket*{\tilde{\Psi}(x_B)}$ is sampled from a different Scrooge ensemble, we must normalize the overlaps $p_{o_A}(x_B)$ by their respective time-averaged values $\mathbb{E}_t[p_{o_A}(x_B)] = \bra{o_A} \bar{\rho}(x_B)\ket{o_A}$ (see~Sec.~\ref{subsec:ergodicity}). Having done so, 
\begin{equation}
    p(o_A,z_B)/\mathbb{E}_t[p(o_A,z_B)] \sim \text{PT}(1)~,
    \label{eq:PT_in_gen_Scrooge}
\end{equation}
as illustrated in Fig.~\ref{fig:generalized_scrooge}(c).

\subsection{Interaction information}
\label{subsec:gen_Scrooge_info}
\label{subsec:interaction_info} We next turn to the information-theoretic properties of the generalized Scrooge ensemble. Specifically, we consider the mutual information $I(O_A;X_B)$ between any basis $\{o_A\}$ on $A$ and projective measurements $\{\ket{x_B}\}$. Unlike in the Scrooge ensemble, there are additional correlations between $x_B$ and the projected states $\ket*{\Psi(x_B)}$. With the generalized Scrooge ensemble, we find that the mutual information $I(O_A;X_B|T)$ (i.e. the mutual information of the distribution $p(o_A z_B,t)$ at fixed, late times) can be written as follows:
\begin{equation}
    I(O_A;X_B|T) \approx I(O_A;X_B) + \sum_{x_B} p_d(x_B) Q\left[\bar{\rho}(x_B)\right]~, \label{eq:interaction_information}
\end{equation}
where the first term $I(O_A;X_B)$ refers to correlations that exist even in the time-averaged distribution $\mathbb{E}_t[p(o_A z_B,t)]$. These correlations are \textit{non-universal} and are thermal in nature, arising from energy correlations between the bases $\{\ket{o_A}\}$ and $\{\ket{x_B}\}$. In contrast, the second term captures the correlations associated with fluctuations in time. These fluctuations in time are universal and are expressed as a weighted sum of subentropies $Q\left[\bar{\rho}(x_B)\right]$. Eq.~\eqref{eq:interaction_information} follows from the following:
\begin{align}
    &I(O_A;X_B|T) - I(O_A;X_B) \nonumber\\
    &=H(O_A|T) - H(O_A|X_B,T) - H(O_A) + H(O_A|X_B) \nonumber\\
    &\approx H(O_A|X_B) - H(O_A|X_B,T) \label{eq:conditional_entropies}\\
    &\approx \sum_{x_B} p_d(x_B) \left[H(O_A|X_B=x_B)-H(O_A|X_B=x_B,T)\right]\label{eq:weighted_mutual_informations}\\
    &\approx \sum_{x_B} p_d(x_B) Q\left[\tilde{\rho}(x_B)\right]~. \label{eq:weighted_subentropies}
\end{align}
Eq.~\eqref{eq:conditional_entropies} we have used the fact that $H(O_A|T)\approx H(O_A)$, because the time-dependent reduced density matrix $\rho_A$ is exponentially close to its time averaged value $\tr_B[\rho_d]$. We then approximate $H(O_A|X_B,T) = \sum_{x_B} p(x_B) H(O_A|X_B=x_B,T)\approx p_d(x_B) H(O_A|X_B=x_B,T)$. This is justified because $p(x_B,t)$ fluctuates about $p_d(x_B)$. When summed over in $H(O_A|X_B,T)$, these fluctuations average away and we can replace $p(x_B)$ with $p_d(x_B)$. Eq.~\eqref{eq:weighted_mutual_informations} is a weighted sum of the mutual information for the conditional distributions $p(o_A,t|x_B)$. Theorem~\ref{thm:moment_of_general_projected_ensemble} allows us to equate this mutual information with the subentropies $Q[\bar{\rho}(x_B)]$, giving the desired Eq.~\eqref{eq:interaction_information}.

The quantity $I(O_A;X_B;T)\equiv I(O_A;X_B|T) - I(O_A;X_B)$ is known as the \textit{interaction information}, which is a multivariate analog of the mutual information~\cite{mcgill1954multivariate}. It is a measure of the correlations that exist due to interaction among all variables, and is symmetric under exchange of all variables, and is also the classical analogue of the (quantum) \textit{tripartite mutual information} studied in many-body physics~\cite{rangamani2015entanglement,rota2016tripartite}. In particular, our result predicts that the interaction information is independent of basis $O_A$, and has expression $\sum_{x_B} p_d(x_B) Q[\bar{\rho}(x_B)]$, which we can bound as $\sum_{x_B} p_d(x_B) Q[\tilde{\rho}(x_B)] \leq Q[\rho_A]$ using the concavity of the subentropy~\cite{datta2014properties}. The interaction information quantifies the amount of information unique to a time-evolved state, i.e. with time-averaged correlations subtracted. Similar to our results in the temporal and Scrooge ensembles, this information in universal, and is minimized subject to constraints.

We demonstrate this result with numerical data in Fig.~\ref{fig:generalized_scrooge}(d), in which we plot the interaction information for three measurement bases $O_A \in \{X^{\otimes 4},Y^{\otimes 4},Z^{\otimes 4}\}$. While the mutual information $I(O_A;X_B|T)$ depends on the choice of $O_A$ (inset), the interaction information $I(O_A;X_B;T)$ is independent and is equal to the weighted average subentropy~[Eq.~\eqref{eq:weighted_subentropies}]. 

\subsection{Recovering the Scrooge ensemble from the generalized Scrooge ensemble}
\label{subsec:gen_scrooge_to_scrooge}
The generalized Scrooge ensemble in Theorem~\ref{thm:moment_of_general_projected_ensemble} specializes to the original Scrooge ensemble when the states $\bar{\rho}(x_B)$ are independent of measurement outcome $x_B$, i.e.~$\bar{\rho}(x_B)\propto \rho_A$. As presented in Section~\ref{sec:scrooge}, this is approximately true when states evolved by the MFIM are measured in the $\{\ket{z_B}\}$ basis. The crucial property of the $Z_B$ measurement basis is that the energy expectation values $E(\ket{z})\equiv\braket{z}{H|z}$ and energy variances $\sigma^2_H(\ket{z})\equiv\braket{z}{(H-E_z)^2|z}$ are independent of the bitstring $z$. This is an example of ``minimal correlation" between measurement basis and energy eigenbasis.

We argue that these two properties are sufficient for the Scrooge ensemble to be recovered approximately for $z_B$ measurements in the MFIM. In particular, we will show that for any state $\ket{\phi_A}\in \mathcal{H}_A$, 
$\braket{\phi_A|\rho_d(z_B)}{\phi_A}\simeq \braket{\phi_A|\rho_d}{\phi_A}$ up to subleading corrections in the size $N$ of the system, and therefore $\rho_d(z_B)\simeq \rho_d$.

From the definition of $\rho_d(z_B)$, the expectation value $\braket{\phi_A|\rho_d(z_B)}{\phi_A}$ can be expressed in terms of the eigenstates $\ket{E}$ of $H$ as
\begin{equation}
\braket{\phi_A|\rho_d(z_B)}{\phi_A}=\sum_E |c_E|^2 |\braket{\phi_A, z_B}{E}|^2 \label{eq:Scrooge_MFIM_distribution}
\end{equation}
We next use the theorem proved in Ref.~\cite{hartmann2004gaussian}, which states that for large $N$, the energy distribution of the product state $\ket{\phi_A, z_B}$ converges to a Gaussian centered about $E(\ket{\phi_A,z_B})=\braket{\phi_A,z_B}{H|\phi_A,z_B}$ with variance $\sigma^2_H(\ket{\phi_A,z_B})) = \braket{\phi_A,z_B}{(H-E(\ket{\phi_A,z_B}))^2|\phi_A,z_B}$.
Furthermore, the mean $E(\ket{\phi_A,z_B})$ and variance $\sigma^2_H(\ket{\phi_A,z_B}))$ are independent of $z_B$ up to subleading corrections. To show this, we write the Hamiltonian as $H=H_A+H_B+H_{AB}$, where $H_A, H_B$ only act on subsystems $A,B$ respectively and $H_{AB}$ is the remaining term, acting on the boundary between $A$ and $B$. As mentioned above, the expectation value of $H_B$ and variance are independent of $z_B$:
\begin{align}
E_B(\ket{z_B})&\equiv\braket{z_B}{H_B|z_B}=0, \\ \sigma^2_{H_B}(\ket{z_B})&\equiv\braket{z_B}{(H_B-E_B(\ket{z_B}))^2|z_B}\\
    &=N_B(h_x^2+h_y^2)+(N_B-1)J^2~.\nonumber
\end{align}
We also have $\braket{\phi_A,z_B}{H_{AB}|\phi_A, z_B}=0$, which gives
\begin{align}
&E(\ket{\phi_A,z_B})=E_A(\ket{\phi_A}),\label{eq:Ezb}\\
 &  \sigma^2_{H}(\ket{\phi_A,z_B}) =\sigma^2_{H_A}(\ket{\phi_A})+\sigma_{H_B}^2+O(1)\label{eq:varzb} 
\end{align}
where the terms $O(1)$ come from the boundary contributions and hence do not scale with system size. From Eqs. (\ref{eq:Ezb}) and (\ref{eq:varzb}), we find that the energy distribution of $\ket{\phi_A, z_B}$, and consequently the quantity $\braket{\phi_A}{\rho_d(z_B)|\phi_A}$ in Eq.~\eqref{eq:Scrooge_MFIM_distribution} are independent of the measurement outcome $z_B$ (up to subleading corrections in system size), recovering the Scrooge ensemble from the generalized Scrooge ensemble in the limit of large system sizes.

\subsection{Projected ensembles of energy eigenstates}
\label{subsec:eigenstates}
\begin{figure*}
    \centering
    \includegraphics[width=0.85\textwidth]{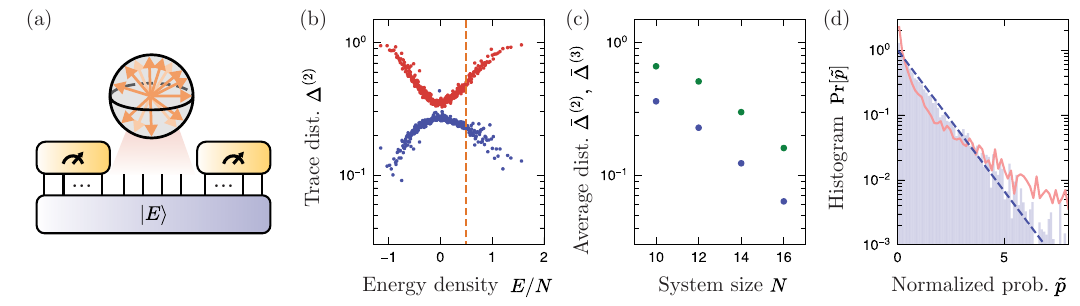}
    \caption{Scrooge and generalized Scrooge ensembles from energy eigenstates. (a) We study the projected ensembles arising from energy eigenstates of the MFIM with broken time-reversal symmetry [Eq.~\eqref{eq:MFIM_no_TRS}], with total size $N=12$ and a subsystem of size $N_A = 4$ in the middle of the chain.
    (b) For each eigenstate $\ket{E}$, we construct its reduced density matrix $\rho_A^{(E)}$ and corresponding Scrooge ensemble $\text{Scrooge}[\rho_A^{(E)}]$. The projected ensemble is always close to the Scrooge ensemble prediction, as measured by the trace distance $\Delta^{(2)}$ of their second moments (blue), and is only close to the Haar ensemble at infinite temperature, when $E/N\approx 0$ (red). Curiously, and unexpectedly, the generalized Scrooge ensemble appears to be better descriptions at the edges of the spectrum. (c) The trace distances $\Delta^{(2)}$ (blue) and $\Delta^{(3)}$ (green), averaged over the 100 eigenstates closest to the target energy $E=0.5 N$ [orange line in (b)] decrease exponentially with system size. (d) For a fixed eigenstate $\ket{E}$, the histogram of normalized measurement probabilities $\tilde{p} \equiv p_{x_A}^{(E)}(x_B)/\mathbb{E}_E[p_{x_A}^{(E)}(x_B)]$ follows a PT distribution, while the raw probabilities $D p_{x_A}^{(E)}(x_B)$ (red line) do not, supporting the presence of the generalized Scrooge ensemble. Here we plot $x_A = 0010$ as a representative.}
    \label{fig:eigenstates}
\end{figure*}
Thus far, we have investigated the projected ensembles of time-evolved states $\ket{\Psi(t)}$, and provided analytical arguments based on the temporal ensemble. In fact, the generalized Scrooge ensemble also seems to describe projected ensembles obtained from energy eigenstates, as verified by numerical simulation results illustrated in Fig.~\ref{fig:eigenstates}. 

Here, we consider the projected ensembles arising from eigenstates of the MFIM, with additional terms which remove time-reversal symmetry [Eq.~\eqref{eq:MFIM_no_TRS} and in Appendix~\ref{app:time_reversal_symmetry}]. 
We first test our simpler prediction: the emergence of the Scrooge ensemble by measuring in the $Z_B$ basis. In Fig.~\ref{fig:eigenstates}(b), the projected ensemble agrees with its respective Scrooge ensemble for energy eigenstates at all energies, but only agrees with the Haar ensemble at infinite temperature, when $E/N\approx 0$. Furthermore, the distance between the projected and Scrooge ensemble decreases exponentially with total system size [Fig.~\ref{fig:eigenstates}(c)]. 

We next test our generalized Scrooge ensemble prediction, by generating the projected ensemble with measurements in the $X_B$ basis. In this setting, it is difficult to make quantitatively accurate predictions for the generalized Scrooge ensemble, because we do not have a prescription of the average state $\bar{\rho}(x_B)$. The expression $\bar{\rho}(x_B) \propto \mathbb{E}_t[\ketbra*{\tilde{\Psi}(x_B,t)}{\tilde{\Psi}(x_B,t)}]$ in time-evolved states does not apply to energy eigenstates. Nevertheless, motivated by the eigenstate thermalization hypothesis (ETH), we expect that the eigenstates with nearby energy values will generate highly-similar projected ensembles: in particular, the average states $\bar{\rho}(x_B,E)$ should be smooth functions of the total energy. Therefore, we can estimate $\bar{\rho}(x_B,E)$ by averaging the projected states $\ketbra*{\tilde{\Psi}_E(x_B)}{\tilde{\Psi}_E(x_B)}$ over energy eigenstates of a sufficiently narrow energy window. While this method is not quantitatively accurate enough to give exponentially-decreasing trace distances (the equivalent of Fig.~\ref{fig:eigenstates}(c), but for the generalized Scrooge ensemble), we can observe the PT distribution in normalized outcome probabilities [Fig.~\ref{fig:eigenstates}(d)]. We first average the probability outcomes $p_{x_A}^{(E)}(x_B) \equiv \abs{\braket{E}{x_A,x_B}}^2$ over a small energy window to obtain $\mathbb{E}_E[p_{x_A}^{(E)}(x_B)]$ for each value of $x_A$ and $x_B$. The histogram of normalized probabilities $\tilde{p}\equiv p_{x_A}^{(E)}(x_B)/\mathbb{E}_E[p_{x_A}^{(E)}(x_B)]$ over all $x_B$'s (with fixed $x_A$) shows a PT distribution, while the raw probabilities $p_{x_A}^{(E)}(x_B)$ do not.

\textit{Time-reversal symmetry and the real Scrooge ensemble ---} Finally, we remark on the role of symmetry. Unlike in time-evolved states, the presence of time-reversal symmetry in a Hamiltonian $H$ qualitatively affects the projected ensembles from energy eigenstates. Time-reversal symmetry --- defined as the existence of a unitary operator $U$ that maps $H$ to its complex conjugate $H^*$: $UHU^\dagger = H^*$~\cite{ryu2010topological} --- ensures that all eigenstates of $H$ have purely real coefficients in a certain basis. If this coincides with the measurement basis (e.g.~the $X_B$ basis in the MFIM), then the projected states are guaranteed to be real vectors. In such cases, the projected ensemble instead follows a \textit{real Scrooge ensemble}, which is an analogous ``$\rho$-distortion" of the ensemble of real random vectors. The real Scrooge ensemble has corresponding signatures such as a ``real PT distribution"~\cite{porter1956fluctuations} which we demonstrate in Appendix~\ref{app:time_reversal_symmetry}. Here we simply break the time-reversal symmetry by adding appropriate terms:
\begin{align}
    H'_\text{MFIM} = &\sum_{j=1}^N h_x X_j + h_y Y_j + h_z Z_j \label{eq:MFIM_no_TRS}\\
    &+ \sum_{j=1}^{N-1} J X_j X_{j+1} + J' Y_j Y_{j+1}~, \nonumber
\end{align}
with $J'= 0.4$, $h_z = 0.5$ and the values of $h_x,h_y,J$ kept the same as the other numerical examples [Eq.~\eqref{eq:MFIM}].

The presence of time-reversal symmetry was not relevant to the projected ensembles of time-evolved states because the time-evolution naturally breaks time-reversal symmetry. However, time-evolved states can be engineered to be purely real, requiring both time-reversal and particle-hole symmetry and suitably symmetric initial states and measurement bases (Appendix~\ref{app:time_reversal_symmetry}).

\section{Porter-Thomas distributions as signatures of randomness: Higher-order ETH} 
\label{sec:eigenstate_PT}
\begin{figure*}
    \centering
    \includegraphics[width=0.8\textwidth]{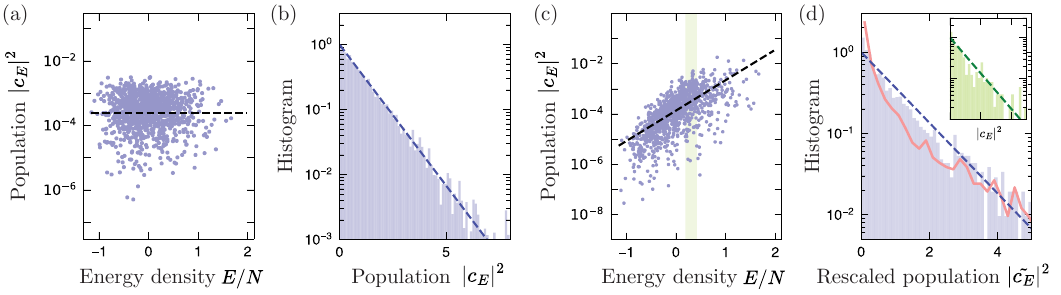}
    \caption{Porter-Thomas (PT) distribution in energy populations. We observe the presence of the PT distribution in the energy populations $\abs{c_E}^2 \equiv \abs{\braket{E}{\Psi_0}}^2$ as a consequence of the ``higher-order ETH"~\cite{kaneko2020characterizing}. In a narrow energy window, the populations are distributed according to a Porter-Thomas distribution, with a smoothly varying mean $f(E)$. (a) A state at infinite effective temperature $\ket{\Psi} =  |0\rangle^{\otimes N}$ has constant mean $f(E)$ (black dashed line). (b) The histogram of $\abs{c_E}^2$ over \textit{all eigenstates} follows a PT distribution. (c) In contrast, for a state $\ket{\Psi} = (\exp(i\theta X/2)|0\rangle)^{\otimes N}$ at finite temperature (details below), $\abs{c_E}^2$ monotonically increases about the smooth function $f(E) \propto \exp(\beta E)$, with $\beta$ fitted (dashed line). (d) The histogram of $\abs{c_E}^2/f(E)$ shows a PT distribution, both over the entire spectrum as well as in a narrow energy window [inset, energy window denoted by green band in (c)], while the raw populations $D \abs{c_E}^2$ at finite temperature (red line) do not.}
    \label{fig:eigenstate_overlaps_PT}
\end{figure*}

In this work we demonstrated instances of the PT distribution emerging from the maximum entropy principle in many-body dynamics: in the temporal and Scrooge ensembles, the PT distribution describes the \textit{probabilities of outcome probabilities} $\text{Pr}[p]$. We found that in the generalized Scrooge ensemble, it was necessary to appropriately rescale the outcome probabilities to obtain the PT distribution. 

Here, we demonstrate another instance of the PT distribution, hinting at the wider applicability of the maximum entropy principle in many-body physics. Specifically, we find the PT distribution in the distribution of overlaps $\abs{c_E}^2 = \abs{\braket{\Psi}{E}}^2$ between energy eigenstates $\ket{E}$ and a simple physical state $\ket{\Psi}$. This is consistent with a ``higher-order" or $k$-ETH~\cite{kaneko2020characterizing}, in which energy eigenstates in a narrow energy window are conjectured to form a state $k$-design in an appropriate microcanonical subspace. This is a generalization of the conventional ETH, which is a statement about the the behavior of properties such as the overlaps $\bra{E_j}O\ket{E_k}$ between energy eigenstates $\ket{E_k}$ and operators $O$. 
Specifically, Ref.~\cite{kaneko2020characterizing} defines the $k$-ETH to hold \textit{for a given $k$-copy operator $O$} if
\begin{equation}
    \mathbb{E}_t[U_t^{\otimes k} \Pi_E O \Pi_E U_{-t}^{\otimes k}] = \mathbb{E}_U[U^{\otimes k} \Pi_E O \Pi_E U^{\dagger,\otimes k}]~,
\label{eq:kETH}
\end{equation}
where $\Pi_E$ is a projector onto a microcanonical subspace $\mathcal{H}_\text{MC}$, spanned by the energy eigenstates in a narrow energy window, $U_t \equiv \exp(-iHt)$, and the average over $U$ is over Haar-random unitaries acting on the microcanonical space $\mathcal{H}_\text{MC}$. They numerically verify that Eq.~\eqref{eq:kETH} approximately holds when $O = o^{\otimes k}$ (for local operators $o$) and when $O = \mathcal{S}_A$ is the two-copy swap operator on a subsystem $A$, which they used to derive the Page correction in the 2-R\'enyi entropy. 

Here, we take $O = \ketbra{\Psi_0}{\Psi_0}^{\otimes k}$. \textit{Assuming that the $k$-ETH holds for such operators}, Eq.~\eqref{eq:kETH} predicts that the eigenstate overlaps $\abs{c_E}^2$ in a microcanonical window follow a Porter-Thomas distribution, with a mean value $f(E)$.
\begin{equation}
    \abs{c_E}^2 \sim \text{PT}(f(E))~.
\end{equation}
We further expect $f(E)$ to be a smooth function of the energy. By dividing $\abs{c_E}^2$ by a fitted smooth function $f(E)$, we expect the \textit{rescaled overlap} to be distributed according to a PT distribution. We empirically find that $f(E)$ is often well described by a Boltzmann distribution $\exp(-\beta E)$. In Fig.~\ref{fig:eigenstate_overlaps_PT}(a-d) we demonstrate the emergence of the PT distribution in the populations over the eigenstates of the time-reversal-broken MFIM $H'$ [Eq.~\eqref{eq:MFIM_no_TRS}] for $N=12$ for: an infinite-temperature state $\ket{\Psi_0} =  |0\rangle^{\otimes N}$ (a,b), and a finite-temperature state $\ket{\Psi_0} = [\exp(i\theta X/2)|0\rangle]^{\otimes N}$ (c,d), with $\theta=0.3557$ chosen such that $\langle \Psi_0|H|\Psi_0\rangle = 0.5 N$. 

The presence of the PT distribution in the energy populations not only serves as an easily visualized consequence of the $k$-ETH, it also accounts for observations in the literature, such as universal ratios of 2 and 3 between the spectral form factor (SFF) and survival probability at late times~\cite{schiulaz2019thouless}, which follow from the second moment of the complex and real PT distributions. Specifically, the SFF is expressed as $|\text{tr}\left[\exp(-iHt)\right]|^2/D^2 = \sum_{E,E'} \exp(i(E'-E)t)/D^2$. At late times, only the $E=E'$ contribution remains and the SFF has a late time value of 1. In contrast, the survival probability of an initial state $|\Psi_0\rangle$ has expression $|\langle \Psi_0| \exp(-iHt)|\Psi_0\rangle |^2 = \sum_{E,E'} |c_E|^2 |c_{E'}|^2 \exp(i(E'-E)t)$. Again, at late times, only the $E=E'$ term survives and we obtain $\sum_E |c_E|^4$. From the discussion above, at infinite effective temperature, this sum is approximately $2/D$ or $3/D$, when $H$ respectively does not and does have time-reversal symmetry, and $|\Psi_0\rangle$ respects these symmetries.

\section{Outlook}
In this work, we have demonstrated that a novel type of maximum entropy principle arises in natural chaotic many-body systems. We substantiate our findings by studying two different ensembles of states,  each underlying the phenomena of Hilbert space ergodicity and deep thermalization, and find them well-described by ensembles with maximum entropy. Furthermore, we present the operational meaning,  information theoretic properties, and simple, experimentally-accessible signatures, i.e., Porter-Thomas distribution in certain rescaled observables, of our claimed universality.  

Our work extends our understanding of deep thermalization, most notably away from infinite temperature. We leave it to future work to generalize more features of infinite-temperature projected ensembles, such as the existence of distinct thermalization times for $k$-th moments~\cite{ippoliti2022dynamical,ippoliti2022solvable,chan2024projected} and the dynamics at short times~\cite{shrotriya2023nonlocality}.

The universality we establish in this work may be exploited for practical purposes. For example, in Ref.~\cite{mark2022benchmarking}, the presence of the PT distribution is leveraged to design a protocol to estimate fidelities in natural many-body systems. In companion experimental work, we explore deviations of the PT distribution in the presence of a bath. By understanding these deviations from universality, we design protocols to learn properties of the bath~\cite{shaw2024universal}. Further applications may include generalizing schemes that use aspects of quantum randomness to natural systems, such as randomized measurement protocols (including classical shadow tomography)~\cite{huang2020predicting,elben2023randomized,tran2023measuring,liu2023predicting} (with applications to characterizing thermal states~\cite{coopmans2023predicting}), as well as random number generation~\cite{bassirian2021certified,aaronson2023certified}.

Finally, our results should be seen as complementary to and distinct from conventional results such as the ETH. The ETH models some aspects of ergodic many-body systems with random matrix theory, which has been successful in describing striking features such as the level repulsion in eigenvalues. In recent years, there has been a growing body of work scrutinizing as well as generalizing the ETH, motivated by recent interest in quantities such as multi-time correlators. See e.g.~Refs.~\cite{foini2019eigenstate,kaneko2020characterizing,wang2022eigenstate,pappalardi2022eigenstate}. Our findings in the temporal and projected ensembles as well as our analysis of the $k$-ETH contribute to such work; providing a guiding principle to understand the emergence of effective randomness in many-body states. We look forward to such work and anticipate a broader understanding of complex quantum many-body systems.

\section*{Acknowledgements}
We thank David Huse, Wen Wei Ho, Daniel Ranard and Sa{\'u}l Pilatowsky-Cameo for insightful discussions on this work. We acknowledge support by the NSF QLCI Award OMA-2016245, the DOE (DE-SC0021951), the Institute for Quantum Information and Matter, an NSF Physics Frontiers Center (NSF Grant PHY-1733907),
the Center for Ultracold Atoms, an NSF Physics Frontiers Center (NSF Grant PHY-1734011), 
the DARPA ONISQ program (W911NF2010021), the AFOSR YIP (FA9550-19-1-0044), Army Research Office MURI program (W911NF2010136), NSF CAREER award 2237244 and NSF CAREER award 1753386. Support is also acknowledged from the U.S. Department of Energy, Office of Science, National Quantum Information Science Research Centers, Quantum Systems Accelerator. AE\ acknowledges funding by the German National Academy of Sciences Leopoldina under the grant number LPDS 2021-02 and by the Walter Burke Institute for Theoretical Physics at Caltech. FMS acknowledges support provided by the U.S.\ Department of Energy Office of Science, Office of Advanced Scientific Computing Research, (DE-SC0020290), by Amazon Web Services, AWS Quantum Program, and by the DOE QuantISED program through the theory consortium ``Intersections of QIS and Theoretical Particle Physics'' at Fermilab. JC acknowledges support from the Terman Faculty Fellowship at Stanford University.

\begin{widetext}
\begin{appendix}

\section{Proof of Theorem~\ref{thm:temp_ens} (The temporal ensemble is the random phase ensemble)}
\label{app:temp_ens_equals_random_phase_ens_proof}
\setcounter{theorem}{0}
\begin{theorem}
\label{thm:temp_ens_app}
Given an initial state $\ket{\Psi_0}$ and a Hamiltonian $H$, the infinite-time temporal ensemble is equal to the random phase ensemble:
\begin{equation}
\lim_{\tau \rightarrow \infty} \mathcal{E}_\text{Temp.}(\tau) = \mathcal{E}_\text{Rand.~Phase}~,    
\end{equation}
with fixed energy populations $|c_E|^2 = \abs{\braket{\Psi_0}{E}}^2$, if and only if $H$ satisfies all \textit{$k$-th no-resonance conditions} modulo degeneracies.    
\end{theorem}

\begin{proof}
We prove Theorem~\ref{thm:temp_ens} by matching the $k$-th moments of the temporal and random phase ensembles, and showing that they agree for all orders $k$. We enumerate the possible terms with \textit{multisets}.
\begin{definition}[Multisets]
We define the \textit{ordered multiset} $\{\alpha_j\}_{j=1}^k$ as the ordered list of elements $\alpha_j \in \{1,\dots,D\}$, with possibly repeated elements.
We denote $M_k(D)$ the set of all such multisets of length $k$ with symbols in $\{1,\dots,D\}$.
\end{definition}

We now compute the $k$-th moment of the temporal ensemble:
\begin{align}
    &\rho^{(k)}_\text{Temp.} \equiv \mathbb{E}_t \left[ U_t^{\otimes k} \ketbra{\Psi_0}{\Psi_0}^{\otimes k} U_{-t}^{\otimes k}\right] \label{eq:temp_kth_moment}\\
    &= \sum_{\pmb{\alpha},\pmb{\beta} \in M_k(D)} \left[\bigotimes_{j=1}^k \ketbra{E_{\alpha_j}}{E_{\alpha_j}}\right]  \ketbra{\Psi_0}{\Psi_0}^{\otimes k} \left[\bigotimes_{j=1}^k \ketbra{E_{\beta_j}}{E_{\beta_j}}\right] \delta\left(\sum_{j=1}^k E_{\alpha_j} - \sum_{j=1}^kE_{\beta_j}\right),\label{eq:temp_kth_moment_time_integration}\\
    &= \sum_{\pmb{\alpha} \in M_k(D)} \bigg[\sum_{\sigma \in \pi(\pmb{\alpha})} \ketbra{E_{\alpha_1}, \cdots, E_{\alpha_k}}{E_{\alpha_1}, \cdots, E_{\alpha_k}}  \ketbra{\Psi_0}{\Psi_0}^{\otimes k} \ketbra{E_{\sigma({\alpha_1})}, \cdots, E_{\sigma({\alpha_k})}}{E_{\sigma({\alpha_1})}, \cdots, E_{\sigma({\alpha_k})}} \bigg],  \label{eq:multiset_line}\\
    &= \sum_{\pmb{\alpha} \in M_k(D)} \sum_{\sigma \in \pi(\pmb{\alpha})}\bigotimes_{j=1}^k \abs{\braket{E_{\alpha_j}}{\Psi_0}}^2  \ketbra{E_{\alpha_j}}{E_{\sigma({\alpha_j})}}\\
    &= \rho_d^{\otimes k} \sum_{\pmb{\alpha} \in M_k(D)} \sum_{\sigma \in \pi(\pmb{\alpha})}\bigotimes_{j=1}^k \ketbra{E_{\alpha_j}}{E_{\sigma({\alpha_j})}}~,
    \label{eq:Ham_k_design}
\end{align}
where $U_t \equiv \exp(-iHt)$, and $\mathbb{E}_t[\cdot] = \lim_{T\rightarrow \infty} \frac{1}{T} \int_0^T \cdot~dt$ is the infinite-time average, $\ket{E}$ are energy eigenstates of the Hamiltonian, and $\rho_d \equiv \sum_E \abs{\braket{E}{\Psi_0}}^2 \ketbra{E}{E}$ is the diagonal ensemble, prominently featured in literature on many-body thermalization~\cite{deutsch2018eigenstate}. 
In going from \eqref{eq:temp_kth_moment} to \eqref{eq:temp_kth_moment_time_integration} we have performed an infinite-time integration to obtain a delta function, and in going to \eqref{eq:multiset_line}, we have used the $k$-th no-resonance condition. 

We then compute the $k$-th moment of the random phase ensemble:
\begin{align}
    \rho^{(k)}_\text{Rand. phase} &= \sum_{\pmb{\alpha},\pmb{\beta} \in M_k(D)} \bigotimes_{j=1}^k c_{\alpha_j} c_{\beta_j}^* \ketbra*{\alpha_j}{\beta_j} \underbrace{\prod_j \left(\frac{1}{2\pi}\int_0^{2\pi} d\theta_j \right) \exp[i(\theta_{\alpha_j} - \theta_{\beta_j})]}_{\delta_{(\alpha_j),(\beta_j)}}\\
    &= \sum_{\pmb{\alpha} \in M_k(D)}  \sum_{\sigma \in \pi(M)} \bigotimes_{j=1}^k \abs{c_{\alpha_j}}^2  \ketbra*{\alpha_j}{\sigma({\alpha_j})}
\end{align}
Therefore, the $k$-th moments of the temporal and random phase ensembles are equal. Identifying a distribution from its moments is known as the ``moment problem." For complex multivariate (vector-valued) distributions $P(z_1,\cdots,z_D)\equiv P(\pmb{z})$, the complex Carleman's condition gives a sufficient condition to the moment problem. That is, define the moments $m^{(k)}_j \equiv \int dz P(\pmb{z}) |z_j|^{2k}$. If the following sum diverges:
\begin{equation}
    \forall j~,~\sum_{k=1}^\infty \left(m^{(k)}_j\right)^{-1/2k} = \infty~,
\end{equation}
then the sequence of moments $(m^{(k)}_j)$ uniquely defines the probability distribution~(Thm.~15.11 of Ref.~\cite{schmudgen2017}). In the random phase ensemble, the complex Carleman condition is satisfied. This allows us to identify the temporal and random phase ensembles, if the $k$-th no-resonance condition is satisfied.

In the other direction, if the $k$-th no-resonance condition is \textit{not} satisfied, the delta function in Eq.~\eqref{eq:temp_kth_moment_time_integration} will give additional contributions. Therefore, the $k$-th moments of the temporal and random phase ensembles do not agree, giving the desired statement.
\end{proof}

\section{$k$-th twirling identity}\label{app:Ham_twirling}
The above computation of $\rho^{(k)}_\text{Temp.}$ in Appendix~\ref{app:temp_ens_equals_random_phase_ens_proof} can be generalized into a statement about the \textit{twirling channel} on any operator $A$ that acts on $k$ copies of the Hilbert space. Note, however, that unlike Theorem~\ref{thm:temp_ens}, we now require the full $k$-th no-resonance conditions and cannot disregard degeneracies. As an illustration, we state and explicitly prove the result for $k=2$. This identity originally appeared in the SM of Ref.~\cite{mark2022benchmarking}.

\begin{lemma}[Second Hamiltonian twirling identity] If a Hamiltonian $H$ satisfies the second no-resonance condition, then for any operator $A$ acting on $\mathcal{H}^{\otimes 2}$, the following identity holds:
\begin{align}
    &\mathbb{E}_t \left[ U_t^{\otimes 2} A U_{-t}^{\otimes 2}\right] =  \mathcal{D}^{\otimes 2}\big[A\big]  \mathbb{I} + \mathcal{D}^{\otimes 2}\big[A \mathcal{S}\big] \mathcal{S} - \sum_E \ketbra{E}{E}^{\otimes 2} A \ketbra{E}{E}^{\otimes 2}~, \label{eq:Ham_2_design}
\end{align}
where $U_t \equiv \exp(-iHt)$, and $\mathbb{E}_t[\cdot] = \lim_{T\rightarrow \infty} \frac{1}{T} \int_0^T \cdot~dt$ is the infinite-time average, and $\ket{E}$ are energy eigenstates of the Hamiltonian. $\mathbb{I}$ and $\mathcal{S}$ are the identity and swap operators on $\mathcal{H}^{\otimes 2}$, defined as $\mathbb{I}(\ket{\Psi_1} \otimes \ket{\Psi_2}) = \ket{\Psi_1} \otimes \ket{\Psi_2}$ and $\mathcal{S}(\ket{\Psi_1} \otimes \ket{\Psi_2}) = \ket{\Psi_2} \otimes \ket{\Psi_1}$ for any $\ket{\Psi_1},\ket{\Psi_2}$. Finally, $\mathcal{D}^{\otimes 2}$ is the dephasing channel in the energy eigenbasis: $\mathcal{D}^{\otimes 2 }\big[A\big] \equiv \sum_{E,E'} \ketbra{E,E'}{E,E'}A\ketbra{E,E'}{E,E'}$.
\end{lemma}

\begin{proof}
\begin{align}
    &\mathbb{E}_t \left[U_t^{\otimes 2} A  U_{-t}^{\otimes 2} \right] =  \lim_{T\rightarrow \infty}\frac{1}{T}\int^T_0 dt~U_t^{\otimes 2} A  U_{-t}^{\otimes 2} \\
    &=  \sum_{E_1,E_2,E_3,E_4} \left[\lim_{T\rightarrow \infty}\frac{1}{T}\int^T_0 dt~e^{-i(E_1-E_2+E_3-E_4)t}\right] \ketbra{E_1,E_3}{E_1,E_3}A \ketbra{E_2,E_4}{E_2,E_4} \\
    &=  \sum_{E_1,E_2,E_3,E_4} \delta(E_1-E_2+E_3-E_4)~ \ketbra{E_1,E_3}{E_1,E_3}A \ketbra{E_2,E_4}{E_2,E_4}\\
    &\overset{\text{$2$-NR}}{=}  \sum_{E,E'} \bigg[ \ketbra{E,E'}{E,E'} A  \ketbra{E,E'}{E,E'} + \ketbra{E,E'}{E,E'} A  \ketbra{E',E}{E',E} \bigg] - \sum_E \ketbra{E,E}{E,E} A  \ketbra{E,E}{E,E}~, \label{eq:lemma_2}
\end{align}
where ``2-NR" indicates the use of the second no-resonance condition. The second no-resonance condition states that the delta function $\delta(E_1-E_2+E_3-E_4)$ is non-zero in only two cases: $E_1 = E_2,~E_3=E_4$ or $E_1 = E_4,~E_3=E_2$, which give the two terms above. The final term in Eq.~\eqref{eq:lemma_2} corresponds to $E_1=E_2=E_3=E_4$. 
This special case belongs to the above two terms simultaneously. It is double counted in the first two terms, hence is subtracted.
\end{proof}

\begin{lemma}[$k$-th Hamiltonian twirling identity] \label{lemma:Hamiltonian_k_design}
If the Hamiltonian satisfies the $k$-th no-resonance condition, its $k$-th twirling channel acting on any operator $A$ on $k$ copies of the Hilbert space $\mathcal{H}^{\otimes k}$ is:
\begin{align}
    \mathbb{E}_t \left[ U_t^{\otimes k} A U_{-t}^{\otimes k}\right]
    &= \sum_{M = (E_1,\cdots, E_k)} \bigg[\sum_{\sigma \in \pi(M)} \ketbra{E_1, \cdots, E_k}{E_1, \cdots, E_k} A  \ketbra{E_{\sigma(1)}, \cdots, E_{\sigma(k)}}{E_{\sigma(1)}, \cdots, E_{\sigma(k)}} \bigg],  \label{eq:multiset_twirling}\\
    &= \sum_{\sigma \in S_k}\mathcal{D}^{\otimes k}\big[A\text{Perm}(\sigma^{-1})\big] \text{Perm}(\sigma) + \mathrm{E.C.} \label{eq:Ham_k_twirling}
\end{align}
where $\text{Perm}(\sigma)$---with $\sigma \in S_k$---is a permutation operator acting on the $k$-copy Hilbert space: 
\begin{align}
    \text{Perm}(\sigma) \equiv \sum_{\{E_i\}} \ketbra*{E_{1},E_{2}, \dots, E_{k}}{E_{\sigma(1)},E_{\sigma(2)}, \dots, E_{\sigma(k)}}~,
\end{align} 
and $\mathcal{D}^{\otimes k}$ is the dephasing channel in the energy basis, $\mathcal{D}^{\otimes k}\big[A\big] \equiv \sum_{\{E_i\}} \ketbra*{\{E_i\}}{\{E_i\}} A \ketbra*{\{E_i\}}{\{E_i\}}$.

In~\eqref{eq:multiset_twirling}, the sum is taken over multisets $M$. If $M$ has $r$ unique indices with multiplicities $n_1,\cdots,n_r$ (such that $\sum_{i=1}^n n_r = k$), its group of permutations $\pi(M)$ is a subgroup of the symmetric group $S_k$ with $k!/\prod_{i=1}^n n_i!$ elements. In~\eqref{eq:Ham_k_twirling}, we re-write Eq.~\eqref{eq:multiset_twirling} in terms of permutations, and additional terms $\text{E.C.}$ associated with multiple counting (similar to the $k=2$ case). At this level, it is difficult to provide a useful bound on the size of $\text{E.C}$. However, bounds may be proven in specific contexts, such as Theorem~\ref{thm:temp_ens_PT}. All terms denoted ``E.C."~can be explicitly derived from the twirling identities.
\end{lemma}

The proof of the $k$-th twirling identity is analogous to that of the second twirling identity: in this case, there are $k!$ permutations relevant for the $k$-th no resonance theorem, along with more terms than the simplest one in the $k=2$ case. 
As an example, for $k=3$, in addition to the terms associated with cases where two energies are the same $E_i = E_j$, there is a higher order term arising when all three energies are the same: $E_1 = E_2 = E_3$.

\section{Entropy of temporal ensemble}
\label{app:temp_ens_entropy}
In Section~\ref{sec:temporal_ensemble}, we stated that the ensemble entropy of the temporal ensemble, a negative quantity which is the negative of its relative entropy to the Haar ensemble, is given by [Eq.~\eqref{eq:temporal_ensemble_entropy}]
\begin{align}
    \text{Ent}(\mathcal{E}_\text{Temp.}) =\sum_E \log(D\abs{c_E}^2)
    +\lim_{\delta\rightarrow 0} D \left(\log(\delta/2)-1\right)~.
\end{align}
Here, we derive this result. Note that any ensemble in the constrained space can be described by the joint distribution $P(\{\theta_E\})$. We emphasize that the ensemble entropy is the differential entropy of the \textit{distribution of states} over Hilbert space. The ensemble entropy for any ensemble of states with constrained magnitudes $\abs{c_E}^2$ is formally negative infinite since they lie on a sub-dimensional (measure zero) manifold of the Hilbert space. However, we can still extract meaningful information from the ensemble entropy as follows. We may take the probability distribution $P(\Psi)$ on the constrained space as the limit of probability distributions with independent phase and magnitude components: $P(\Psi) = P(\{\abs{\Psi_E}\})P(\{\theta_E\})$. This separates the 
entropy into contributions from the phase and magnitude degrees of freedom:
\begin{align}
    &-\mathcal{D}(\mathcal{E} \Vert \text{Haar}) =  -\int d\Psi~P(\Psi) \log\left(\frac{P(\Psi)}{P_\text{Haar}(\Psi)}\right) 
    = S_\theta + S_{\abs{\Psi_E}^2}~.
\end{align}

The phase contribution is
\begin{equation}
    S_\theta = -\prod_E\int_0^{2\pi} d\theta_E P(\{\theta_E\}) \log[(2\pi)^D P(\{\theta_E\})]~,
\end{equation}
which is uniquely maximized when the phases are uniformly distributed, taking value 0.

The magnitude component is divergent:
However, if we replace the delta function with an integral over the thin shell $\abs{\Psi_E}^2 \in [\abs{c_E}^2(1-\delta),\abs{c_E}^2(1+\delta)]$, we obtain a constant contribution in addition to the divergent one.
\begin{align}
    S_{\abs{\Psi_E}^2} &= -\lim_{\delta\rightarrow 0}\prod_E \int_{\abs{c_E}^2(1-\delta)}^{\abs{c_E}^2(1+\delta)} \frac{d\abs{\Psi_E}^2}{2\abs{c_E}^2 \delta}  \log(\frac{ 1/(\abs{c_E}^2\delta)}{2D \exp(-D\abs{\Psi_E}^2)}) = \lim_{\delta\rightarrow 0}\left[\sum_E (\log(\abs{c_E}^2))-D\abs{c_E}^2 + D \log(\delta D/2)\right] \nonumber\\
    &=\sum_E \log(\abs{c_E}^2)+\lim_{\delta\rightarrow 0} D \left(\log(\delta D/2)-1\right)~.
\end{align}
We interpret the constant contribution to the entropy as the volume of Hilbert (sub)space spanned by the temporal ensemble. This subspace is in fact a hypertorus, with radii given by ${\abs{c_E}}$. In the limit of large Hilbert space dimension $D$, the volume of this hypertorus is approximately $\prod_E \abs{c_E}$. Therefore our measure of entropy is quantitatively equivalent to the volume of the Hilbert space.

\section{Proof of Theorem~\ref{thm:temp_ens_product_form} (Asymptotic product formula)} \label{app:temporal_ensemble_error}
\begin{theorem}
The random phase ensemble asymptotically approaches the \textit{product form}:
\begin{align}
    \rho^{(k)}_\text{Temp.} = \rho_d^{\otimes k} \sum_{\sigma \in S_k} \text{Perm}(\sigma) + O(\tr(\rho_d^2))
\end{align}
Here, $\rho_d \equiv \sum_E \abs{c_E}^2 \ketbra{E}{E}$ is the \textit{diagonal ensemble}, a common construction in the thermalization literature~\cite{rigol2008thermalization,polkovnikov2011nonequilibrium}, and its purity $\tr(\rho_d^2)$ is typically exponentially small in total system size. 
\end{theorem}

\begin{proof}
    We want to show that:
\begin{align}   \rho^{(k)}_\text{Temp.} &= \rho_d^{\otimes k}\sum_{\pmb{\alpha} \in M_k(D)} \sum_{\sigma \in \pi(\pmb{\alpha})}\bigotimes_{j=1}^k   \ketbra{E_{\alpha_j}}{E_{\sigma({\alpha_j})}}\overset{D\rightarrow \infty}{\longrightarrow} \rho_d^{\otimes k} \sum_{\sigma \in S_k} \text{Perm}(\sigma) = \rho_d^{\otimes k} \sum_{\pmb{\alpha}\in M_k(D)} \sum_{\sigma \in S_k}\bigotimes_{j=1}^k \ketbra{E_{\alpha_j}}{E_{\sigma({\alpha_j})}}
\end{align}
Here, we derive the bound on the error
\begin{equation}
    \Vert\rho^{(k)}_\text{Temp.} - \rho_d^{\otimes k} \sum_{\sigma \in S_k} \text{Perm}(\sigma) \Vert_* \leq \frac{k!}{2} \tr(\rho_d^2),
\end{equation}
where $\Vert \cdot \Vert_*$ is the trace norm, or the sum of absolute values of the eigenvalues. To do so, we note that
\begin{equation}
    \rho_d^{\otimes k} \sum_{\sigma \in S_k} \text{Perm}(\sigma) - \rho^{(k)}_\text{Temp.} = \rho_d^{\otimes k}\left[\sum_{\pmb{\alpha} \in M_k(D)} \left(\sum_{\sigma \in S_k} - \sum_{\sigma \in \pi(\pmb{\alpha})}\right) \bigotimes_{j=1}^k   \ketbra{E_{\alpha_j}}{E_{\sigma({\alpha_j})}} \right]
\end{equation}
We then divide the Hilbert space $\mathcal{H}^{\otimes k}$ into classes $s$ of multisets $\pmb{\alpha}$, grouped by their degeneracy structure: the number of degenerate indices $\alpha_j$ and their multiplicity. For example, $s =()$ indicates the class of multisets with no degeneracies, $s=(2)$ indicates the presence of 1 degeneracy of multiplicity 2, and $s=(2,2)$ indicates two degeneracies of multiplicity 2, respectively. The size of the permutation group $\pi(\pmb{\alpha})$ is uniquely determined by the class $s(\pmb{\alpha})$, and is a divisor of $|S_k|=k!$.
Therefore, 
\begin{align}
    \rho_d^{\otimes k} \sum_{\sigma \in S_k} \text{Perm}(\sigma) - \rho^{(k)}_\text{Temp.} 
    &= \rho_d^{\otimes k}\left[\sum_{s \subset M_k(D)} \sum_{\pmb{\alpha} \in s} \left(\sum_{\sigma \in S_k} - \sum_{\sigma \in \pi(\pmb{\alpha})}\right) \bigotimes_{j=1}^k   \ketbra{E_{\alpha_j}}{E_{\sigma({\alpha_j})}} \right]\\
    &= \sum_{s \subset M_k(D)} \rho_d^{\otimes k}
    P_s \sum_{\sigma \in S_k}\text{Perm}(\sigma) \label{eq:sum_over_classes}
\end{align}
Where $P_s$ is a projector onto the space of multisets in the class $s$. Each term in the sum over $s$, $\rho_d^{\otimes k}
    P_s \sum_{\sigma \in S_k}\text{Perm}(\sigma)$,is positive semi-definite. We then use the triangle inequality to bound the trace distance, noting that the term from the class $s = ()$ gives 0 contribution. 
In our sum \eqref{eq:sum_over_classes} over multisets with at least one degeneracy, we can replace $\rho_d^{\otimes k}$ with $\rho_d^{(2)} \otimes \rho_d^{\otimes (k-2)}$, where $\rho_d^{(2)} \equiv \sum_E \abs{c_E}^4 \ketbra{E,E}{E,E}$. Finally, the projector $P_s$ only decreases the trace distance and can be omitted.
\begin{align}
    \Vert \rho_d^{\otimes k} \sum_{\sigma \in S_k} \text{Perm}(\sigma) - \rho^{(k)}_\text{Temp.}  \Vert_* &\leq \sum_{s \neq ()} \Vert \left(\rho_d^{(2)} \otimes \rho_d^{\otimes (k-2)}\right) \sum_{\sigma \in S_k}\text{Perm}(\sigma) \Vert_*\\
    &= \sum_{s \neq ()} \tr\left[\left(\rho_d^{(2)} \otimes \rho_d^{\otimes (k-2)}\right) \sum_{\sigma \in S_k}\text{Perm}(\sigma) \right] \label{eq:temporal_product_form_bound}\\
    &\leq \sum_{s \neq ()} k! \tr(\rho_d^2) = (p(k)-1) k! \tr(\rho_d^2) \leq k!\exp(\pi \sqrt{\frac{2k}{3}})\tr(\rho_d^2) \label{eq:temporal_product_form_final_bound}
\end{align}
where $p(k)$ is the \textit{partition number} of $k$ (the number of possible classes $s$), which is upper bounded by $\exp(\pi \sqrt{\frac{2k}{3}})$ \cite{erdos1942elementary}. In going from Eq.~\eqref{eq:temporal_product_form_bound} to Eq.~\eqref{eq:temporal_product_form_final_bound} we have used the fact that 
\begin{equation}
    \tr\left[\left(\rho_d^{(2)} \otimes \rho_d^{\otimes (k-2)}\right) \sum_{\sigma \in S_k}\text{Perm}(\sigma) \right]  = \tr(\rho_d^2) + (k-1) \tr(\rho_d^3) + \binom{k-1}{2} \tr(\rho_d^2)^2 + \cdots \leq k! \tr(\rho_d^2)~,
\end{equation}
this bound is weak since we expect all terms to be exponentially smaller (in system size) than the first. This implies that the $O(k!)$ prefactor is likely not tight in terms of scaling with $k$. However, we will typically restrict our attention to small values of $k$, and are more interested in the factor of $\tr(\rho_d^2)$. In typical systems, we expect $\tr(\rho_d^2) \sim D^{-1}$ to exponentially decrease as the size of the many-body system increases. 
\end{proof}

\section{Convergence of finite-time temporal ensembles: Further numerical results} \label{app:finite_time_ensembles}

\begin{figure}
    \centering
    \includegraphics[width=\textwidth]{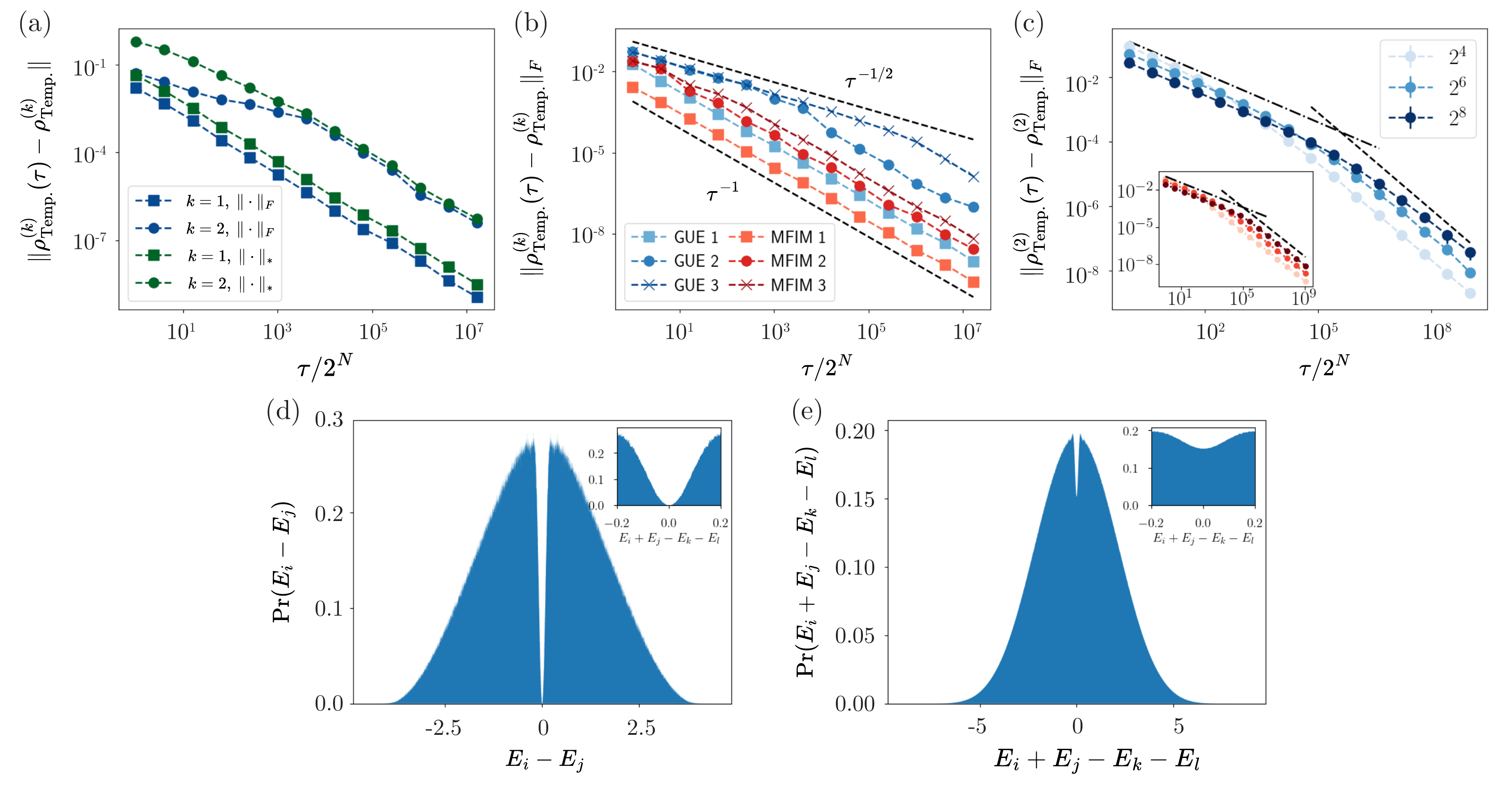}
    \caption{Convergence of finite time ensembles. (a) Comparison of convergence of first ($k=1$, squares) and second moment ($k=2$, circles) in trace (green) and Frobenius distance (blue) for a temporal ensemble generated by a single, randomly drawn GUE Hamiltonian with Hilbert space dimension $D=2^6$. (b) Convergence in Frobenius distance of the first three moments $k=1,2,3$ (squares, circles, crosses) for single, randomly drawn GUE Hamiltonian with Hilbert space dimension $D=2^6$ (blue) and the MFIM with $N=6$ spins (red). (c) Ensemble averaged Frobenius distance of the temporal ensembles generated by a randomly drawn GUE Hamiltonians in various Hilbert space dimension $D=2^4,2^6,2^8$. Mean is takes over 10000 ($D=2^4,2^6$) and 2000 random instances ($D=2^8$). Error bars denote two standard errors of the mean, 2000. Inset shows the median of the Frobenius distances of the same data. (d,e) Histogram of gaps and differences of gaps accumulated over 50000 randomly chosen GUE Hamiltonian with Hilbert space dimension $D=2^4$. Indices $i,j,k,l$ take all values in $1,\dots, D$ and we excluded the (trivially vanishing) diagonal elements  $i=j$  in (d) and $i=k, j=l$ and $i=l, j=k$ in (e). The insets enlarge the regions close to zero.}
    \label{fig:app_finiteT}
\end{figure}

Here, we present additional numerical results investigating the convergence of finite time temporal ensembles discussed in Section~\ref{subsec:finite_temp_temp_ens}. 

In Fig.~\ref{fig:app_finiteT}(a), we compare convergence of first and second moment of the finite time temporal ensemble in trace and Frobenius distance. We consider a single Hamiltonian drawn at random from the GUE ensemble in a system with Hilbert space dimension $d=2^6$. For both, first and second moment, we observe the same behavior at long times: trace and Frobenius distance decay as $\tau^{-1}$, where $\tau$ is the interval length. In contrast, the intermediate $\tau^{-1/2}$ decay of the Frobenius distance and its pronounced cross-over to the late time $\tau^{-1}$ appears to be absent when considering the trace distance. This suggests, that while the  early and intermediate time behavior of trace and Frobenius distance can differ, the Frobenius distance is a quantitative proxy for the trace distance at long times.

In Fig.~\ref{fig:app_finiteT}(b), we complement the results of the main text by comparing convergence of the $k$-th moments of the finite time temporal for different $k=1,2,3$. We consider a single Hamiltonian drawn at random from the GUE ensemble in a system with Hilbert space dimension $D=2^6$ and the MFIM with $N=6$ spins. For computational simplicity, we focus on the Frobenius norm. Independent of the value of $k$, we observe a late time decay $\propto \tau^{-1}$. As discussed in the main text [c.f.  Eq.~\eqref{eq:finite_time_ensemble}], we attribute this to the presence of a finite minimal gap $E_i-E_j$ (difference of gaps for $k\leq 2$) for an individual Hamiltonian and the late time decay of the $\textrm{sinc}$ function.

In Fig.~\ref{fig:app_finiteT}(c), we investigate the convergence of ensemble-averaged distances. Notably, we find that for the first moment, $k=1$, the GUE ensemble mean closely resembles the behavior of individual instances and decays as $\tau^{-1}$ (not shown). For higher moments, $k=2$, we find that, at early times,  ensemble average distances decays as $\tau^{-1/2}$ whereas at very late times we observe a decay consistent with $\tau^{-1}$. The cross-over between these two regimes is, however, less sharply  pronounced as for single fixed Hamiltonians [as visualized by the ensemble median, representing a typical instance,  in inset of Fig.~\ref{fig:app_finiteT}(c)].  We relate these findings to statistical properties  of  the energy eigenlevels of GUE Hamiltonians. It is well known that the statistics of eigenvalues exhibits (strong) level repulsion, i.e., the probability to find two distinct eigenvalues close to each other $E_i-E_j \approx 0$ is vanishing. 
In contrast, the numerical results presented Figs.~\ref{fig:app_finiteT}(d, e) suggest the gaps $E_i-E_j+E_k-E_l$ of eigenvalues exhibit only weak repulsion, i.e., the probability of finding two distinct gaps close to each other $E_i-E_j+E_k-E_l \approx 0$ is suppressed but finite. We leave a detailed investigation of such higher-order gap repulsion and its consequences for the late time decay of the ensemble-averaged Frobenius (and trace) distance to future study (see also Section \ref{subsec:conv}).

\section{Convergence of finite-time temporal ensembles: Analytical results from random matrix theory} \label{app:finite_time_ensembles}
Here we prove the $\tau^{-1}$ scaling of the convergence of the finite-time temporal ensemble for Hamiltonians from the Gaussian Unitary Ensemble (GUE), discussed in Section~\ref{subsec:finite_temp_temp_ens}. We focus on the (ensemble-averaged) squared Frobenius norm $\Vert \cdot \Vert_F$ for its analytic tractability. Note that, since we average over random Hamiltonians $H$, the initial state can be fixed to be an arbitrary state.
\begin{align}
\mathbb{E}_{H\sim \text{GUE}}\left[ \Vert  \rho^{(k)}_\text{Temp.}(\tau) - \rho^{(k)}_\text{Temp.} \Vert_F^2\right] &= \mathbb{E}_{H\sim \text{GUE}}\left[ \tr[\left(\rho^{(k)}_\text{Temp.}(\tau)\right)^2] - 2 \tr[\rho^{(k)}_\text{Temp.}(\tau) \rho^{(k)}_\text{Temp.}] + \tr[\left(\rho^{(k)}_\text{Temp.}\right)^2]\right] \\
&= \mathbb{E}_{H\sim \text{GUE}}\left[ \tr[\left(\rho^{(k)}_\text{Temp.}(\tau)\right)^2] -  \tr[\left(\rho^{(k)}_\text{Temp.}\right)^2]\right]
\end{align}
where we used that the long time limit average moment does not evolve in time.
\begin{align}
    \tr[\rho^{(k)}_\text{Temp.}(\tau) \rho^{(k)}_\text{Temp.}] = \tr[\int_{-\tau}^\tau \frac{\text{d}t}{2\tau}  \left(e^{-iHt}\ket{\Psi_0}\bra{\Psi_0}e^{iHt}\right)^{\otimes k} \rho^{(k)}_\text{Temp.}] =   \bra{\Psi_0}^{\otimes k} \rho^{(k)}_\text{Temp.}\ket{\Psi_0}^{\otimes k} = \tr[\rho^{(k),2}_\text{Temp.}].
\end{align}
To carry out the average over the GUE, we decompose $H=U\Lambda U^\dagger$ with $\Lambda$ a diagonal matrix of eigenvalues of $H$ and $U$ the unitary matrix formed by its eigenvectors. Importantly, for $H$ from GUE,  $\Lambda$ and $U$ are statistically independent and $U$ is a random unitary matrix distributed according to the Haar measure on the unitary group. We can thus evaluate the average over the eigenvectors $U$  separately using Weingarten calculus. We find for the first term
\begin{align}
    \mathbb{E}_{H\sim \text{GUE}}\left[ \tr[\rho^{(k),2}_\text{Temp.}(\tau)]\right] =& \frac{ 1}{ 4\tau^2} \int_{0}^{2\tau}  \! \! \text{d}\bar t \int_{0}^{2\bar t} \! \text{d}t'     \int \! \text{d}\Lambda  \,P(\Lambda)  \int \! \text{d}U  \tr\left [\ket{\Psi_0}\bra{\Psi_0} ^{\otimes 2k} U^{\otimes 2k} \left(e^{-i\Lambda t'}\otimes e^{+i\Lambda t'}\right)^{\otimes k}U^{\dag,\otimes 2k} \right ] \nonumber \\
     =&  \frac{ 1}{ 4\tau^2} \int_{0}^{2\tau}  \! \! \text{d}\bar t \int_{0}^{2\bar t} \! \text{d} t'   \frac{1}{D(D+1)\cdots (D+k-1)}\int \! \text{d}\Lambda  \, P(\Lambda)  \sum_{\pi\in S_{2k}} \tr\left[W_\pi\left(e^{-i\Lambda t'}\otimes e^{+i\Lambda t'}\right)^{\otimes k}\right].
\end{align}
where $P(\Lambda)$ denotes the joint probability distribution of the eigenvalues of Hamiltonians $H$ from GUE, $ \int \! \text{d}U$ denotes the Haar integral of the unitary group, $W_\pi $ with $\pi \in S_{2k}$ are $2k$ permutation operators acting on the $2k$-copy space and $S_{2k}$ denotes the symmetric group of order $2k$. In the case $k=1$, we can simplify the expressions further. We find
\begin{align}
    \int \! P(\Lambda)   \text{d}\Lambda   \sum_{\pi\in S_{2}} \tr\left[W_\pi \left(e^{-i\Lambda\tau}\otimes e^{+i\Lambda\tau}\right)\right]  = \frac{1}{D(D+1)} \left[K(\tau) + D\right] \nonumber
\end{align}
where 
\begin{align}
       K(\tau) = \mathbb{E}_{H\sim \text{GUE}}\left[ \abs{\tr(e^{-iH\tau })}^2\right] =&  \int \! P(\Lambda)   \text{d}\Lambda ~\abs{\tr( e^{-i\Lambda\tau})}^2 \nonumber
\end{align}
is the \textit{spectral form factor} averaged over Hamiltonians $H$. Furthermore, $\rho_{H}^{(1)}=\sum_E \abs{\braket{E}{\Psi_0}}^2 \ketbra{E}$ with $\{\ket{E}\}$ the eigenvectors of $H$. Thus, we have $\tr[\varrho_H^{(1),2}] = \sum_E \abs{\braket{E}{\Psi_0}}^4$, which averaged over $H$ from GUE,
\begin{align}
   \mathbb{E}_{H\sim \text{GUE}} \left[\sum_E |\braket{E}{\Psi_0}|^4\right] = \int \! \text{d}U \left[\sum_{i=1}^D |\braket{i|U}{\Psi_0}|^4\right]  =  \frac{2}{D+1}
\end{align}
Finally, we obtain
\begin{align}
    \mathbb{E}_{H\sim \text{GUE}}\left[ \Vert  \rho^{(1)}_H(\tau) - \rho^{(1)}_H \Vert_F^2\right] = \frac{ 1}{ 4\tau^2} \int_{0}^{2\tau}  \! \! \text{d}\bar t \int_{0}^{2\bar t} \! \text{d}t'   \frac{K(t')-D} {D(D+1)} \sim \frac{D}{\tau^2} + \mathcal{O}(\tau^{-4})\label{eq:frob_dist_final}
\end{align}
Using the improved approximation~\cite{liu2018spectral} for the spectral form factor $K(\tau)$ and an asymptotic expansion of the occuring Bessel function \cite{abramowitz_stegun_1964} gives a $\tau^{-2}$ scaling of the squared Frobenius norm. We note that if the standard box approximation is used, Eq.~\eqref{eq:frob_dist_final} does not converge to 0. As we detail below, an alternative approach is to compute the convergence for the GUE ensemble for a finite dimension $D$. The finite $D$ expressions, in fact, do not require any regularization, and the large $D$ limit can be safely taken on the final result. 

\subsection{Convergence for GUE ensemble and finite dimension}
\label{subsec:conv}

Here we show how to study the convergence to the temporal ensemble for Hamiltonians from GUE ensembles of finite dimension $D$. These finite $D$ expressions do not suffer from the problems of the infinite-$D$ spectral form factor, which has to be properly regularized. Therefore, we can compute the distance from the temporal ensemble for finite $D$, and take the $D\rightarrow \infty$ limit at the end.

We are interested in the average of the following distance for $H$ belonging to the GUE ensemble:
\begin{equation}
    \mathrm{Tr} [(\rho^{(k)}_{\mathrm{Temp.}}(\tau)-\rho^{(k)}_{\mathrm{Temp.}}(\infty))^2]=\sum_{\{\alpha_j\}}\; \sum_{\{\beta_j\}\notin \mathrm{Perm}(\{\alpha_j\})}\mathrm{sinc}^2 \left(\sum_j(E_{\alpha_j}-E_{\beta_j})\tau\right)\prod_j |c_{\alpha_j}|^2|c_{\beta_j}|^2.
\end{equation}
Note that the sum is over the sequences of indices $\{\alpha_j\}, \{\beta_j\}$ under the condition that the two sequences are not related by a permutation. 
Let us consider, as a first example, the case $k=1$. We have:
\begin{equation}
    \mathbb{E}_{H \sim \text{GUE}}\left[\mathrm{Tr} [(\rho^{(1)}_{\mathrm{Temp.}}(\tau)-\rho^{(1)}_{\mathrm{Temp.}}(\infty))^2]\right]=\sum_{i}\sum_{j\neq i } \mathbb{E}_{H \sim \text{GUE}}\left[\mathrm{sinc}^2\left((E_{i}-E_{j})\tau\right)\right]\mathbb{E}_{H \sim \text{GUE}}\left[ |c_{i}|^2|c_{j}|^2\right],
\end{equation}
where we used that the eigenvalues and the eigenvectors of $H$ are independently distributed. 

For the eigenvectors we get:
\begin{equation}
    \mathbb{E}_{H \sim \text{GUE}}\left[ |c_{i}|^2|c_{j}|^2\right]=\frac{1}{D(D+1)}.
\end{equation}

The convergence in time depends on the following average over the distribution of the eigenvalues:
\begin{align}
    \mathbb{E}_{H \sim \text{GUE}}\left[\mathrm{sinc}^2\left((E_{i}-E_{j})\tau\right)\right] &= \int \mathrm{d}E_1 \dots \mathrm{d}E_{D}  \,\mathrm{sinc}^2 \left[(E_1-E_2)\tau\right] p^D(E_1,\dots,E_{D})\nonumber\\
    &= \int \mathrm{d}u \,\mathrm{sinc}^2 \left(u \tau\right) q_{1|1}^D(u).
\end{align}
Here $p^D(E_1,\dots,E_{D})$ is the joint probability distribution of the (unordered) eigenvalues $E_1,\dots E_D$,
 and we defined $q_{1|1}^D(u)$ as the probability distribution of the energy difference $E_1-E_2$:

\begin{equation}
    q_{1|1}^D(u)=\int \int \mathrm{d}E_1 \dots \mathrm{d}E_{D}  \,\delta (E_1-E_2-u) p^D(E_1,\dots,E_{D}).
\end{equation}

Hence we get:

\begin{equation}
    \mathbb{E}_{H \sim \text{GUE}}\left[\mathrm{Tr} [(\rho^{(1)}_{\mathrm{Temp.}}(\tau)-\rho^{(1)}_{\mathrm{Temp.}}(\infty))^2]\right]=\frac{D-1}{D+1}\int \mathrm{d}u \,\mathrm{sinc}^2 \left(u\tau\right) q_{1|1}^D(u).
\end{equation}

Before evaluating this integral, let us consider the case $k=2$. By enumerating the possible sequences of indices, we get
\begin{align}
    \mathrm{Tr} [(\rho^{(2)}_{\mathrm{Temp.}}(\tau)-\rho^{(2)}_{\mathrm{Temp.}}(\infty))^2]&=\frac{D!}{(D-4)!}\mathbb{E}_{H \sim \text{GUE}}\left[\mathrm{sinc}^2\left((E_{1}+E_2-E_{3}-E_4)\tau\right)\right] \mathbb{E}_{H \sim \text{GUE}}\left[ |c_{1}|^2|c_{2}|^2 |c_{3}|^2|c_{4}|^2\right] \nonumber\\
    &+ 2\frac{D!}{(D-3)!}\mathbb{E}_{H \sim \text{GUE}}\left[\mathrm{sinc}^2\left((E_{1}+E_2-2E_{3})\tau\right)\right] \mathbb{E}_{H \sim \text{GUE}}\left[ |c_{1}|^2|c_{2}|^2 |c_{3}|^4\right]\nonumber\\
    &+ 4\frac{D!}{(D-3)!}\mathbb{E}_{H \sim \text{GUE}}\left[\mathrm{sinc}^2\left((E_{1}-E_{2})\tau\right)\right] \mathbb{E}_{H \sim \text{GUE}}\left[ |c_{1}|^4|c_{2}|^4\right]\\
    &+ D(D-1)\mathbb{E}_{H \sim \text{GUE}}\left[\mathrm{sinc}^2\left((2E_{1}-2E_{2})\tau\right)\right] \mathbb{E}_{H \sim \text{GUE}}\left[ |c_{1}|^4|c_{2}|^4\right]\nonumber\\
    &+ 4D(D-1)\mathbb{E}_{H \sim \text{GUE}}\left[\mathrm{sinc}^2\left((E_{1}-E_{2})\tau\right)\right] \mathbb{E}_{H \sim \text{GUE}}\left[ |c_{1}|^2|c_{2}|^6\right].\nonumber
\end{align}

Similar expressions can be obtained for generic $k$-th moments: in general, to study the time dependence we need to evaluate the probability distribution of linear combinations of energy eigenstates, i.e., quantities of the form
\begin{equation}
    q_{m_1\dots m_s|n_1\dots n_t}^D(u)\equiv \int \mathrm{d}E_1 \dots \mathrm{d}E_{D} \,\delta \left[\left(\sum_{i=1}^s m_i E_i\right)-\left(\sum_{j=1}^t n_j E_{s+j}\right)-u\right] p^D(E_1,\dots,E_D).
\end{equation}

The probability distribution $q_{m_1\dots m_s|n_1\dots n_t}^D(u)$ is then used to compute the following integral 

\begin{align}
    F^D_{m_1\dots m_s|n_1\dots n_t}(\tau)&=\int \mathrm{d}E_1 \dots \mathrm{d}E_{D} \,\mathrm{sinc}^2 \left[\left(\sum_{i=1}^s m_i E_i\right)\tau-\left(\sum_{j=1}^t n_j E_{s+j}\right)\tau\right] p^D(E_1,\dots,E_D)\\
    &= \int \mathrm{d}u \,\mathrm{sinc}^2 \left(u\tau\right) q_{m_1\dots m_s|n_1\dots n_t}^D(u).\label{eq:Fintegral}
\end{align}

\subsubsection{Evaluation of the integral}
To evaluate the integral in Eq.~\eqref{eq:Fintegral} we use the following property: the probability distribution $q_{m_1\dots m_s|n_1\dots n_t}^D(u)$ has, in general, the expression

\begin{equation}
\label{eq:qgen}
    q_{m_1\dots m_s|n_1\dots n_t}^D(u)=e^{-\gamma u^2}P_{m_1\dots m_s|n_1\dots n_t}^D(\gamma u^2),
\end{equation}
where
\begin{equation}
\gamma=\frac{D}{2}\left[\sum_{i=1}^s m_i^2+\sum_{j=1}^t n_j^2\right]^{-1},
\end{equation}
and $P_{m_1\dots m_s|n_1\dots n_t}^D(x)$ is a polynomial in $x$. This expression can be obtained from the joint probability distribution of $N$ levels, which can be computed with standard random matrix theory techniques \cite{mehta2004random}:
\begin{equation}
    p^D_N(E_1,\dots E_N)\equiv \int \mathrm{d}E_{N+1} \dots \mathrm{d}E_{D} \,  p^D(E_1,\dots E_D)=\frac{(D-N)!}{D!} \mathrm{det} \left(K_D(E_i, E_j) \right)|_{i,j=1}^N.
\end{equation}
Here the kernel $K_D(x,y)$ is defined as
\begin{equation}
K_D(x,y)=e^{-\frac{D}{4}(x^2+y^2)}\sum_{j=0}^{N-1}\frac{H_j(x\sqrt{D/2})H_j(y\sqrt{D/2})}{\sqrt{2\pi D}2^j j!},
\end{equation}
where $H_j(x)$ is the $j$-th Hermite polynomial. Using this expression, through a change of variables we can derive in general an expression of the form in Eq.~\eqref{eq:qgen}.

We can then evaluate the integral for each term of the polynomial $P_{m_1\dots m_s|n_1\dots n_t}^D(x)$ as follows:

\begin{equation}
    \int \mathrm{d}u \,\mathrm{sinc}^2 \left(u\tau\right) e^{-\gamma u^2}(\gamma u^2)^n=\begin{dcases}
\frac{\gamma^{1/2}}{\tau^2}\left(\frac{\sqrt{\pi}}{2^{2n-1}}(2n-3)!!+e^{-\tau^2/\gamma}O(\gamma^{-1}\tau^2(n-1))\right), & \text{\;for }n\ge 1\\
\frac{\sqrt{\pi}}{\tau}\left[\sqrt{\pi} \mathrm{Erf}(\gamma^{-1/2}\tau)-\frac{\gamma^{1/2}}{\tau}(1-e^{-\tau^2/\gamma})\right]. &\text{\;for } n=0
    \end{dcases}
\end{equation}
The leading term at long time hence has a power-law dependence $1/\tau^2$ if $P_{m_1\dots m_s|n_1\dots n_t}^D(0)=0$ and 
$1/\tau$ if $P_{m_1\dots m_s|n_1\dots n_t}^D(0)\neq 0$. For the case $k=1$, it is possible to prove (e.g., from an explicit calculation) that $P_{1|1}^D(0)=0$: this is the effect of level repulsion, as observed in Fig.~\ref{fig:app_finiteT}(d). For generic $P_{m_1\dots m_s|n_1\dots n_t}^D(x)$, we do not expect similar vanishing for $x=0$ [Fig.~\ref{fig:app_finiteT}(e)]. Therefore we get the following scaling at long times $\tau$:

\begin{equation}
\mathbb{E}_{H\sim \text{GUE}}\left[ \Vert  \rho^{(k)}_H(\tau) - \rho^{(k)}_H \Vert_F^2\right] \propto\begin{dcases}
    1/\tau^2 \qquad \text{for } k=1, \\
    1/\tau\qquad \text{for } k>1.
\end{dcases}
\end{equation}

\section{Finite-time temporal mutual information}
\label{app:finite_temporal_mutual_information}
In Sec.~\ref{subsubsec:finite_temp_mutual_info} we derived an expression for the mutual information of temporal ensembles of finite intervals $\tau$: $I(Z;T) \approx 1-\gamma - \frac{\sqrt{\pi}}{2 \sigma_H \tau}$, which we provide a full derivation of here.
\subsection{Bitstring correlation time}
We will first introdyce the ``bitstring correlation time" $t_c$, which characterizes the time scale over which the bitstring distribution $p(z,t)$ becomes independent (after rescaling their time-averaged values $p_\text{avg}(z)$). That is, how long does it take to form a new, independent speckle pattern? We will show that $t_c = \sigma_H^{-1}$, where $\sigma_H^2 = \tr(H^2)-\tr(H)^2$. To do this, we compute the correlator:
\begin{align}
  A_z(t) &=  \frac{\int dt'~p(z,t') p(z,t'+t)}{\left[\int dt' p(z,t')\right]^2} -1 = \frac{\abs{\sum_E \abs{\braket{z}{E}}^2 \abs{\braket{E}{\Psi_0}}^2 \exp(-i E t)}^2}{\abs{\sum_E \abs{\braket{z}{E}}^2 \abs{\braket{E}{\Psi_0}}^2}^2} + \cdots~,
  \label{eq:bitstring_autocorrelation}
\end{align}
where we use ``$\cdots$" to denote subleading terms. In order to make this calculation tractable, we make the following approximations:
\begin{align}
    &\sum_E \rightarrow \frac{1}{\sqrt{2\pi \sigma_H^2}}  \int_{-\infty}^{\infty} dE \exp(-\frac{(E-E_0)^2}{2\sigma_H^2})~,\\
    &\abs{\braket{z}{E}}^2 \rightarrow \exp(-\beta_z E)/Z_z~,\\
    &\abs{\braket{E}{\Psi_0}}^2 \rightarrow \exp(-\beta_\Psi E)/Z_\Psi~,
\end{align}
where $E_0 = \tr(H)$, $\beta_z,\beta_E$ are appropriate inverse temperatures of the state $\ket{z}$ and $\ket{\Psi_0}$, and $Z_z, Z_\Psi$ are their corresponding partition functions. The first approximation takes the density of states of the spectrum to be a Gaussian~\cite{keating2015spectra}, while the second and third approximations treat the energy populations of the states $\ket{z}$ and $\ket{\Psi_0}$ as Boltzmann distributions (Section~\ref{sec:eigenstate_PT}). We obtain 
\begin{align}
\sum_E \abs{\braket{z}{E}}^2 \abs{\braket{E}{\Psi_0}}^2 \exp(-i E t) &\rightarrow \frac{1}{\sqrt{2\pi \sigma_H^2}} \int_{-\infty}^{\infty} dE \exp(-\frac{(E-E_0)^2}{2\sigma_H^2}) \exp(-(\beta_z+\beta_E+i t)E) \\
&= \exp((\beta_z+\beta_E+i t)^2 \sigma_H^2/2) \label{eq:bitstring_autocorr_function}  
\end{align}
Taking the ratio of the absolute values of Eq.~\eqref{eq:bitstring_autocorr_function} at nonzero $t$ and $t =0 $ gives the desired result.

\subsection{Temporal mutual information}
We can then derive the observed scaling of the mutual information
\begin{align}
    I(Z; T) = H(Z) - H(Z|T) &= -\sum_z \bar{p}_{[t,t+\tau]}(z) \log_2(\bar{p}_{[t,t+\tau]}(z)) + \int_t^{t+\tau} dt' \frac{1}{\tau} \sum_z p_{t'}(z) \log_2(p_{t'}(z))~,
\end{align}
where $\bar{p}_{[t,t+\tau]}(z) = \tau^{-1} \int_t^{t+\tau} dt' p_{t'}(z)$ is the outcome probability averaged over an interval. Using the facts that: (i) $p_{t'}(z) \approx p_\text{avg}(z) \tilde{p}$, where $\tilde{p} \sim \text{PT}(1)$, and (ii) the bitstring correlation time is $t_c = \sigma_H^{-1}$, we approximate $\bar{p}/p_\text{avg}$ to be the average of $n = \tau/t_c$ \textit{independent PT variables}. The average of $n$ Porter-Thomas (exponential) variables follows an \textit{Erlang distribution} with mean 1~\cite{shaw2024universal}. Following convention, we denote this distribution as $\bar{p}/p_\text{avg} \sim \text{Erlang}(1/n,n)$, with PDF
\begin{equation}
    P_{\text{Erlang}(1/n, n)}(x) dx = \frac{ (nx)^{n-1} \exp(-n x)}{(n-1)!} n dx~.
\end{equation}
This approximation allows us to calculate the entropy of $\bar{p}$:
\begin{align}
     &-\sum_z \bar{p}_{[t,t+\tau]}(z) \log_2(\bar{p}_{[t,t+\tau]}(z)) \approx - \mathbb{E}[\bar{p}] \sum_z p_\text{avg}(z) \log_2(p_\text{avg}(z)) - \mathbb{E}[\bar{p} \log_2(\bar{p})] \sum_z p_\text{avg}(z)\\
     & = H_{p_\text{avg}} +\log_2(n) - \Psi(n+1) \overset{n\rightarrow \infty}{\longrightarrow} (H_{p_\text{avg}} - \frac{1}{2n})/\ln 2 = \left(H_{p_\text{avg}} - \frac{1}{2 \sigma_H \tau}\right)/\ln 2 
\end{align}
where $\Psi(x)$ is the digamma function, related to the Gamma function $\Gamma(x)$ by $\Psi(x) = \Gamma'(x)/\Gamma(x)$. A more careful analysis modifies the effective number $n$ of the Erlang distribution by considering the second moment of $\bar{p}$:
\begin{align}
    \mathbb{E}_t[\bar{p}(z,t)^2] &= \frac{1}{\tau^2} \int_0^\tau dt_1 \int_0^\tau dt_2 p(z,t_1) p(z,t_2) \\
    &\approx 1 + \frac{1}{\tau} \int_{-\infty}^{\infty} d t' \exp(-\sigma_H^2 t'^2) = 1 + \frac{\sqrt{\pi}}{\sigma_H \tau} \equiv 1 + \frac{1}{n_\text{eff}}~,
\end{align}
i.e.~the effective parameter $n_\text{eff}$ of the Erlang distribution is $n_\text{eff} = \sigma_H \tau/\sqrt{\pi}$.
This gives the overall temporal mutual information for a finite time interval:
\begin{align}
    I(Z;T) \approx \left(1-\gamma - \frac{\sqrt{\pi}}{2 \sigma_H \tau}\right)/\ln 2~,
\end{align}
which agrees with numerical results [inset of Fig.~\ref{fig:temporal_ensemble}(d)].

\section{$k$-th moments of the Scrooge ensemble}
\label{app:Scrooge_kth_moment}
Here, we derive our expression for the $k$-th moments of the Scrooge ensemble [Eq.~\eqref{eq:scrooge_kth_moment}] $\text{Scrooge}[\rho]$:
\begin{equation}
\rho^{(k)}_{\text{Scr.},\pmb{m}}=\left(\prod_i \lambda^{-1}_i\right)\frac{\partial^{k} \Lambda_k}{\partial \mu_{m_1}\dots \partial \mu_{m_k}} {\Bigg\vert}_{\mu_i = \lambda_i^{-1}}~,  
\end{equation}
where 
\begin{align}
    \Lambda_k&\equiv \sum_{j=1}^{D} \frac{\mu^{k-2}_j\ln \mu_j}{\prod_{k\neq j} (\mu_j-\mu_k)},
\end{align}
and $\{\lambda_i\}$ are the eigenvalues of the given state $\rho$.

As an explicit example, we consider the simplest non-trivial parameters, $k=2$ and $D=2$. Let the eigenstates of $\rho$ be $|\uparrow\rangle$ and $|\downarrow\rangle$, and the eigenvalues be $\lambda,1-\lambda$. Here, the second moment $\rho^{(2)}$ is a $4\times 4$ matrix. However, in the (doubled) eigenbasis of $\rho$, only 6 entries are non-zero: $\langle \uparrow \uparrow |\rho^{(2)}| \uparrow \uparrow\rangle \equiv \rho^{(2)}_{\text{Scr.}, 11} $, $\langle \uparrow \downarrow |\rho^{(2)}| \uparrow \downarrow\rangle = \langle \uparrow \downarrow |\rho^{(2)}| \downarrow \uparrow\rangle = \langle \downarrow \uparrow |\rho^{(2)}| \uparrow \downarrow\rangle = \langle \downarrow \uparrow |\rho^{(2)}| \downarrow \uparrow\rangle \equiv \rho^{(2)}_{\text{Scr.}, 12}$, and $\langle \downarrow \downarrow |\rho^{(2)}| \downarrow \downarrow\rangle \equiv \rho^{(2)}_{\text{Scr.}, 22}$, where
\begin{align}
    \rho^{(2)}_{\text{Scr.}, 11} = \frac{\lambda ^2 \left(2 \lambda  (4 \lambda -5)-2 (1-\lambda )^2 \log \left(\frac{1-\lambda}{\lambda }\right)+3\right)}{(2 \lambda -1)^3}\\
    \rho^{(2)}_{\text{Scr.}, 12} = \frac{(1-\lambda) \lambda  \left(2 \lambda +2 (1-\lambda) \lambda  \log \left(\frac{1-\lambda}{\lambda}\right)-1\right)}{(2 \lambda -1)^3}\\
    \rho^{(2)}_{\text{Scr.}, 22}= -\frac{(1-\lambda)^2 \left(8 \lambda ^2+2 \lambda ^2 \log \left(\frac{1-\lambda}{\lambda }\right)-6 \lambda +1\right)}{(2 \lambda -1)^3}
\end{align}

\subsection{Second moment of the Scrooge ensemble}
First, we explicitly compute the second moment of the Scrooge ensemble for a $D$-dimensional density matrix $\rho$. We write $\rho$ using its eigenvalues and eigenvectors $\rho=\sum_{m=1}^D \lambda_m \ket{s_m}\bra{s_m}$. We next use the definition of the Scrooge ensemble as the $\rho$-distortion of the Haar ensemble to write the second moment of the Scrooge ensemble as
\begin{align}
    \rho^{(2)}&=\int \frac{\Big(\sqrt{D\rho}\ket{\Psi}\bra{\Psi}\sqrt{D\rho}\Big)^{\otimes 2}}{\lVert\sqrt{D\rho}\ket{\Psi}\rVert^2}P_{\mathrm{Haar}}(\Psi)d\Psi\\ &=\sum_{m,n,m',n'=1}^D\int \frac{D \sqrt{\lambda_m\lambda_n\lambda_{m'}\lambda_{n'}}\braket{s_n}{\Psi}\braket{s_m}{\Psi}\braket{\Psi}{s_{n'}}\braket{\Psi}{s_{m'}}}{\sum_{k=1}^D \lambda_k^2 |\braket{s_k}{\Psi}|^2} P_{\mathrm{Haar}}(\Psi)d\Psi \ket{s_n}\ket{s_m}\bra{s_{n'}}\bra{s_{m'}}.
\end{align}
From this equation, we clearly see that the only non-zero elements in the basis of the eigenvectors of $\rho$ are the ones with $n=n'$, $m=m'$ and/or $n=m'$, $m=n'$. All the other terms are zero because of the average over the complex phases. We can then focus, without loss of generality, on the following matrix element of $\rho^{(2)}$
\begin{align}
    \rho^{(2)}_{\text{Scr.}, nm}&\equiv \bra{s_n}\bra{s_m}\rho^{(2)}\ket{s_n}\ket{s_m}=\bra{s_n}\bra{s_m}\rho^{(2)}\ket{s_m}\ket{s_n}\\
    &=\int \frac{D \lambda_m \lambda_n |\Psi_n|^2|\Psi_m|^2}{\sum_{k=1}^D \lambda_k^2 |\Psi_k|^2} \frac{(D-1)!}{\pi^D} \delta(1-\sum_{j=1}^D |\Psi_j|^2) d^2\Psi_1\dots d^2\Psi_D.
\end{align}
where we used the short notation $\Psi_j=\braket{s_j}{\Psi}$ and the explicit expression of $P_\mathrm{Haar}$. Since the integrand only depends on the absolute values of the complex variables $\Psi_1\dots \Psi_D$ we can perform the integral over the phases and then change variable to $x_j=\lambda_j |\Psi_j|^2$, obtaining

\begin{equation}
    \label{eq:rho2mel}
    \rho^{(2)}_{\text{Scr.}, nm}=\left(\prod_{j=1}^D \lambda_j^{-1}\right)D!\int_0^\infty dx_1\dots \int_0^\infty dx_D \frac{ x_n x_m}{\sum_{k=1}^D x_k}  \delta(1-\sum_{j=1}^D \lambda_j^{-1}x_j) .
\end{equation}
To remove the delta function, it is useful to define the following quantity
\begin{equation}
    f^{(2)}_{nm}(t) =\left(\prod_{j=1}^D \lambda_j^{-1}\right)D!\int_0^\infty dx_1\dots \int_0^\infty dx_D \frac{ x_n x_m}{\sum_{k=1}^D x_k}  \delta(t-\sum_{j=1}^D \lambda_j^{-1}x_j) 
\end{equation}
such that $\rho^{(2)}_{\text{Scr.}, nm}=f^{(2)}_{nm}(1)$, and the Laplace transform
\begin{equation}
    \tilde f^{(2)}_{nm}(s)= \int_0^\infty f^{(2)}_{nm}(t)e^{-st}dt =\left(\prod_{j=1}^D \lambda_j^{-1}\right)D!\int_0^\infty dx_1\dots \int_0^\infty dx_D \frac{ x_n x_m}{\sum_{k=1}^D x_k}  \exp\left({-s\sum_{j=1}^D \lambda_j^{-1}x_j}\right).
\end{equation}
The change of variables $x_j\rightarrow x_j/s$ yields

\begin{equation}
    \tilde f^{(2)}_{nm}(s) =\left(\prod_{j=1}^D \lambda_j^{-1}\right)D!s^{-D-1}\int_0^\infty dx_1\dots \int_0^\infty dx_D \frac{ x_n x_m}{\sum_{k=1}^D x_k}  \exp\left({-\sum_{j=1}^D \lambda_j^{-1}x_j}\right).
\end{equation}
Since the inverse Laplace transform of $D!s^{-D-1}$ is $t^D$, we find 
\begin{equation}
    \rho^{(2)}_{\text{Scr.}, nm}=f^{(2)}_{nm}(1) =\left(\prod_{j=1}^D \lambda_j^{-1}\right)\int_0^\infty dx_1\dots \int_0^\infty dx_D \frac{ x_n x_m}{\sum_{k=1}^D x_k}  \exp\left({-\sum_{j=1}^D \lambda_j^{-1}x_j}\right).
\end{equation}
Note that these steps allowed us to get from the integral over the Haar ensemble of normalized states in Eq. (\ref{eq:rho2mel}) to an integral over the unnormalized Gaussian ensemble.

We now define
\begin{equation}
\label{eq:Lambda2_0}
    \Lambda_2 =  \int_0^\infty dx_1\dots \int_0^\infty dx_D \frac{ 1}{\sum_{k=1}^D x_k}  \exp\left({-\sum_{j=1}^D \lambda_j^{-1}x_j}\right)
\end{equation}
such that we can now write all the $ \rho^{(2)}_{\text{Scr.}, nm}$ from this single ``partition function":
\begin{equation}
\label{eq:rho2elements}
\rho^{(2)}_{\text{Scr.}, nm} = \left(\prod_{j=1}^D \lambda_j^{-1}\right) \frac{\partial}{\partial \lambda_n^{-1}}\frac{\partial}{\partial \lambda_n^{-1}} \Lambda_2 .
\end{equation}

To find the expression for $\Lambda_2$ in a closed form, we introduce an additional integration variable $z$
\begin{equation}
    \Lambda_2  = \int_0^\infty dx_1\dots \int_0^\infty dx_D \int_0^\infty dz\,   \exp\left[{-\sum_{j=1}^D (\lambda_j^{-1}+z)x_j}\right].
\end{equation}
The equivalence with Eq. (\ref{eq:Lambda2_0}) is evident upon performing the integral in $z$. Since the integral is now factorized, we can easily compute the integral in the variables $x_1,\dots x_D$, obtaining
\begin{equation}
\label{eq:L2A}
    \Lambda_2  = \int_0^\infty dz \prod_{j=1}^D (\lambda_j^{-1}+z)^{-1}=\int_0^\infty dz \sum_{j=1}^D \frac{A_j}{\lambda_j^{-1}+z}
\end{equation}
where 
\begin{equation}
    A_j=\prod_{k\neq j} (\lambda_k^{-1}-\lambda_j^{-1})^{-1}.
\end{equation}
We can now perform the integral in $z$ and we get
\begin{equation}
    \Lambda_2 =\lim_{M\rightarrow \infty}\sum_{j=1}^D A_j[\ln(\lambda_j^{-1}+M)-\ln (\lambda_j^{-1})].
\end{equation}
Using the property $\sum_j A_j=0$, we find
\begin{equation}
  \lim_{M\rightarrow \infty}  \sum_j A_j \ln (\lambda_j^{-1}+M)=\lim_{M\rightarrow \infty} \left[\sum_j A_j \ln (M)+ \sum_j A_j \ln \left(1+\frac{\lambda_j^{-1}}{M}\right)\right]=0.
\end{equation}

We then get the final expression for $\Lambda_2$
\begin{equation}
    \Lambda_2=\sum_{j=1}^D A_j \ln (\lambda_j)=\sum_{j=1}^D \left[\prod_{k\neq j} (\lambda_k^{-1}-\lambda_j^{-1})^{-1}\right] \ln (\lambda_j).
    \label{eq:second_moment_partition_function}
\end{equation}
This expression can be straightforwardly used to evaluate the matrix elements of the second moment of the Scrooge ensemble using Eq. (\ref{eq:rho2elements}).

\subsection{Second moment in the limit of small purity of $\rho$}
Here, show that in the limit that $\rho$ has small purity, the second moment of the Scrooge ensemble has a simple form $\rho^{(2)}\simeq(1+\mathbb S)\rho^{\otimes2}$, equivalent to the product form for the temporal ensemble~(Theorem~\ref{thm:temp_ens_product_form}).
To do so, we compute an approximate formula for $\Lambda_2$ in the case of a state $\rho$ with small purity $\mathrm{Tr}[\rho^2]=\sum_k \lambda_k^2\ll 1$. Using that $\sum_k \lambda_k^n\le \sum_k \lambda_k^2\ll 1$ for every $n\ge 2$, from Eq. (\ref{eq:L2A}) we get:

\begin{align}
    \Lambda_2&=\left(\prod_{j=1}^{D}\lambda_j\right)\int_0^{\infty} \mathrm{d}z\prod_{k=1}^{D}\left(1+\lambda_k z\right)^{-1}\\
    &=\left(\prod_{j=1}^{D}\lambda_j\right)\int_0^{\infty} \mathrm{d}z\left\{1+\left(\sum_{k=1}^D \lambda_k\right)z+\frac{1}{2}\left[\left(\sum_{k=1}^D\lambda_k\right)^2-\sum_{k=1}^D \lambda_k^{2}\right]z^2+\dots\right\}^{-1}\\
    &\approx \left(\prod_{j=1}^{D}\lambda_j\right)\int_0^{\infty} \mathrm{d}z\left[\sum_{n=0}^{\infty}\frac{1}{n!}\left(\sum_{k=1}^D \lambda_k\right)^n z^n\right]^{-1}\\
    &= \left(\prod_{j=1}^{D}\lambda_j\right)\int_0^{\infty} \mathrm{d}z\, \exp\left(-z\sum_{k=1}^D \lambda_k\right)
    = \left(\prod_{j=1}^{D}\lambda_j\right)\left(\sum_{k=1}^D \lambda_k\right)^{-1}
\end{align}

From this equation we obtain
\begin{equation}
    \rho_{nm}^{(2)} \approx \lambda_n\lambda_m[(1+\delta_{nm})(1-\lambda_n-\lambda_m)+2\lambda_m\lambda_n]
\end{equation}
and, assuming that the Schmidt coefficients are small we have ${\rho_{\text{Scr.}, nm}^{(2)}\simeq\lambda_n\lambda_m(1+\delta_{nm})}$, which implies $\rho^{(2)}\simeq(1+\mathbb S)\rho^{\otimes2}$.

\subsection{Higher moments}
We now generalised the above derivation to compute the matrix elements of the $k$-th moment ($k<D$). The non-zero matrix elements in the eigenbasis of $\rho$ are of the form 

\begin{equation}
    \rho_{m_1,\dots m_k}^{(k)}  \equiv \bra{s_{m_1}}\dots\bra{s_{m_k}}\rho^{(2)}\ket{s_{m_1'}}\dots\ket*{s_{m_k'}}
\end{equation}
where $m_1', \dots, m_k'$ is a permutation of the indices $m_1,\dots m_k$. With the same procedure used above for the second moment, we get

\begin{equation}
    \rho_{m_1,\dots m_k}^{(k)}=\left(\prod_{j=1}^D \lambda_j^{-1}\right)\int_0^\infty dx_1\dots \int_0^\infty dx_D \frac{ x_{m_1}\dots x_{m_k}}{(\sum_{j=1}^D x_j)^{k-1}}  \exp\left({-\sum_{j=1}^D \lambda_j^{-1}x_j}\right).
\end{equation}
We can compute the this matrix element for arbitrary indices using the generating function $\Lambda_k$ here defined:

\begin{equation}
    \rho_{m_1,\dots m_k}^{(k)}=\left(\prod_j \lambda_j^{-1}\right)\frac{\partial^{k} \Lambda_k}{\partial \lambda^{-1}_{m_1}\dots \partial \lambda^{-1}_{m_k}},
\end{equation}

\begin{equation}
    \Lambda_k=\int_0^\infty dx_1\dots \int_0^\infty dx_D \frac{(-1)^k}{(\sum_{j=1}^D x_j)^{k-1}}  \exp\left({-\sum_{j=1}^D \lambda_j^{-1}x_j}\right)
\end{equation}
Simlarly to the case $k=2$, the integral can be computed by introducing an additional integration variable:
\begin{align}    
    \Lambda_k&=\int_0^\infty dx_1\dots \int_0^\infty dx_D \,\exp\left({-\sum_{j=1}^D \lambda^{-1}_j x_j}\right)\int_0^\infty \mathrm{d}z \frac{(-1)^k z^{k-2}}{(k-2)!}\, \exp\left({-z\sum_{j} x_j}\right)\\
    &=\int_0^\infty \mathrm{d} z\, \frac{(-1)^k z^{k-2}}{(k-2)!}\prod_{j=1}^{D} (\lambda_j^{-1}+z)^{-1}\\
    &=\int_0^\infty  \mathrm{d} z\, \frac{(-1)^k z^{k-2}}{(k-2)!}\sum_{j=1}^{D} \frac{A_j}{\lambda_j^{-1}+z}\\
    \label{eq:Lambdak}
    &=\sum_{j=1}^{D} \lambda_j^{-k
    +2} A_j \ln \lambda_j.
\end{align}
This derivation holds for $k>1$, but we can define in general $\Lambda_k$ using Eq. (\ref{eq:Lambdak}).
Then the subentropy derived in Ref.~\cite{jozsa1994lower} can be written as
\begin{equation}
    Q(\rho)=-\left(\prod_{j=1}^D \lambda_j^{-1}\right)\sum_{j=1}^D A_j \lambda_j^{2}\ln \lambda_j=-\left(\prod_{j=1}^D \lambda_j^{-1}\right)\Lambda_0
\end{equation}

\section{Proof of Lemma~\ref{lemma:temp_ens_of_proj_st}: $k$-th moments of temporal ensemble of projected states}
\label{app:proof_temp_ens_of_proj_st}
Here, we provide the proof for Lemma~\ref{lemma:temp_ens_of_proj_st}.
\setcounter{lemma}{0}
\begin{lemma}[Temporal $k$-th moment of the projected state]
Consider the temporal ensemble comprised of the (unnormalized) projected states $\ket*{\tilde{\Psi}(x_B,t)} \equiv (\mathbb{I}_A \otimes \bra{x_B})\ket{\Psi(t)}$. The $k$-th moment of this temporal ensemble is equal to the $k$-th moment of the corresponding unnormalized Scrooge ensemble, up to small corrections:
\begin{align}
    &\rho^{(k)}_\text{Temp. Proj.}(x_B) \equiv \mathbb{E}_t \left[\ketbra*{\tilde{\Psi}(x_B,t)}{\tilde{\Psi}(x_B,t)}^{\otimes k} \right]= p_d(x_B)^k\tilde{\rho}^{(k)}_{\tilde{\text{S}}\text{cr.}}[\bar{\rho}(x_B)] + O(\Delta_\beta(x_B)p_d(x_B)^{k})~,\nonumber 
\end{align}
where $\Delta_\beta(x_B) \equiv  \tr[(\mathbb{I}_A\otimes \ketbra{x_B}{x_B})^{\otimes 2}\rho_d^{(2)}]/p_d(x_B)^2$ is a correction term, that is typically exponentially small in system size. To see this, note that $\Delta_\beta(x_B)$ can be interpreted as the purity of a probability distribution over $D$ elements. Specifically, we define the quantity $q_{x_B}(E) \equiv \abs{c_E}^2 \bra{E}(\mathbb{I}_A\otimes \ketbra{x_B}{x_B})\ket{E} / p_d(x_B)$ which is a probability distribution because it is non-negative and satisfies $\sum_E q_{x_B}(E) = 1$. $\Delta_\beta(x_B) = \sum_E q_{x_B}(E)^2$ is the purity of $q_{x_B}(E)$, which we expect to be typically $O(D^{-1})$
\end{lemma}

\begin{proof}
  This is a simple application of an intermediate step from Appendix~\ref{app:temporal_ensemble_error}, which states that
    \begin{align}
        \left \Vert \bar{\rho}(x_B)^{\otimes k} \sum_{\sigma \in S_k} \text{Perm}(\sigma) - \rho^{(k)}_\text{Temp. Proj.} \right \Vert_* 
        &= \left \Vert  \sum_{s \subset M_k(D)} \bar{\rho}(x_B)^{\otimes k}\left(1-\frac{\abs{\pi(s)}}{\abs{S_k}}\right) P_s \sum_{\sigma \in S_k}\text{Perm}(\sigma) \right \Vert_*  \nonumber\\
         \leq  C_k \sum_{\overset{s \subset M_k(D)}{s \neq ()}} \tr\left[\bar{\rho}(x_B)^{\otimes (k-2)} \otimes \bar{\rho}^{(2)}(x_B) \sum_{\sigma \in S_k}\text{Perm}(\sigma) \right ] 
        &  \leq (C_k k!) \text{tr}[\bar{\rho}^{(2)}(x_B)] p_d(x_B)^{k-2} \equiv C_k' \Delta_\beta(x_B)p_d(x_B)^{k}~.
    \end{align}  
\end{proof}

\section{Proof of Theorem~\ref{thm:moment_of_general_projected_ensemble}: weighted $k$-th moments of projected ensemble}
\label{app:gen_Scrooge_proof}
Here, we prove Theorem~\ref{thm:moment_of_general_projected_ensemble}, 
\setcounter{theorem}{3}
\begin{theorem}
    At \textit{typical} (late) times $t$, the weighted $k$-th moment $\tilde{\rho}^{(k)}_\text{Proj.}(t)$ of the unnormalized projected ensemble is equal to the $k$-th moment of the generalized Scrooge ensemble:
    \begin{align}
       \tilde{\rho}^{(k)}_\text{Proj.}(t) &\equiv \sum_{x_B} \frac{\ketbra*{\tilde{\Psi}(x_B,t)}{\tilde{\Psi}(x_B,t)}^{\otimes k}}{p_d(x_B)^{k-1}} \label{eq:kth_moment_proof}\\
       &= \sum_{x_B} p_d(x_B) \tilde{\rho}^{(k)}_{\tilde{\text{S}}\text{cr.}}[\bar{\rho}(x_B)]+O(\Delta^{1/2}_\beta) ~,\label{eq:general_projected_moment}
    \end{align}
    where $\Delta_\beta \equiv \sum_{x_B} p_d(x_B) \Delta_\beta(x_B)$ the weighted mean of $\Delta_\beta(x_B)$.
\end{theorem}

\begin{proof}
    First, we compute the time average:
\begin{equation}
\mathbb{E}_t[\tilde{\rho}^{(k)}_\text{Proj.}(t)] \equiv  \sum_{x_B} p_d(x_B) \tilde{\rho}^{(k)}_{\text{S}\text{cr.}}(\bar{\rho}(x_B)) + O(\Delta_\beta)~. 
\end{equation}
This follows from Appendix~\ref{app:proof_temp_ens_of_proj_st}. Next, we compute the variance $\text{var}_t[\tilde{\rho}^{(k)}_\text{Proj.}(t)]$ by computing the trace norm
\begin{align}
    \Vert \mathbb{E}_t[ \tilde{\rho}^{(k)}_\text{Proj.}(t)^{\otimes 2}] - \mathbb{E}_t[ \tilde{\rho}^{(k)}_\text{Proj.}(t)]^{\otimes 2}\Vert_* = O(\tr(\rho_d^2)) = O(\Delta_\beta)~.
\end{align}
The temporal averages can be in terms of: permutations of $2k$ and $k$ elements respectively, and correction terms $\text{E.C.}$ in Lemma~\ref{lemma:Hamiltonian_k_design}. The dominant term will be from the former: permutations of $2k$ elements that cannot be written as disjoint permutations of two groups of $k$ elements. The largest possible contribution comes from the permutation where only two elements from each group are swapped, and all others are not permuted. This has trace
\begin{equation}
    \sum_{x_B,x_B'}\tr[(\mathbb{I}_A \otimes \ketbra{x_B}{x_B})\rho_d(\mathbb{I}_A \otimes \ketbra{x'_B}{x'_B})\rho_d] = \tr(\rho_d^2) ~,\label{eq:contribution_of_swap}
\end{equation}
and the contributions of all other relevant permutations and correction terms can be bounded by Eq.~\eqref{eq:contribution_of_swap}. 
Finally, we can bound $\tr(\rho_d^2) \leq \Delta_\beta$ as follows: we use Sedrakyan's inequality, a specialization of the Cauchy-Schwarz inequality. It states that for positive $u_i$ and $v_i$, $\sum_i u_i^2/\sum_i v_i \leq \sum_i u_i^2/v_i$. Applied to our context, we have
\begin{align}
    \tr(\rho_d^2) = \sum_E \abs{c_E}^4 = \sum_E \abs{c_E}^4 \frac{\left(\sum_{x_B} \bra{E}(\mathbb{I}_A\otimes \ketbra{x_B}{x_B}) \ket{E} \right)^2}{\sum_{x_B} p_d(x_B)} &\leq \sum_E \abs{c_E}^4 \sum_{x_B}\frac{\bra{E}(\mathbb{I}_A\otimes \ketbra{x_B}{x_B}) \ket{E}^2}{p_d(x_B)} \nonumber \\
    &\leq \sum_{x_B} p_d(x_B) \Delta_\beta(x_B) \equiv \Delta_\beta~. 
\end{align}
While the tightest general bound we can achieve is $\Delta_\beta \leq 1$ (saturated when the initial state is an energy eigenstate $\ket{\Psi_0} = \ket{E}$), we expect $\Delta_\beta$ to be exponentially small in total system size. 
As discussed above, $\Delta_\beta(x_B) = \sum_E q_{x_B}(E)^2$ is the purity of a distribution $q_{x_B}(E)$ and hence is typically $O(D^{-1})$. $\Delta_\beta$ is the weighted average of such purities.
\end{proof}

\section{Additional numerical evidence}

\subsection{Chaotic XXZ model}
\label{app:U1_symm}
In this Appendix, we support our results on the projected ensemble with another chaotic model. Specifically, we study a variant of the XXZ model~\cite{leblond2021universality}:
\begin{equation}
    H =\frac{1}{4}\sum_{j=1}^{N-1} \left[J (X_{j}X_{j+1}+Y_{j}Y_{j+1}) + \Delta Z_{j} Z_{j+1} \right]+\frac{\Delta'}{4} \sum_{j=1}^{N-2} Z_{j} Z_{j+2}~,
    \label{eq:chaotic_XXZ}
\end{equation}
where $X_j,Y_j,$ and $Z_j$ are the Pauli matrices on site $j$. Following Ref.~\cite{leblond2021universality}, we have set $J=\sqrt{2}, \Delta = (\sqrt{5}+1)/4$, and $\Delta' = 1$, and performed simulations on a one-dimensional chain with open boundary conditions.

This model serves as a further check of our theoretical conclusions, but also serves to illustrate the effect of conserved quantities on the projected ensemble. This model has a $U(1)$ symmetry, leading to a conserved global magnetization $M \equiv \sum_j Z_j$. 

The $U(1)$ symmetry leads to significant correlations between projected state and outcome. For simplicity, we consider $Z$-basis measurements on the complementary system. If the measured bitstring $z_B$ has magnetization $m$, the projected state $\ket{\Psi(z_B)}$ is guaranteed to have magnetization $M-m$. As a result, even in the case where the outcomes $z_B$ are not ``energy-revealing," they reveal charge, another conserved quantity. We no longer expect the entire projected ensemble to be described by the Scrooge ensemble. Rather, following Ref.~\cite{cotler2023emergent}, we form the projected ensemble only out of measurement outcomes $z_B$ with fixed magnetization $m$, and conjecture that this is described by a Scrooge ensemble of the ``sector reduced density matrix," $\rho_m \propto \sum_{z_B, |z_B| = m} \ketbra{\Psi(z_B)}{\Psi(z_B)}$. 

In the more general case where the outcomes $z_B$ are energy-revealing, Theorem~\ref{thm:moment_of_general_projected_ensemble} applies without modification: each measurement outcomes $z_B$ has its own average state $\bar{\rho}(z_B)$; the projected state $\ket{\Psi(z_B)}$ is sampled from the corresponding Scrooge ensemble $\text{Scrooge}[\bar{\rho}(z_B)]$. We illustrate the validity of this result in Fig.~\ref{fig:XXZ}(a). We plot the trace distances between the third moments of the projected and generalized Scrooge ensemble, for a time-evolution of the initial state with one domain wall $\ket{\Psi_0} = |0\rangle^{\otimes (N/2-1)}|1\rangle^{\otimes (N/2+1)}$. The trace distances decay exponentially with increasing system size, consistent with Theorem~\ref{thm:moment_of_general_projected_ensemble}. 
Meanwhile, the Scrooge ensemble $\text{Scrooge}[\rho_m]$ is a good description of the projected ensemble for smaller system sizes, but deviations become noticeable at larger system sizes, illustrated by the trace distance between the third moments of the projected and Scrooge ensemble, which saturates to a non-zero value (light green). As with the MFIM, the Porter-Thomas distribution is visible in the outcome probabilities $p(z_A,z_B)$, serving as a measurable signature of the generalized Scrooge ensemble.

\subsection{Transverse field Ising model}
We also demonstrate an example where the projected ensemble is not described by the generalized Scrooge ensemble, using the time-evolved states of the transverse field Ising model (TFIM), which is the mixed field Ising model [Eq.~\eqref{eq:MFIM}] with $h_x=0$. We time-evolve the initial state $\ket{0}^{\otimes N}$, and construct the projected ensemble on two sites near the middle of the chain, with measurements in the $X_B$ basis. Unlike the MFIM, the projected ensemble from the TFIM is not close to the generalized Scrooge ensemble [Fig.~\ref{fig:XXZ}(c)]. This is because the TFIM is integrable and has solution by free-fermions, and hence does not satisfy the $k$-th no-resonance conditions. We have also taken care to ensure that our projected ensemble should not be described by the real Scrooge ensemble (Appendix~\ref{app:time_reversal_symmetry} below): the time-reversal symmetry is broken by time evolution, and the particle-hole symmetry is broken by the choice of initial state.

\begin{figure}
    \centering
    \includegraphics[width=0.8\textwidth]{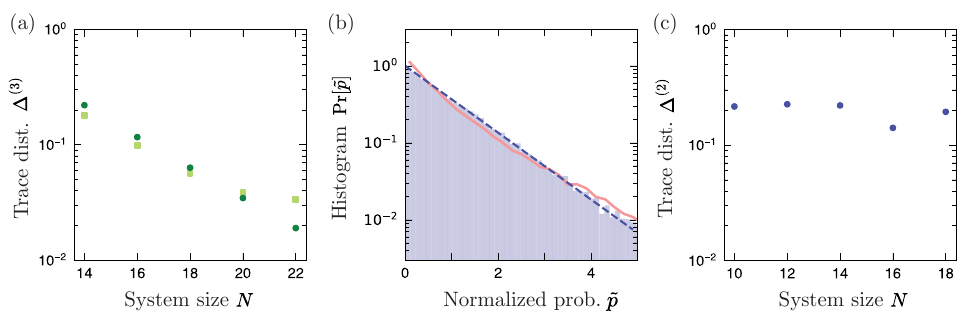}
    \caption{Projected ensembles in a chaotic and integrable model. (a) We time-evolve an initial state and plot the trace distances between the third moments of the Scrooge (light green) and generalized Scrooge (dark green) ensembles and the projected ensemble generated by a evolving a state under a variant of the XXZ model [Eq.~\eqref{eq:chaotic_XXZ}] for a subsystem of four qubits and measurements in the $Z_B$ basis. Like in the MFIM, the projected ensemble is close to the Scrooge ensemble but deviates from it at large system sizes. In contrast, the generalized Scrooge ensemble provides good agreement, even for the largest system sizes studied in this work. (b) The PT distribution is visible in the statistics of the normalized outcome probabilities $\tilde{p} \equiv p_{z_A}(z_B)/\mathbb{E}_t[p_{z_A}(z_B)]$ (while the raw probabilities $D p_{z_A}(z_B)$ show small deviations) (histogram plotted over all $z_A$ and $z_B$). (c) In contrast, the trace distance $\Delta^{(2)}$ between the Scrooge ensemble and the projected ensemble generated from the state $\exp(-i H_\text{TFIM}t)\ket{0}^{\otimes N}$ remains high, where $H_\text{TFIM}$ is the integrable transverse field Ising model [Eq.~\eqref{eq:MFIM} with $h_x=0$], the subsystem $A$ consists of two qubits, and the complement $B$ is measured in the $X$ basis. }
    \label{fig:XXZ}
\end{figure}

\section{Time-reversal symmetry and real Scrooge ensembles}
\label{app:time_reversal_symmetry}
Here, we discuss a notable deviation from the Scrooge ensemble: in the presence of relevant symmetries, states in the ensemble can be constrained to be purely real. In this case, one has to construct a maximally-entropic ensemble of \textit{real} states, which we dub the \textit{real Scrooge ensemble}.

A simple way that this can emerge is the presence of time-reversal symmetry: the existence of an unitary operator $U$ that maps the Hamiltonian $H$ to its complex conjugate $H^*$: $UHU^\dagger = H^*$. The eigenstates of such Hamiltonians are guaranteed to be real in an appropriate basis. Under appropriate measurement bases $z_B$, the projected states are guaranteed to be real. The real Scrooge ensemble can be defined as an analogous $\rho-$distortion of the maximally entropic ensemble of \textit{real states}, which has been studied in e.g.~Ref.~\cite{harrow2013church} and we dub the \textit{real Haar ensemble}.

As an example, in analogy to Eq.~\eqref{eq:rho2elements}, the second moment of the real Scrooge ensemble has non-zero entries:
\begin{equation}
    \rho_{nm}^{(2)}=\left(\prod_{k=1}^{d_A}\lambda_k^{-1}\right)^{1/2}\frac{\partial}{\partial \mu_m}\frac{\partial}{\partial \mu_n}\Lambda_2^\text{real} \big \vert_{\mu_i = \lambda_i^{-1}},
    \label{eq:real_scrooge_second_moment}
\end{equation}
where $\lambda_i$ are the eigenvalues of the first moment $\rho$, and
\begin{equation}
    \Lambda_2^\text{real}=2\int_0^\infty \mathrm d z \prod_{k=1}^{d_A}(\mu_k+z)^{-1/2}.
\end{equation}
While there is no closed-form expression for $\Lambda_2^\text{real}$ analogous to Eq.~\eqref{eq:second_moment_partition_function}, for a given first moment $\rho$, we can numerically evaluate Eq.~\eqref{eq:real_scrooge_second_moment} to obtain the higher moments of the real Scrooge ensemble.

In Fig.~\ref{fig:real_Scrooge}(b), we show that the real Scrooge ensemble describes the projected ensemble of eigenstates of the MFIM. As discussed in Sec.~\ref{subsec:eigenstates}, the MFIM is time-reversal symmetric. By measuring in the $x_B$ basis, the projected states can be written as real states. We compute the trace distance between the second moments of the real Scrooge ensemble (blue) and the real Haar ensemble (red). The real Haar ensemble only describes the projected ensemble at infinite temperature $E\approx 0$, while the real Scrooge ensemble describes the projected ensemble at all energies/temperatures, with exponentially increasing accuracy with increasing system size [Fig.~\ref{fig:real_Scrooge}(b)]. 

In addition to trace distances, we take the \textit{real Porter-Thomas distribution} as smoking-gun evidence of the real Scrooge ensemble. The real Porter-Thomas distribution~\cite{porter1956fluctuations} is an analogous universal distribution, given by
\begin{equation}
    P_{\text{PT}, \mathbb{R}}(x) dx = \frac{\exp(-x/2)}{\sqrt{2\pi x}} dx~.
    \label{eq:real_PT}
\end{equation}

In Fig.~\ref{fig:real_Scrooge}(d), we show the real PT distribution in relevant measurement outcomes $p_E(o_A,x_B)$, following the same procedure as Fig.~\ref{fig:eigenstates}. By using the same normalization procedure in Fig.~\ref{fig:eigenstates}(e), we see even better agreement with the real PT distribution (red), indicating that the \textit{generalized} real Scrooge ensemble is likely a better description than the real Scrooge ensemble, as expected.

Finally, we remark that time-evolved states can also be made to be all real. Unlike with energy eigenstates, time-reversal symmetry is not sufficient, because time-evolution breaks this symmetry. Rather, we need both time-reversal symmetry and \textit{particle-hole symmetry}, which ensures the energy spectrum is symmetric about 0. With appropriate initial states, the time-evolved state will always be real, and as such the projected ensemble will be described by the real Scrooge and generalized real Scrooge ensembles. See SM of Ref.~\cite{mark2022benchmarking} for more details.

\begin{figure}
    \centering
    \includegraphics[width=0.9\textwidth]{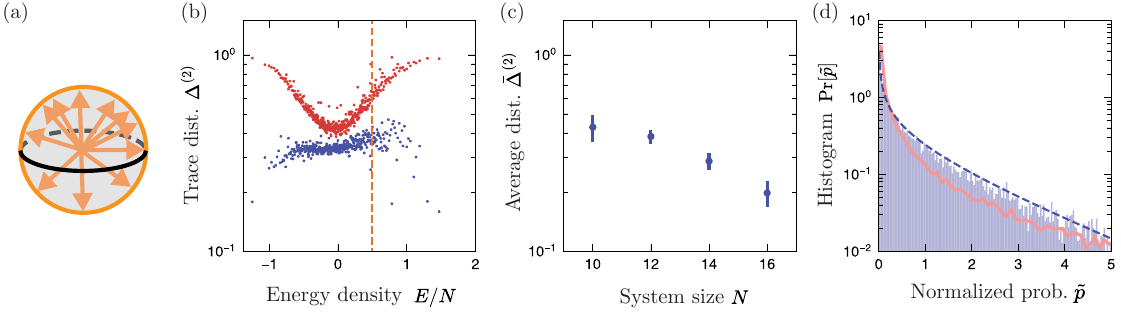}
    \caption{Real Scrooge ensemble. (a) Under appropriate conditions, the states in the temporal and projected ensembles can be restricted to be always real, indicated here with single-qubit states lying in the $X-Z$ plane of the Bloch sphere (orange circle). (b) We demonstrate this with the eigenstates of the MFIM, a time-reversal-symmetric Hamiltonian. We compare the second moments of the projected ensemble with those of: the real Haar ensemble (red) and the real Scrooge ensemble (blue). The real Scrooge ensemble consistently has the smaller trace distance $\Delta^{(2)}$. (c) With increasing system size, the projected ensembles appear to converge to the the real Scrooge ensemble. (d) We verify the generalized real Scrooge ensemble with a similar procedure as Fig.~\ref{fig:eigenstates}(e): the quantity $\tilde{p}_{z_A}(z_B) = p^{(E)}_{z_A}(z_B)/\mathbb{E}_{E}[p^{(E)}_{z_A}(z_B)]$ (blue) obeys a real PT distribution [Eq.~\eqref{eq:real_PT}, dashed line, note cusp at $\tilde{p}=0$], for $z_A = 0100$, while the raw probabilities $D p^{(E)}_{z_A}(z_B)$ do not (red).}
    \label{fig:real_Scrooge}
\end{figure}

\end{appendix}

\end{widetext}
\end{document}